\documentclass{article}
\RequirePackage{etex} 
\usepackage{upgreek}
\usepackage[letterpaper, left=1in, right=1in, top=1in, bottom=1in]{geometry}
\usepackage{paralist}

 \usepackage{natbib}
 \bibpunct{(}{)}{;}{a}{,}{,}

\RequirePackage{fancyhdr}
\usepackage{csquotes}
\usepackage[english]{babel}   
\RequirePackage{algorithm}
\RequirePackage{algorithmic}

\newif\ificmlmfs %mathfontsize of main text
\icmlmfsfalse

\newcommand{\maybeicml}[2]{%
  \ificmlmfs          % TRUE branch
     {#1}
  \else                  % FALSE branch
     {#2}
  \fi}

\usepackage{paralist}
\usepackage{xparse}
\usepackage[dvipsnames]{xcolor}

\usepackage{amsmath,amsfonts,amssymb}
\usepackage{mathtools}
\usepackage{bbm}

\usepackage{bm}
\usepackage{microtype}
\usepackage{graphicx}

\usepackage{subfigure}
\usepackage{booktabs} % for professional tables
\usepackage{enumitem}
\setlist{leftmargin=10pt,labelindent=10pt,topsep=10pt}
\setlist[enumerate]{wide=0pt, leftmargin=10pt, labelwidth=10pt, align=left}

\usepackage{xpatch}

\usepackage[colorlinks=true, linkcolor=Orchid, citecolor=Plum,urlcolor=Fuchsia,breaklinks=true]{hyperref}
\hypersetup{
colorlinks=true,
linktocpage=true,
pdfstartview=FitH,
breaklinks=true,
pdfpagemode=UseNone,
pageanchor=true,
pdfpagemode=UseOutlines,
plainpages=false,
bookmarksnumbered,
bookmarksopen=false,
bookmarksopenlevel=1,
hypertexnames=true,
pdfhighlight=/O,
pdftitle={},
pdfauthor={},
pdfsubject={},
pdfkeywords={},
}

\usepackage{multirow,array}

\def\1{\bm{1}}

\newcommand{\calA}{\ensuremath{\mathcal{A}}}
\newcommand{\calB}{\ensuremath{\mathcal{B}}}

\newcommand{\calD}{\ensuremath{\mathcal{D}}}

\newcommand{\calG}{\ensuremath{\mathcal{G}}}

\newcommand{\calM}{\ensuremath{\mathcal{M}}}

\newcommand{\calS}{\ensuremath{\mathcal{S}}}

\newcommand{\calU}{\ensuremath{\mathcal{U}}}

\newcommand{\calX}{\ensuremath{\mathcal{X}}}
\newcommand{\calY}{\ensuremath{\mathcal{Y}}}
\newcommand{\calZ}{\ensuremath{\mathcal{Z}}}

\let\E\undefined
\let\R\undefined

\newcommand{\E}{\mathbb{E}}

\newcommand{\R}{\mathbb{R}}

\renewcommand{\bar}[1]{\overline{#1}}

\newcommand{\norm}[1]{\left\| #1 \right\|}
\DeclareMathOperator*{\argmax}{arg\,max}
\DeclareMathOperator*{\argmin}{arg\,min}

\DeclareMathOperator{\dist}{dist}

\DeclareMathOperator{\pr}{\mathbb{P}}

\newcommand{\braces}[1]{\left\{ #1 \right\}}

\newcommand{\barx}{\bar{x}}
\newcommand{\ellphi}{\ell_{\Phi}}
\newcommand{\qg}{\mathrm{QG}}

\newcommand{\at}[1]{_{#1}}
\newcommand{\etay}{\tau_y}
\newcommand{\etax}{\tau_x}
\newcommand{\mux}{\mu_x}
\newcommand{\muy}{\mu_y}

\newcommand{\inprod}[2]{\left\langle #1, #2 \right\rangle}

\newcommand{\paren}[1]{\left( #1 \right)}
\newcommand{\Bigparen}[1]{\Big( #1 \Big)}

\newcommand{\hatgradfb}{\hat{g}}
\newcommand{\gradfb}{g}
\newcommand{\hatgradfby}{\hat{g}_y}
\newcommand{\gradfby}{g_y}
\newcommand{\hatgradfbx}{\hat{g}_x}
\newcommand{\gradfbx}{g_x}

\newcommand{\Lcinv}{L_{c^{-1}}}

\newcommand{\muqg}{\mu_{\mathrm{QG}}}
\newcommand{\proj}{{\Pi}}
\newcommand{\hatg}{\hat{g}}

\newcommand{\resx}[1]{\calG_{x,\etax}^{#1}}
\newcommand{\sresx}[1]{\hat{\calG}_{x,\etax}^{#1}}
\newcommand{\resy}[1]{\calG_{y,\etay}^{#1}}
\newcommand{\sresy}[1]{\hat{\calG}_{y,\etay}^{#1}}

\newcommand{\Tout}{T_{\mathrm{out}}}
\newcommand{\Tin}{T_{\mathrm{in}}}
\newcommand{\tin}{s}
\newcommand{\tout}{t}

\newcommand{\Llambda}{L_\lambda}
\newcommand{\elllambda}{\ell_\lambda}
\newcommand{\Llambdainv}{L_{\lambda^{-1}}}
\mathchardef\mhyphen="2D

\newcommand{\nepg}{\ensuremath{\mathtt{Nest\mhyphen PG}}}

\newcommand{\apgda}{\ensuremath{\mathtt{Alt\mhyphen PGDA}}}

\newcommand{\gdmax}{\ensuremath{\mathtt{GDmax}}}

\newcommand{\altgda}{\ensuremath{\mathtt{Alt\mhyphen GDA}}}

\newcommand{\ellphimu}{\ell_{\Phi,\mu}}
\newcommand{\diam}[1]{\mathrm{D}_{#1}}

\newcommand{\hatr}{\hat{r}}
\newcommand{\minrho}{{\min_s\varrho(s)}}

\newcommand{\mupl}{\mu_{\mathrm{PL}}}

\newcommand{\muplx}{\mu_{\mathrm{pl},x}}
\newcommand{\muply}{\mu_{\mathrm{pl},y}}
\newcommand{\muqgx}{\mu_{\mathrm{qg},x}}
\newcommand{\muqgy}{\mu_{\mathrm{qg},y}}

\usepackage{amsthm}
\usepackage[capitalize,noabbrev]{cleveref}
\newtheorem*{theorem*}{Theorem}

\theoremstyle{plain}
\newtheorem{theorem}{Theorem}[section]
\newtheorem{proposition}[theorem]{Proposition}
\newtheorem{lemma}[theorem]{Lemma}
\newtheorem{claim}[theorem]{Claim}
\newtheorem{corollary}[theorem]{Corollary}
\newtheorem{fact}{Fact}[section]

\theoremstyle{definition}
\newtheorem{definition}{Definition}
\newtheorem{assumption}{Assumption}
\newtheorem*{model*}{Model}
\newtheorem*{informallemma*}{Informal Leamm}
\theoremstyle{remark}
\newtheorem{remark}{Remark}

\crefname{lemma}{Lemma}{Lemmata}
\crefname{claim}{Claim}{claims}
\crefname{proposition}{Proposition}{Propositions}
\crefname{fact}{fact}{facts}
\crefname{line}{Lines}{lines}
\crefname{fact}{facts}{facts}
\crefname{assumption}{assumption}{assumptions}

\newcommand{\poly}{\mathsf{poly}}

\usepackage{autonum}

\usepackage{multicol}
\newcommand{\fk}[1]{}
\newcommand{\gp}[1]{}
\newcommand{\mv}[1]{}
\newcommand{\ig}[1]{}
\usepackage[textsize=tiny]{todonotes}

\usepackage{wasysym}           % for \fleur (fleur‐de‐lis “flower”)
\usepackage{marvosym}          % for \Flower (five‐petal flower)
\usepackage{pifont}            % for \ding{…} dingbats

\let\tempone\itemize
\let\temptwo\enditemize
\renewenvironment{itemize}{\tempone\addtolength{\itemsep}{-0.2\baselineskip}}{\temptwo}

\usepackage[toc,page,header]{appendix}
\usepackage{minitoc}

\newcommand{\titleI}{Solving Zero-Sum Convex Markov Games}

\usepackage{authblk} % Optional, for multiple affiliations per author
\usepackage{hyperref} % For clickable email link

\author[a]{Fivos Kalogiannis\thanks{Corresponding author: \texttt{fkalogiannis@ucsd.edu}}}
\author[b,c]{Emmanouil-Vasileios Vlatakis-Gkaragkounis}
\author[d]{Ian Gemp}
\author[d]{Georgios Piliouras}

\usepackage[utf8]{inputenc}
\usepackage{amssymb}
\affil[a]{\small Department of Computer Science \& Engineering, University of California San Diego, La Jolla, CA, USA}
\affil[b]{Department of Computer Sciences, University of Wisconsin-Madison, Madison, WI, USA}
\affil[c]{Archimedes, Athena Research Center, Greece}
\affil[d]{Google DeepMind, London, UK}

\title{\titleI}

\begin{document}

\maketitle

\begin{abstract}
We contribute the first provable guarantees of global convergence to Nash equilibria (NE) in two-player zero-sum convex Markov games (cMGs) by using independent policy gradient methods. Convex Markov games, recently defined by~\citet{gemp2024convex}, extend Markov decision processes to multi-agent settings with preferences that are convex over occupancy measures, offering a broad framework for modeling generic strategic interactions. However, even the fundamental min-max case of cMGs presents significant challenges, including inherent nonconvexity, the absence of Bellman consistency, and the complexity of the infinite horizon. 

We follow a two-step approach. First, leveraging properties of hidden-convex–hidden-concave functions, we show that a simple nonconvex regularization transforms the min-max optimization problem into a nonconvex–proximal Polyak-\L ojasiewicz (NC-pPL) objective. 
Crucially, this regularization can stabilize the iterates of independent policy gradient methods and ultimately lead them to converge to equilibria. Second, building on this reduction, we address the general constrained min-max problems under NC-pPL and two-sided pPL conditions, providing the first global convergence guarantees for stochastic nested and alternating gradient descent-ascent methods, which we believe may be of independent interest.
\end{abstract}

\section{Introduction}
The field of multi-agent reinforcement learning (MARL)—often framed as Markov games (MGs) \citep{littman1994markov}---studies how multiple agents interact within a shared, dynamic environment to optimize their individual cumulative rewards \citep{silver2017mastering,lanctot2019openspiel}. However, many real-world applications require a more expressive formulation of agent preferences that do not simply decompose additively over time \citep{puterman2014markov,zahavy2021reward}. To address this limitation, the emerging framework of convex Markov games (cMGs) \citep{gemp2024convex} has been proposed as a principled yet flexible model for capturing complex multi-agent interactions in dynamic environments. 
Unlike traditional MGs, cMGs allow for convex utility functions over the state-action occupancy measure, enabling a richer set of players' preferences 
that better reflect practical applications.
In this context, the occupancy measure of each player reflects the frequency of visiting any particular state of their corresponding Markov decision processes (MDPs), over a {potentially infinite time horizon}.

In practice, cMGs are useful for modeling a variety of challenging multi-agent settings, including:
\begin{inparaenum}[(i)]
\item Fostering \emph{creativity} in machine gameplay, such as discovering novel strategies in chess \citep{zahavy2022discovering,zahavy2023diversifying,bakhtin2022mastering},
\item multi-step language model alignment \citep{wu2025multi},
\item enhancing multi-agent \emph{exploration} in robotic systems \citep{burgard2000collaborative,rogers2013coordination,tan2022deep},
\item ensuring \emph{safety} in autonomous driving \citep{shalev2016safe},
\item enabling \emph{imitation learning} from expert demonstrations, and
\item promoting \emph{robustness} and \emph{fairness}, in multi-agent decision-making \citep{hughes2018inequity}.
\end{inparaenum}
While the former (i-iii) are direct instantiations of cMGs, the latter (iv-vi) will profit when from a cMG formulation.

In general, the authoritative desirable outcome of a multi-agent scenario is some sort of game theoretic \emph{equilibrium}. Given the plethora of cMG applications, a natural question arises:~\textit{are there algorithmic solutions for provable equilibrium computation in these games already?} Surprisingly, this is not the case! Notwithstanding the model's appeal, an array of technical challenges arise that impede a straightforward algorithmic solution. Yet, empirical results suggest that variants of policy gradient methods \citep{williams1992simple} can actually lead to favorable outcomes. 
Following this thread, we consider the simplest reasonable setting, \emph{competition between two agents}, and pose our central question:
\begin{equation}
\parbox{40em}{\center\textit{\enquote{Do policy gradient methods converge in zero-sum convex Markov games?}\\~\\}}
    \tag{\ding{95}}
    \label{central-question}
\end{equation}
\label{sec:intro}
We answer this question in the affirmative. More explicitly,
we provide an algorithmic framework that is \textit{efficient}, adheres to the \textit{independent learning} desideratum~\citep{daskalakis2020independent}, and is simultaneously \textit{easy-to-implement}. 
This effectiveness  is powered by novel technical insights to address  numerous technical challenges native to cMGs.

Getting into the weeds, we enumerate the technical challenges that cMGs pose. As cMGs are a strict generalization of cMDPs, the Bellman consistency of the agents' utility functions fails to hold~\citep{zhang2020variational}---in short, we cannot define individual \textit{value} and \textit{action-value} functions. Remarkably, this is a family of long-horizon strategic interactions where agents \emph{cannot use dynamic programming} as opposed to conventional MGs. We note that, even though most policy optimization methods rely on a gradient-based approach, their majority implicitly performs an approximate dynamic programming subroutine such as value-iteration \citep{jin2021v,wei2021last,ding2022independent,erez2023regret}.

Consequently, the absence of state-wise value-functions translates into a failure of the elementary arguments that were used to prove the existence of Nash equilibria \citep{fink1964equilibrium}, even in two-player zero-sum Markov games \citep{shapley1953stochastic}. In light of the latter, \citet{gemp2024convex} prove the existence of Nash equilibria in cMGs by going beyond the fundamental toolbox of the Brouwer and the Kakutani fixed point theorems. What renders these conventional theorems obsolete is the inherent \emph{nonconvexity of the individual best-response mappings}
in terms of players' policies---where, each best-response mapping yields the set of the agent's set of utility maximizing strategy deviations.

The careful reader might already grasp that the failure of Bellman consistency rules out the seamless application of the majority of MARL algorithms; most of them rely on computing \textit{value} and \textit{action-value} functions. We direct our hopes to policy gradient methods which theory of cMDPs and practice of MARL suggest can work. Directly optimizing a policy means facing an optimization landscape that is nonconvex. However, since \emph{all local optima are also global in cMGs} should bring some hope. Nonetheless, why should policy gradient methods work in cMGs when vanilla gradient following methods cycle~\cite{bailey2018multiplicative} and exhibit non-convergent chaotic trajectories even in zero-sum two-player normal-form games? Observe that the latter are nothing but cMGs on a single state with individually linear utilities. 
Further, even if we stabilize gradient dynamics towards Nash equilibria, attaining strong convergence guarantees of an algorithm %our algorithms 
requires the deepening of our understanding of the optimization landscape of cMDPs and cMGs. For instance, in the min-max case, the computational complexity of the general problem of computing a saddle-point of a nonconvex-nonconcave smooth function remains far from being settled~\citep{daskalakis2021complexity}, while no algorithm is yet guaranteed to efficiently compute them without additional structural assumptions.  

Searching for structure, we observe that the utilities in cMGs exhibit a property known  as \textit{hidden convexity}~\citep{ben1996hidden,li2005hidden,wu2007peeling,xia2020survey,vlatakis2021solving,fatkhullin2023stochastic}. However, the existing results on hidden convexity either focus on single-objective minimization or, in the context of continuous games, make assumptions that do not hold in cMGs (separability of the hidden mapping and unconstrained variables). Having identified the key  difficulties of zero-sum cMGs, in the following section we offer a detailed account of how we tackle them. Briefly, we manage to overcome the aforementioned challenges using some distinct combination of the following algorithmic techniques:
\begin{inparaenum}[(i)]
    \item a conceptually simple but \emph{nonconvex regularization} of the utility function;
    \item \emph{alternating gradient iterations}~\citep{tammelin15solving,chavdarova2021tamingganslookaheadminmax,chavdarova19reducing,Zhang2019ConvergenceOG,Chambolle2011AFP};
    \item a careful \emph{timescale separation} of the individual step-size, and lastly \item \emph{nested-gradient iterations}.
\end{inparaenum}
We additionally note that regularization of the objective using a function that is strongly convex w.r.t. the occupancy measure has been explored broadly in the literature of RL \citep{ratliff2006maximum,syed2007game,ziebart2008maximum,peters2010relative,zimin2013online,neu2017unified,nachum2019algaedice,kalogiannis2024learning}.

\subsection{Technical Overview}
The proof of the proposed theoretical guarantees comes after a combination of several key observations---old and new. We believe that exposing them in the following itemized manner will serve in conveying the technical elements of our work and an intuitive overview of why our approach works. 
\begin{itemize}
    \item Hidden convex functions are nonconvex functions that can be expressed as a convex function of a reparametrization of their original arguments, termed the ``hidden (latent) mapping,'' while the original arguments are called ``control variables''. Yet, in some settings, like reinforcement learning (RL), this mapping is not at all hidden but directly computable or observable.
    \item More specifically, in cMGs, players' utilities are hiddenly concave through the state-action occupancy measure, which is similarly computable (in planning) or observable (in RL). 
    \item 
    Although analyzing these games via occupancy measures yields a quasi-variational inequality problem~\citep{kinderlehrer2000introduction}--unnecessarily increasing the technical challenges \citep{bernasconi2024role}---yet, two crucial facts remain: \begin{inparaenum}[(i)]
    \item the utility is concave with respect to the hidden mapping, and \item the hidden mapping is accessible. These properties ensure that after regularization, the hidden function becomes strongly concave. In turn, the resulting function satisfies the proximal Polyak-\L ojasiewicz (pPL) condition, with respect to the control variables (\textit{i.e.}, the policies).
    \end{inparaenum} 
    \item  Consequently, each player’s utility remains pPL while the feasibility sets for their control variables remain independent. This observation
    wraps up the discussion on
    % concludes our interest with 
    hidden concavity and naturally shifts our focus on the constrained optimization of min-max nonconvex-pPL objectives. {In other words, we reduce solving two-player zero-sum convex Markov games to the problem of computing saddle-points of functions that are nonconvex-pPL over constrained domains.}
\item 
Leveraging this reduction, our final steps center on developing policy gradient methods that can provably converge to saddle-points of nonconvex-pPL functions over constrained domains, a technical development in optimization theory of independent interest as well.
\end{itemize}
Echoing \citep{nesterov2005smooth}, an alternative way to view our approach is as \textit{smooth minimization} of \textit{nonconvex nonsmooth functions} in the presence of \textit{hidden convexity}.  In short,~\citet{nesterov2005smooth} minimizes a convex nonsmooth objective, $\Phi(x) = \max_{y} f(x,y)$, where $f(x,\cdot)$ is concave for any $x$ by subtracting a strongly convex regularization term $\mu d(y)$. In turn, we are faced with the case of $f(x,\cdot)$ being hidden concave and we in turn add a term which is hidden strongly convex.

\subsection{Our Contributions} \
Before proceeding to enumerate all parts of our contributions, we remark that we deliver a definitive answer to question \eqref{central-question}:
\begin{theorem*} There exist decentralized policy gradient methods (\cref{alg:nepg,alg:apgda}) that compute an $\epsilon$-approximate Nash equilibrium for any $\epsilon>0$ in any two-player zero-sum convex Markov game using iterations and samples that are:
$$\poly \paren{\frac{1}{\epsilon}, |\calS|, |\calA|+ |\calB|, \frac{1}{1-\gamma}},$$
with $\calS$ denoting the cMG's state-space, $\calA, \calB$ the two players' action-spaces, and $\gamma>0$ the discount factor.
\end{theorem*}

Our contributions span two key areas: constrained nonconvex min-max optimization and equilibrium computation in convex Markov games.
\begin{itemize}
    \item In \cref{theorem:contiuity-of-maximizers-cmgs} we demonstrate that  the best-response mapping for an objective $f(x,y)$,
    %the mapping 
    $x\mapsto y^\star(x):= \argmax_{y\in\calY}f(x,y)$,
    %$y^\star(\cdot):= \argmax_{y\in\calY}f(\cdot,y)$, 
    is \textit{Lipschitz continuous} in $x$ for regularized hidden-convex---and more generally, NC-pPL games. The significance of this result lies in guaranteeing the stability of the iterates of policy gradient methods and serves as a suggestion to practitioners looking for a
     regularization technique that is intuitive, simple, and easily implementable.
    \item In \cref{sec:cmgs-main-results}, we provide the first provable guarantees of convergence to a Nash equilibirum for two policy gradient algorithms \nepg\ %, \ttpgda,
    and \apgda~(\cref{theorem:informal-cmg-nepg-hc}). 
\end{itemize}
\paragraph{Key ingredients.}  
To establish our results, we leverage a set of incorporating the distinct combination of the following non-trivial components:
\begin{inparaenum}[(i)]
    \item an intuitive stability-inducing specialized \emph{regularization} of the utility function and 
    \item \emph{alternating} or
    \item \emph{nested-gradient iterations}.
\end{inparaenum}
A notable element of our approach is that it ensures convergence even with \emph{inexact} gradients on top of being \emph{stochastic}.
The algorithm's robustness to gradient inexactness preserves each player’s autonomy, allowing them to optimize independently without exchanging private policy information. \textit{I.e.}, unlike the single-agent setting, where exact gradients can be stochastically estimated, our framework’s regularizer depends on \textit{both} players’ policies, making exact estimation infeasible without policy sharing. Consequently, each player must rely solely on an inexact estimation of their own gradient.

Finally, a noteworthy product of our approach is the Lipschitz continuity of the best-response mapping (which returns the optimal strategy given the opponent’s choices) in the hidden-convex case--and more generally nonconvex-pPL---despite the opponent's utility being nonconvex—contrasting with the typical $\tfrac{1}{2}$-H\"older continuity of maximizers in constrained optimization~\cite{li2014holder}. 
Indeed, while \citep{kalogiannis2024learning} employ a similar technique of regularizing the value function through the occupancy measure perspective, their result only establishes a weaker notion of continuity due to the \textit{coupledness} of individual players' state-action occupancy measure. This claim has been independently supported by \citep[Robust Berge Theorem]{papadimitriou2023computational}, which considers general nonconvex--strongly-concave functions where the feasibility sets of two different strategies depend on each other in a Hausdorff continuous manner. 
To significantly strengthen the continuity of the best-response mapping, we leverage the pPL condition in the individual policy spaces of agents that remain \emph{uncoupled}. This result was only known for the significantly simpler case of unconstrained two-sided PL functions \citep{nouiehed2019solving}.
% \newpage
% \paragraph{Paper Roadmap.} In \cref{sec:prelims}, we introduce the key concepts that form the foundation of our approach. \cref{sec:min-max-opt,sec:cmgs-main-results} present our novel results on structured nonconvex min-max constrained optimization, along with algorithms for computing Nash equilibria in zero-sum convex Markov games. Finally, \cref{sec:experiments} provides numerical experiments, while \cref{sec:future} concludes our work.
\section{Preliminaries}
\paragraph{Notation.} In general, $x,y,z,u,v,w$ and $\lambda,\lambda_1,\lambda_2$ will denote vectors. Scalars will be denoted using $\alpha, \beta, \gamma, \delta, \epsilon, \zeta, \kappa, \mu, \nu$ and $a,b,c,d$. Matrices will be denoted with bold uppercase letters. The probability simplex supported on a finite set $\calM$ will be denote with $\Delta(\calM)$. For a compact convex set $\calZ\subseteq \R^d$, we will denote its Euclidean diameter as $\diam{\calZ}$, \textit{i.e.}, $\diam{\calZ}:= \max_{x,y\in\calZ}\norm{x-y}_2$. Lastly, the global optimum of a  function $f$ will be denoted as $f^\star$.

\label{sec:prelims}
% We dedicate this section to introducing the key concepts that we build on and leverage.
\subsection{Convex Markov Games}
In this subsection, we define two-player zero-sum convex Markov games~\citep{gemp2024convex} and introduce necessary notation. We then present the \textit{occupancy measure} and remark its continuity properties. Subsequently, we review the \textit{policy gradient theorem} for convex MDPs and describe a stochastic policy gradient estimator. Finally, we define the perturbed utility function $U^\mu$, obtained by adding a regularization term to the original utility function $U$.

\begin{definition}[Two-player zero-sum cMG]
An infinite-horizon zero-sum convex Markov game is a tuple 
$
\Gamma =\paren{\calS, \calA, \calB, \pr, F, \gamma, \varrho }$:
\setlist{itemsep=2.5pt,parsep=0pt,topsep=5pt}
\begin{itemize}
    \item a maximizing and a minimizing player,
    \item a finite state space $\calS$, and an initial state distribution $\varrho \in \Delta(\calS)$,
    \item finite action spaces $\calA, \calB$,
    \item a state transition function $\pr: \calS \times \calA \times \calB \to \Delta(\calS)$,
    \item a discount factor $\gamma \in [0, 1)$, and
    \item two continuous utilities $F_1, F_2$ functions corresponding to each player's occupancy measure, \textit{i.e.}, there exist concave $F_1, F_2$
    \maybeicml{
                \begin{align}
                    F_1: \Delta(\calS \times \calA) \times \Delta(\calS \times \calB)\to \mathbb{R};\\
                    F_2: \Delta(\calS \times \calA) \times \Delta(\calS \times \calB)\to \mathbb{R},
                \end{align}
                }
                {
                \begin{align}
                    F_1: \Delta(\calS \times \calA) \times \Delta(\calS \times \calB)\to \mathbb{R};\quad \text{and} \quad
                    F_2: \Delta(\calS \times \calA) \times \Delta(\calS \times \calB)\to \mathbb{R},
                \end{align}
                }
    with $-F_1 = F_2 =: F$.    \label{def:cmg}
\end{itemize}
\end{definition}

With all this in hand, 
 we adopt the following standard assumptions, which are widely used in prior work \cite{zhang2020variational}:
\begin{assumption}
    For the initial state distribution, it holds that $\varrho(s)>0, \forall s$.
    \label{assum:fullsupport}
\end{assumption} 

Additionally, we assume \emph{direct policy parametrization} and define the Markovian and stationary policy of the minimizing and the maximizing player to be $x \in \Delta(\calA)^{|\calS|}=: \calX$ and $y \in \Delta(\calB)^{|\calS|}=: \calY$ respectively. Throughout, we only consider Markovian stationary policies. After the agents fix their policies, $x,y$, the transition matrix $\pr(x,y)\in \R^{|\calS|\times|\calS|}$ dictates how they traverse the state space. The occupancy measures are defined as:
\begin{align}
 \textstyle
    \lambda_{1}^{s,a} &:= (1-\gamma) \E_{x,y}\left[    \textstyle
 \sum_{h}^\infty \gamma^{h}\mathbbm{1}\{s^{(h)}=s,a^{(h)}=a\}|\varrho \right];\\
     \lambda_{2}^{s,b} &:=(1-\gamma) \E_{x,y}\left[    \textstyle
 \sum_{h}^\infty \gamma^{h}\mathbbm{1}\{s^{(h)}=s,b^{(h)}=b\}|\varrho \right].
\end{align}

\paragraph{Further notation.}
Further, we denote $\lambda$ to be the state-joint-action occupancy measure, $\lambda\in \Delta(\calS\times\calA\times\calB)$. Overloading notation, $\lambda(x, y)$ stands for the unique occupancy measure that corresponds to the policy pair $x,y$. Additionally, $\lambda_1 \in \Delta(\calS \times \calA), \lambda_2 \in \Delta (\calS \times \calB)$ will signify the marginal occupancy measures with respect to the minimizing and the maximizing player respectively. Again, overloading notation, $\lambda_1(x,y)$ and $\lambda_2(x,y)$ are the unique occupancy measures for a policy pair $x,y$. Finally, we will at times suppress the notation $F_1(\lambda_1(x,y);y)$ in place of $F_1(\lambda_1(x,y),\lambda_2(x,y))$ and similarly for $F_2$. 

% $\lambda,\varrho, \upnu,\upmu$
% Given \cref{assum:fullsupport}, there exist two ``$1$--$1$'' mapping, between each player's state-action visitation measures and policies.  
%Yet, the occupancy measure and its inverse operators satisfy the following continuity properties for each player:
%, the following continuity properties hold for the occupancy measure and its inverse occupancy operators for each player:
 Crucially, both the occupancy measure and its inverse operators satisfy the following continuity properties:
\begin{lemma}[Continuity of the occupancy measure]\label{lemma:continuity}
Let $\lambda \in \Delta(\calS\times\calA\times\calB)$ be the occupancy measure in a (convex) Markov game and let  $\lambda_1^{-1}:\Delta(\calS\times\calA)\to\calX$ and $\lambda_2^{-1}:\Delta(\calS\times\calB)\to\calY$ be the occupancy-to-policy mapping such that:
$ \lambda_{1}^{-1}\paren{\lambda_1(x,y)}= x; 
  \lambda_{2}^{-1}\paren{\lambda_2(x,y)} = y
$. 
Then,
\begin{itemize}
    \item 
$\lambda$ is $\Llambda$-Lipschitz continuous and has an $\elllambda$-Lipschitz continuous gradient with respect to the policy pair $(x,y)$ in $\calX\times\calY$. Specifically, for all $(x,y)$ and $(x',y')$,
\begin{align}
\norm{\lambda(x,y) - \lambda(x',y')}
&\leq\Llambda\norm{(x,y) - (x',y')};\\
\norm{\nabla \lambda(x,y) - \nabla \lambda(x',y')}
&\leq \elllambda\norm{(x,y) - (x',y')},
\end{align}
where $\Llambda:=\frac{{|\calS|^{\frac{1}{2}}}\paren{|\calA|+|\calB|}}{(1-\gamma)^2}$, and $\elllambda:=  \frac{2\gamma{|\calS|^{\frac{1}{2}}}\paren{|\calA|+|\calB|}^{\frac{3}{2}}}{(1-\gamma)^3}$.
\item For any fixed $y$ (respectively, $x$), $\lambda^{-1}$ is $\Llambdainv$-Lipschitz continuous with respect to $\lambda_1$ (respectively, $\lambda_2$), \textit{i.e.}, for all $\lambda_1(x,y), \lambda_1(x',y)$---respectively, $\lambda_2(x,y), \lambda_2(x,y')$,
\begin{align}
\norm{x-x'}&\leq\Llambdainv \norm{\lambda_1^{-1}\paren{\lambda_1(x,y) - \lambda_1(x',y)}};\\
\norm{y-y'}&\leq\Llambdainv \norm{\lambda_2^{-1}\paren{\lambda_2(x,y) - \lambda_2(x,y')}},
\end{align}
with $\Llambdainv:=  \frac{2}{\min_s \varrho(s)(1-\gamma)}$.
\label{continuity}
\end{itemize}
\end{lemma}
Next, our solution concept corresponds to the following min-max Nash equilibrium of the following function $U:\Delta(\calA)^{|\calS|}\times \Delta(\calB)^{|\calS|}\to \R $ be such that:
\begin{align}
    U(x,y) &:= F\paren{\lambda_1(x,y),\lambda_2(x,y)}.
\end{align}

\begin{definition}[$\epsilon$-NE]
    A policy profile $(x^\star,y^\star)\in \calX\times \calY$ is said to be an $\epsilon$-approximate Nash equilibrium ($\epsilon$-NE), if for any $x\in\calX$ and any $y\in\calY$, it holds that:
    \begin{align}
         U(x^\star,y) - \epsilon \leq U(x^\star,y^\star) \leq U(x,y^\star)  + \epsilon.\tag{$\epsilon$-NE}\label{eq:NE}
        % U(x^\star,y^\star) \geq U(x,y^\star)  -\epsilon;
    \end{align}
\end{definition}
Finally, we will denote $U^\mu$ to be:
    $U^\mu(x,y):= U\paren{\lambda(x,y)} - \frac{\mu}{2}\| \lambda_2(x,y) \|^2.$ The following Lemma is a direct consequence of hidden-strong-convexity.

\begin{lemma}
    When $\mu>0$, $U^\mu(x,\cdot)$ has a unique maximizer $y^\star(x)$, for all $x$.
\end{lemma}

\subsection{Policy Gradient Estimators}
As discussed, the inexistence of value or action-value functions for general utility MDPs, leads us to focus on policy gradient methods; the direct application of vanilla Q-learning methods is out of the question. To compute the policy gradient of a utility that is nonlinear in $\lambda_1$, we make use of the \citep{williams1992simple} along the chain rule of differentiation to write:
\begin{align}
    \nabla_{x} U(x,y) = \sum_{s\in\calS}\sum_{a\in\calA} \frac{\partial F_1 }{\partial \lambda_1^{s,a}(x,y)} \nabla_{x}\lambda_1^{s,a}(x,y).
\end{align}
By sampling trajectories, each agent can stochastically estimate the policy gradient using the following estimator.
\begin{definition}[Gradient Estimator] \label{def:state_action_grad_estimator}
    Given a trajectory $\xi=\paren{s^{(0)}, a^{(0)}}, s^{(1)}, \dots, s^{(H-1)},a^{(H-1)})$ of length $H$ sampled under a policy profile $x,y$ and initial distribution $\varrho$, the gradient estimator, $\hatgradfbx(\xi|x,z)$, is defined as:
    % For a trajectory $\xi = ( s_0, a_0, s_1, a_1, \cdots, s_{H-1}, a_{H-1}) $ of length $H$ sampled under initial distribution $\varrho$ and policies $x,y$ and a the pseudo-reward vector $z := \nabla_{\lambda} F( \lambda (y) )$. The estimator for gradient $\nabla_{x} F(\lambda(x,y)) $ using the sampled trajectory $\xi$ is defined as 
    % $$\hatgradfbx(\xi|x; z) := \sum_{h = 0}^{H - 1} \gamma^h  z\paren{s^{(h)}, a^{(h)}} \left(\sum_{h' = 0}^{h} \nabla_{x} \log x \paren{a^{(h')} | s^{(h')}}\right),$$
    \maybeicml{
        \begin{align} \small
            &\hatgradfbx(\xi|x, z) := \\& \quad \small\sum_{h = 0}^{H - 1} \gamma^h  z\paren{s^{(h)}, a^{(h)}} \left(\sum_{h' = 0}^{h} \nabla_x\log x\paren{a^{(h')} | s^{(h')} } \small\right),
        \end{align}
        }
        {
         \begin{align} 
            \hatgradfbx(\xi|x, z) := \sum_{h = 0}^{H - 1} \gamma^h  z\paren{s^{(h)}, a^{(h)}} \left(\sum_{h' = 0}^{h} \nabla_x\log x\paren{a^{(h')} | s^{(h')} } \right),
    \end{align}
        }
    with $z = \nabla_{\lambda_1} F_1(\lambda_1(x,y);y)$.
\end{definition}
\paragraph{Sufficient exploration} In order to ensure that the agents sufficiently explore the environment and ultimately control the variance of the estimators, we assume that both agents are following $\varepsilon$-greedy direct policy parametrization, \textit{i.e.,} for a given parameter value $x\in\Delta^{|\calS|}(\calA)$, the agent plays according to $\pi_x = (1-\varepsilon)x + \frac{\varepsilon \bm{1}}{|\calA|}$, where $\bm{1}$ is an all-ones vector of appropriate dimension.

\subsection{Optimization Theory}
Next, we introduce several key concepts from nonconvex and min–max optimization, focusing on hidden convexity and gradient domination conditions. We demonstrate how strong hidden convexity implies the proximal \emph{Polyak-\L ojasiewicz condition} (pPL) and the \emph{quadratic growth condition} (QG). We then 
%define saddle-points for min–max problems and
show that when a min–max objective satisfies a two-sided gradient domination condition, it enjoys a zero duality gap.

\begin{definition}[Hidden Convex Function]
    Consider a function $f:\calX\to \R$ where $\calX$ is a compact convex set such that $f(x):= H\paren{c(x)},~\forall x \in \calX$ for some mapping $c$ and function $H$. If the following conditions are satisfied:\label{def:hsc}
        \begin{itemize}
\setlist{itemsep=0.5pt,parsep=0pt,topsep=0pt}
        \item the mapping $c$ is invertible and its inverse $c^{-1}$ is $1/\mu_c$ Lipschitz continuous.
        \item the set  $\calU:= c(\calX)$ is convex and the function $H:\calU\to\R$ is $\mu_H$-strongly convex.
    \end{itemize}
    The function is said to be $(\mu_c, \mu_H)$-\emph{hidden strongly convex} (HSC), while for $\mu_H = 0$, it is referred merely  as \emph{hidden convex} (HC).
    % \begin{enumerate}[label=(HC.\arabic*)]
    %     \item {Convexity of $H$ and Domain:}
    %     \begin{itemize}
    %         \item The domain $\calU = c(\mathcal{X})$ is convex.
    %         \item The function $H : \calU \to \R$ satisfies for all $u, v \in 
    %         \calU$ and $0\leq\alpha\leq 1$,
    %         \begin{align}
    %             &H((1-\alpha)u + \alpha v) \\
    %             &\leq (1-\alpha)H(u) + \alpha H(v) - \frac{(1-\alpha)\alpha \mu_H}{2} \|u - v\|^2.
    %         \end{align}
    %         \item The problem admits a solution $u^\star \in \calU$.
    %     \end{itemize}
        
    %     \item {Invertibility and Lipschitz Continuity of $c^{-1}$:}
    %     \begin{itemize}
    %         \item The mapping $c : \calX \to \calU$ is invertible.
    %         \item There exists $\mu_c > 0$ such that for all $x, y \in \mathcal{X}$,
    %         $$
    %             \norm{x - y} \leq \frac{1}{\mu_c}\norm{c(x) - c(y)}.
    %         $$
    %     \end{itemize}
    % \end{enumerate}
    % Furthermore, if $\mu_H > 0$, the problem is referred to as \emph{$(\mu_c, \mu_H)$-hidden strongly convex}.
\end{definition}
Notably, the convergence analysis of our proposed methods begins with the following claim, which connects cMG utilities to hidden convexity.
\begin{claim}\label{claim:reduction}
The function $U$ is hidden convex (resp., hidden concave) for the min. player (resp., max. player), given a fixed action of the opponent. Similarly, the perturbed utility function \( U^\mu \) is \( (\mu,\Llambdainv) \)-hidden strongly concave for the max player due to the structure of the regularizer.
\end{claim}\vspace{-1em}
\begin{proof}
Follows  from \cref{def:cmg,def:hsc} \& \Cref{lemma:continuity}.
\end{proof}

% \begin{remark}
% Having established the necessary background, it is now easy to observe that, by combining \cref{def:cmg} with the continuity of occupancy measures in \cref{prop:hc_to_graddom}, the function $U$ corresponds to a hidden-(convex/concave) function for the (min/max) player, given a fixed action of the opponent. Similarly, the perturbed utility function thanks to the regularizer form corresponds to $(\mu_{reg},\Llambdainv)$-hidden strongly concave.
% \end{remark}

% \begin{definition}[Hidden concavity]
%     Let the optimization problem: $$
%     \max_{x\in\calX} f(x)$$
%     where $\calX$ is a compact convex set. The problem is called hidden $\rho_H$-strongly concave, if there exists a space $\calU$ and an invertible $\rho_c$-Lipschitz continuous mapping $c:\calX \to \calU$ along with a $\rho_H$-strongly concave function $H:\calU\to \R$ such that:
%     $$ f(x) = H\big(c(x)\big),~\forall x \in \calX.$$
%     Where, it holds $\rho_H\geq 0$ and $\rho_c> 0$.
% \end{definition}

\begin{proposition}[HC implies gradient domination; \citep{fatkhullin2023stochastic}] Let $f:\calX\to \R$ be an $\ell$-smooth and $(\mu_c,\mu_H)$-hidden convex function and $I_\calX$ be the indicator function of the set $\calX$. Further, let $F(\cdot):= f(\cdot) + I_{\calX}(\cdot)$. Also, assume that the map $c(\cdot)$ is continuously differentiable on $\calX$.

\begin{enumerate}[label=(\roman*)]
    \item  If the set $\calX = c(\mathcal{\calX})$ is bounded with diameter $\diam{\calU}$, then
    \begin{equation}
    \inf_{s_x \in \partial F(x)} \|s_x\| \geq     \frac{\mu_c }{\diam{\calU}} \label{eq:grad-domI}
 \left(F(x) - F^*\right), \quad \forall x \in \calX. 
    \end{equation}

    \item If $f(\cdot)$ is $(\mu_c, \mu_H)$-hidden strongly convex, then
    \begin{equation}
        \inf_{s_x \in \partial F(x)} \|s_x\|^2 \geq 2\mu_c^2\mu_H \left(F(x) - F^*\right), \quad \forall x \in \calX. \label{eq:grad-domII}
    % \frac{2 \mu_H }{\Lcinv^2} \left(F(x) - F^*\right), \quad \forall x \in \calX.
        \end{equation}
\end{enumerate}
\label{prop:hc_to_graddom}
\end{proposition}
\begin{definition}[pPL] Let an $\ell$-smooth function $f:\calX\to\R$ defined over the convex and compact set $\calX\subseteq \R^d$. Let $\calD_{\calX}(\cdot, \ell)$ be defined as:
$$\calD_{\calX}(x, \ell):= -2\ell \min_{y\in \calX}\braces{\inprod{\nabla f(x)}{y-x}  + \frac{\ell}{2}\norm{x - y}^2 }. $$
Then, $f$ is said to satisfy the proximal Polyak-\L ojasiewicz condtion with modulus $\mu$ if for all $x\in\calX$, it holds true that:
\begin{equation}\frac{1}{2}\calD_{\calX}(x,\ell) \geq \mu \paren{f(x) - f^*}. \label{eq:grad-domIII}\end{equation}
\end{definition}
Our following lemma establishes a variant of the ``gradient dominance'' for the case of convex Markov games.
Namely we show that an approximate constrained first-order stationary point ensures approximately optimal policies in our game. 

% (a sufficiently small gradient magnitude in the feasible directions) of the value function necessarily corresponds to a Nash equilibrium. This property, a direct consequence of the so-called ‘gradient dominance’ in conventional stochastic games,
\begin{lemma}[Gradient Dominance]
    For a zero-sum convex Markov game, it holds that:
    \begin{align}
         \max_{x'\in \calX }\inprod{\nabla_x U(x,y)}{x- x' } &\geq \mu_x\paren{U(x,y) - U(x^\star,y) } ;\\
         \max_{y'\in \calY }\inprod{\nabla_y U(x,y)}{y' -y} &\geq \mu_y\paren{U(x,y^\star) - U(x,y) },
    \end{align}
    for $\mu_x, \mu_y = \frac{(1-\gamma )\min_s \rho(s) }{2\sqrt{2}}$.
\end{lemma}
\begin{proof}
Fix an arbitrary $y\in\mathcal{Y}$. By the hidden convexity of $U(\cdot,y)$ and \cref{prop:hc_to_graddom}, we have:
\maybeicml{
            \begin{align}
                &\frac{(1-\gamma)\min_{s}\varrho(s) }{2\sqrt{2}}
                \paren{U(x,y) -\min_{x^\star\in\calX } U(x^\star,y)} \\
                &\quad\leq \min_{s_x \in \partial_x U(x,y) + \partial_x I_{\calX}(x) }\norm{s_x}.
            \end{align}}
        {
            \begin{align}
                \frac{(1-\gamma)\min_{s}\varrho(s) }{2\sqrt{2}}
                \paren{U(x,y) -\min_{x^\star\in\calX } U(x^\star,y)} 
                \leq \min_{s_x \in \partial_x U(x,y) + \partial_x I_{\calX}(x) }\norm{s_x}.
            \end{align}
    }
Where, ${\sqrt{2}}$ is the diameter of the state-action occupancy measure. 
Applying \citep[Proposition 8.32]{rockafellar2009variational}, we obtain:
\begin{align}
    %\min_{\substack{s_x \in \\\partial_x U(x,y) + \partial_x I_{\calX}(x) }}\norm{s_x} 
    \min_{s_x \in \partial_x U(x,y) + \partial_x I_{\calX}(x) }\norm{s_x}
    &= \max_{\substack{x' \in \calX, \\\norm{x - x'}\leq 1}}\inprod{\nabla U(x,y)}{x-x'} \\
    &\leq \max_{x' \in \calX}\inprod{\nabla U(x,y)}{x-x'}.
\end{align}
Thus, we set $\mu_x :=\frac{(1-\gamma)\min_{s}\varrho(s) }{2\sqrt{2}}$, and by symmetry, the same holds for $\mu_y$.
\end{proof}
Finally, in the Appendix we show that both HSC and pPL imply QG, which states that the optimality gap at any point \( x \) is lower-bounded by a quadratic term in its distance from the minimizer. This ensures that progress toward optimality can bound the proximity to the solution.
\begin{proposition}[QG from HSC and pPL]
\label{prop:QG_unified}
Let \( f:\mathcal{X} \to \mathbb{R} \) be \( \ell \)-smooth and satisfy either:  
(i) $(\mu_c, \mu_H)$-hidden strong convexity (HSC), or  
(ii) proximal Polyak-\L ojasiewicz (pPL) with modulus $\mu$.  
Then, \( f \) satisfies the \emph{quadratic growth (QG) condition}:
\begin{equation}
    f(x) - f(x^\star) \geq \frac{\mu_{\qg}}{4} \|x - x^\star_p\|^2, \quad \forall x \in \mathcal{X},
\end{equation}
where \( x^\star_p \) is the closest minimizer   in \( \mathcal{X}^\star = \argmin_{x \in \mathcal{X}} f(x) \), and  
\begin{equation}
    \mu_{\qg} = \mu_c^2\mu_H \quad \text{(HSC)}, \quad \mu_{\qg} = \mu \quad \text{(pPL)}.
\end{equation}
\end{proposition}

\paragraph{Piecing the Framework Together.}
At a high level, the standard tool in proving convergence to Nash equilibria for gradient-based methods is the construction of a potential---or \emph{Lyapunov}---function (See \citep{bof2018lyapunov}). A Lyapunov function needs to be lower-bounded and decreasing for consecutive iterates of the algorithm. However, proving the latter beyond the convex-concave setting, requires a careful examination of HSC and pPL. As such, in order to shed light on the \textit{structured} nonconvex landscape, we note: 
\begin{itemize}
\item Starting from the taxonomy of structured nonconvex functions, we show in the appendix that HSC and pPL are equivalent given a proper transformation of their moduli. 
\item Hence, the stationarity surrogate
$ \calD_{\calX}(x, \ell)$
also serves as a surrogate for optimality \( f(x) - \min_{x'\in\calX}f(x') \). More importantly, by \cref{claim:reduction}, in the min-max setting, for a fixed opponent strategy, a sufficiently small $ \calD_{\calX}(x, \ell)$ indicates that the player has found an approximate best response.
\end{itemize}
Furthermore, we highlight an additional crucial detail. Since, the pPL condition already guarantees quadratic growth (with an equal modulus), the contribution of HSC in this regard may appear redundant. Nonetheless, HSC is translated to the pPL condition going through an argument which degrades the modulus of HSC $\mu$ to a pPL modulus of $\mupl = O(\mu^2)$. Hence, directly guaranteeing QG from HSC improves convergence rates of the next sections by at least an order of ${O}(\epsilon^{3/2})$.

\section{New Insights in Structured Nonconvex  Min-Max Optimization}
As a preliminary step of independent interest, we present our results on optimization schemes for constrained min-max problems under nonconvex-pPL and pPL-pPL conditions. To maintain completeness, we restate our solution concept under constrained min-max optimization perspective:% as finding an $\epsilon$-saddle point 
\begin{definition}[$\epsilon$-SP]
    Assume a smooth function $f:\calX\times \calY\to\R$ where $\calX,\calY$ are two compact and convex sets. Then, $(x^\star,y^\star) \in \calX\times \calY $ is an \emph{$\epsilon$-saddle-point ($\epsilon$-SP)} of $f$ if,
    \begin{align}\left\{
    \begin{array}{ll}
        \max_{x\in\calX}\inprod{\nabla_x f(x^\star,y^\star)}{x^\star-x} &\leq \epsilon;\\
        \max_{y\in\calY}\inprod{\nabla_y f(x^\star,y^\star)}{y - y^\star} &\leq \epsilon
        \end{array}\right.
    \label{eq:SP}\tag{$\epsilon$-SP}
    \end{align}
\end{definition}

%\section{Main Results -- Min-Max Optimization}
We adopt a context-agnostic approach, assuming each player receives a ``black-box'' representation of their gradient vector. These hypotheses remain intentionally abstract, as we impose no specific modeling assumptions on how players' payoff signals are generated. In this sense, they serve as an ``inexact gradient oracle''~\citep{devolder2014first} that captures a wide range of settings.  
Reflecting on our MARL scenario, this framework allows each player to independently select and apply their preferred gradient estimator, as in \cite{zhang2020variational,zhang2021convergence,barakat2023reinforcement}.
\label{sec:min-max-opt}
\begin{model*}[Stochastic Inexact First-Order Oracle]
For any \( t \), the gradient estimators \( \hat{g}_x(x\at{t}, y\at{t}) \) and \( \hat{g}_y(x\at{t}, y\at{t}) \) satisfy:  
\begin{align}
\maybeicml
    {\begin{array}{cc}
        \mathbb{E}[\hat{g}_x(x\at{t}, y\at{t})] = g_x(x\at{t}, y\at{t}), & 
        \mathbb{E}[\hat{g}_y(x\at{t}, y\at{t})]  =  g_y(x\at{t}, y\at{t}); \\
        \mathbb{E} \|\hat{g}_x(x\at{t}, y\at{t})\|^2  \leq \sigma_x^2, &  
        \mathbb{E} \|\hat{g}_y(x\at{t}, y\at{t})\|^2 \leq \sigma_y^2.
    \end{array}}
    {
    \left\{\begin{array}{l}
        \mathbb{E}[\hat{g}_x(x\at{t}, y\at{t})] = g_x(x\at{t}, y\at{t}),\\
        \mathbb{E} \|\hat{g}_x(x\at{t}, y\at{t})\|^2  \leq \sigma_x^2,
    \end{array}
    \right. 
    \quad\text{and}\quad
    \left\{
    \begin{array}{l}
        \mathbb{E}[\hat{g}_y(x\at{t}, y\at{t})]  =  g_y(x\at{t}, y\at{t}), \\  
        \mathbb{E} \|\hat{g}_y(x\at{t}, y\at{t})\|^2 \leq \sigma_y^2.
    \end{array}\right.
    }
\end{align}
Additionally, scalars $\delta_x, \delta_y > 0$ bound the inexactness error:
\begin{align}
    \| g_x(x\at{t}, y\at{t}) - \nabla_x f(x\at{t}, y\at{t}) \| &\leq \delta_x, \\
    \| g_y(x\at{t}, y\at{t}) - \nabla_y f(x\at{t}, y\at{t}) \| &\leq \delta_y.
\end{align}
\label{assumption:main-text-stoch-grad-oracle}
Hence, $\hatgradfbx$ and $\hatgradfby$ provide unbiased estimates of $\gradfbx$ and $\gradfby$, respectively, which in turn serve as inexact approximations of $\nabla_x f(x\at{t}, y\at{t})$ and $\nabla_y f(x\at{t}, y\at{t})$. 
As a result, we encounter two types of errors:  
(i) \emph{systematic bias} (non-zero mean), bounded by $\delta_x, \delta_y$, and  
(ii) \emph{random noise} (zero mean), with variance bounded by $\sigma_x^2, \sigma_y^2$.
\end{model*}
\subsection{Alternating Methods}
%Since vanilla gradient methods in \emph{hidden convex-hidden concave} settings are generally divergent
A widely used stabilization technique in nonconvex min-max optimization, including adversarial training and GANs, is \emph{alternating} updates between players \citep{lu2019alternating,nouiehed2019solving,gidel2019negative,bailey2020finite,wibisono2022alternating,cevher2023alternation,lee2024fundamental}. Building on this approach, we analyze nested and alternating gradient iteration schemes.
\subsubsection{Nested Gradient Iterations}
As a warm-up, we focus on solving the minimax problem \eqref{eq:SP} under the assumption that we have access to an \emph{oracle for approximate inner maximization}, similar to \citep[Section 4]{jin2020local}. Specifically, for any given \( x \), the oracle provides a \( y' \) such that:  
$    f(x, y' ) \geq \max_y f(x, y) - \epsilon.
$
% This setting aligns with \emph{sequential game models}, where one player acts first, and the other responds. 
    \begin{align}
    \left\{
\arraycolsep=1.4pt\def\arraystretch{1.5}
        \begin{array}{rll}
        y\at{t+1} &\gets \mathtt{ARGMAX}(f(x\at{t}, \cdot), \epsilon_y);\\
        x\at{t+1} &\gets \proj_{\calX}\paren{x\at{t} - \eta \nabla f(x\at{t}, y\at{t+1})}
        \end{array} \label{eq:gdmax}\tag{\gdmax}
        \right.
    \end{align}

\begin{theorem}[NC-pPL; Informal version of \cref{theorem:formal-nestedgit-nc-ppl}]
     Let $f:\calX\times\calY$ be an $L$-Lipschitz and $\ell$-smooth function and $\calX, \calY$ be compact convex sets with diameters $\diam{\calX},\diam{\calY}$ respectively. Also, assume that $f(x,\cdot)$ satisfies the pPL condition with modulus $\mu>0$. Then, the iteration scheme %\eqref{agda}
     %\ref{alg:nepg}
     (\gdmax) 
     run with a tuning of 
     \begin{itemize}
         \item step-sizes: $\etax = \Theta\paren{\frac{1}{\ell\kappa}}$ and
     $\etay = \Theta\paren{\frac{1}{\ell}}$
         \item batch-sizes
     $ M_x =\Theta\paren{ \frac{ \sigma_x^2}{\epsilon^2}} $ and
     $M_y = \Theta\paren{\frac{ \kappa \sigma_y^2}{\epsilon^2} } $,
     \end{itemize}%
     where $\kappa:= \frac{\ell}{\mu}$,
     outputs an $(\epsilon + \delta_x + \delta_y )$-saddle-point of $f$ after a total number of outer-loop and inner-loop iterations, $T$, that is at most $$T = O\paren{\frac{\kappa^3 L \paren{\diam{\calX} +\diam{\calY} }}{\epsilon^2} \log\paren{\frac{1}{\epsilon}} }.$$
\end{theorem}

\begin{theorem}[pPL-pPL; Informal version of \cref{theorem:formal-nestedgit-2sided-ppl}]
\label{theorem:informal-nestedgit-2sided-ppl}
     Let $f:\calX\times\calY$ be an $L$-Lipschitz and $\ell$-smooth function and $\calX, \calY$ be compact convex sets with diameters $\diam{\calX},\diam{\calY}$ respectively. Further, assume that $f$ satisfies the two-sided pPL condition with moduli $\mux,\muy$. Then, the iteration scheme 
     \eqref{eq:gdmax}
      run with a tuning of 
     \begin{itemize}
         \item step-sizes:$\etax = \Theta\paren{\frac{1}{\ell\kappa }}$ and
     $\etay = \Theta\paren{\frac{1}{\ell}}$
         \item batch-sizes
     $ M_x =\Theta\paren{ \frac{ \kappa_x \sigma_x^2}{\epsilon}} $ and
     $M_y = \Theta\paren{\frac{ \ell \kappa_y \sigma_y^2}{\epsilon^2} } $,
     \end{itemize}%
     outputs an $\Big(\epsilon+\sqrt{\ell_{\Phi}/\mux}(\delta_x+\delta_y)\Big)$-saddle-point of $f$ after a number of iterations, $T$, that is at most $$T = O \paren{\frac{\ell^2}{\mu_x\mu_y}\log{\frac{\ell L \kappa_x \diam{\calX} }{\epsilon}}\log{\frac{\ell L \kappa_y \diam{\calY} }{\epsilon}} },$$
     where $\kappa_x := \frac{\ell}{\mux}$ and $\kappa_y := \frac{\ell}{\muy}$.
\end{theorem}

% \subsection{Two-Timescale Gradient Descent Ascent}
% \begin{align}
% \left\{
% \arraycolsep=1.4pt\def\arraystretch{1.5}
% \begin{array}{rll}
%     x\at{t}&\gets \proj_{\calX}\paren{x\at{t-1}-\etax\hatgradfbx(x\at{t-1}, y\at{t-1})}\\
%     y\at{t}&\gets \proj_{\calY}\paren{y\at{t-1}-\etay\hatgradfby(x\at{t-1}, y\at{t-1})}
% \end{array}
% \right.\tag{\ttgda}
% \label{ttgda}
% \end{align}
% \begin{algorithm}[htb]
%   \caption{Two-Timescale GDA}
%   \label{alg:ttgda}
% \begin{algorithmic}
%     \INPUT{$(x\at{0},y\at{0})$, step-sizes $\etax, \etay, T>0$.}
%   \FOR{$t=1$ {\bfseries to} $T$}
%         \STATE $x\at{t} \gets \proj_{\calX}\paren{x\at{t-1} - \etax \hatgradfbx (x\at{t-1},y\at{t-1})}$
%         \STATE $y\at{t} \gets \proj_{\calY}\paren{y\at{t-1} + \etay \hatgradfby (x\at{t-1},y\at{t-1}) }$
%   \ENDFOR
% \end{algorithmic}
% \end{algorithm}

% \begin{theorem}
% \end{theorem}

\subsubsection{Alternating Gradient Descent Ascent}
Given the simplicity and practical superiority of single-loop methods over two-loop alternatives, a natural question arises: \emph{can we achieve comparable convergence rates without resorting to multi-loop procedures?} To this end, \altgda\ leverages the sequential computation of  $x_{t+1}$  and  $y_{t+1}$ , ensuring that each update benefits from fresher gradient information. 
%can we achi stochas- tic GDA-type algorithms achieve the better sample complexity of O(ε−4) without a large batch size?
% \renewcommand{\proj}{\Pi}
\begin{align}
\left\{
\arraycolsep=1.4pt\def\arraystretch{1.5}
\begin{array}{rll}
    x\at{t}&\gets \proj_{\calX}\paren{x\at{t-1}-\etax\hatgradfbx(x\at{t-1}, y\at{t-1})}\\
    y\at{t}&\gets \proj_{\calY}\paren{y\at{t-1}-\etay\hatgradfby(x\at{t}, y\at{t-1})}
\end{array}
\right.\tag{\altgda}
\label{agda}
\end{align}
% \begin{algorithm}[htb]
%   \caption{Alternating GDA}
%   \label{alg:agda}
% \begin{algorithmic}
%     \INPUT{$(x\at{0},y\at{0})$, step-sizes $\etax, \etay, T>0$.}
%   \FOR{$t=1$ {\bfseries to} $T$}
%         \STATE $x\at{t} \gets \proj_{\calX}\paren{x\at{t-1} - \etax \hatgradfbx (x\at{t-1},y\at{t-1})}$
%         \STATE $y\at{t} \gets \proj_{\calY}\paren{y\at{t-1} + \etay \hatgradfby (x\at{t},y\at{t-1}) }$
%   \ENDFOR
% \end{algorithmic}
% \end{algorithm}

\begin{theorem}[NC-pPL; Informal version of
     \cref{theorem:formal-agda-ncppl}]
     Let $f:\calX\times\calY$ be an $L$-Lipschitz and $\ell$-smooth function and $\calX, \calY$ be compact convex sets with diameters $\diam{\calX},\diam{\calY}$ respectively. Also, assume that $f(x,\cdot)$ satisfies the pPL condition with modulus $\mu>0$ for all $x\in\calX$. Then, the iteration scheme \eqref{agda} run with a tuning of 
     \begin{itemize}
         \item step-sizes:$\etax = \Theta\paren{\frac{1}{\ell\kappa^2}}$ and
     $\etay = \Theta\paren{\frac{1}{\ell}}$
         \item batch-sizes
     $ M_x =\Theta\paren{ \frac{ \ell^2\kappa^2\sigma_x^2}{\epsilon^2}} $ and
     $M_y = \Theta\paren{\frac{ \kappa^2\sigma_y^2}{\epsilon^2} } $,
     \end{itemize}%
     outputs an $(\epsilon + \delta)$-saddle-point of $f$ after a number of iterations, $T$, that is at most $$T = O\paren{\frac{\kappa^2\ell L \paren{\diam{\calX} +\diam{\calY} }}{\epsilon^2}}.$$
\label{theorem:informal-agda-ncppl}
\end{theorem}
    %  $ M_x =\Theta\paren{ \frac{ \ell^2\kappa^2\sigma_x^2}{\epsilon^2}} $ and
    % $\delta_x = \Theta\paren{\frac{\epsilon}{ \ell\sqrt{\kappa} } }$
    %  $M_y = \Theta\paren{\frac{ \kappa^2\sigma_y^2}{\epsilon^2} } $
    %  $\delta_y = \Theta\paren{\frac{\epsilon}{\kappa}}$

\begin{theorem}[pPL-pPL; Informal version of \cref{theorem:formal-agda-2sided-ppl}]
\label{theorem:informal-agda-2sided-ppl}
     Let $f:\calX\times\calY$ be an $L$-Lipschitz and $\ell$-smooth function and $\calX, \calY$ be compact convex sets with diameters $\diam{\calX},\diam{\calY}$ respectively. Also, assume that $f$ satisfies the pPL condition in $x$ and $y$ with moduli $\mux,\muy>0$ respectively. Then, the iteration scheme \eqref{agda} run with a tuning of 
     \begin{itemize}
         \item step-sizes:$\etax = \Theta\paren{\frac{\muy}{\ell^3}}$ and
     $\etay = \Theta\paren{\frac{1}{\ell}}$
         \item batch-sizes
     $ M_x =\Theta\paren{ \frac{\sigma_x^2}{\mux\epsilon}} $ and
     $M_y = \Theta\paren{\frac{ \sigma_y^2}{\mux\muy^2\epsilon} } $,
     \end{itemize}%
     outputs an $\epsilon$-saddle-point of $f$ after a number of iterations, $T$, that is at most $$T =O\paren{\frac{\ell^3}{\mux\muy^2}\log{\frac{L(\diam{\calX} + \diam{\calY} )}{\epsilon}}}.$$
\end{theorem}
\section{Main Results -- Convex Markov Games}
\label{sec:cmgs-main-results}
In this section we present our main results regarding the convex Markov games (cMGs). We demonstrate that a utility function that is hidden concave in the occupancy measure satisfies the proximal-PL condition with respect to the player's policy. Then, we show that the maximizers of $U^\mu(x,\cdot)$ are Lipschitz continuous in the minimizing player's policy $x$. 
% Not only does this substantially improve the similar continuity result of \citep[Theorem 3.2]{kalogiannis2024learning}, but it also allows us to get faster rates of convergence.
Finally, we describe two policy gradient algorithms that enjoy convergence to an approximate Nash equilibrium using a finite number of samples and iterations. Namely, \begin{inparaenum}[(i)] \item nested policy gradient (\nepg), and \item alternating policy gradient descent-ascent (\apgda).
\end{inparaenum} 

Throughout, $\delta_x$ is implicitly tuned through the coefficient of the regularizer, $\mu$. Also, $\delta_y = 0$ since the maximizing player has access to unbiased stochastic estimates of the gradients. To see why $\mu$ controls $\delta_x$, we note that by the $L_{\mathrm{reg}}$-Lipschitz continuity of the regularizer, the norm of its gradient is bounded. As such, $\delta_x = O(\mu_{} L_{\mathrm{reg}})$ and since we also pick $\mu=O(\epsilon)$, the bound on $\delta_x$ follows suit.

\begin{theorem}[Continuity of maximizers]
    Consider the mapping $x\mapsto y^\star(x):= \argmax_{y\in\calY}f(x,y)$. 
    % Then, for any two points $x_1,x_2 \in\calX$ and $y^\star(x_1),y^\star(x_2) \in \calY$ with $y^\star(x_1):= \argmax_{y\in\calY} U^\mu(x_1, y)$ and , that is, $y^\star(x_2):= \argmax_{y\in\calY} U^\mu(x_2, y)$. 
    Then, for any two points $x_1, x_2$ it holds true that:
    \begin{align}
        \norm{y^\star(x_1) - y^\star(x_2)}\leq L_\star \norm{x_1 - x_2},
    \end{align}
    where $L_\star := \frac{\ell}{\sqrt{\mu\muqg}} $. 
    % \fk{fix}
    \label{theorem:contiuity-of-maximizers-cmgs}
\end{theorem}

\begin{corollary}
    Define the function $\Phi^\mu:\calX\to \R$ as $\Phi^\mu := \max_{y}\braces{U(x,y) - \frac{\mu}{2}\norm{\lambda_2(x,y)}^2 }.$ Then, for any $x,x'\in\calX$, the following inequality holds:
    \begin{align}
        \norm{\nabla_x \Phi^\mu(x) - \nabla_x \Phi^\mu(x') } \leq \ellphimu \norm{x - x'}.
    \end{align}
    with $\ellphimu:= \ell (1 + L_\star ).$ 
    \label{cor:maxfunc_holder}
\end{corollary}

\subsection{Nested Policy Gradient}
\begin{algorithm}[htb]
   \caption{\nepg: Nested Policy Gradient}
   \label{alg:nepg}
\begin{algorithmic}
    \INPUT{$(x\at{0},y\at{0})$, step-sizes $\etax, \etay$, regul. coeff. $\mu\geq0$}
   \FOR{$\tout=1$ {\bfseries to} $\Tout$}
       \STATE $y\at{t,0}\gets y\at{t-1}$
       \FOR{$\tin=1$ {\bfseries to} $\Tin$}
       \STATE $y\at{t,\tin} \gets \proj_{\calY} \paren{y\at{\tin-1} +
      \etay \hat{\nabla}_y U^\mu \paren{x\at{\tout-1}, y\at{t,\tin-1} }} $
      \STATE $y\at{t}\gets y\at{t,\Tin}$
       \ENDFOR
    \STATE $x\at{t}\gets \proj_\calX\paren{x\at{\tout-1} - \etax \hat{\nabla}_x U\paren{x\at{t-1}, y\at{t} } }$
   \ENDFOR
   \STATE Pick $t^\star\in\{1,\dots,T\}$ of the best iterate.
   \OUTPUT{$(x\at{t^\star}, y\at{t^\star+1})$.}
\end{algorithmic}
\end{algorithm}

\begin{theorem}
    Consider a two-player zero-sum cMG, $\Gamma$, and let $\epsilon$ be a desired accuracy $\epsilon>0$. Then, \cref{alg:nepg} run with appropriately tuned step-sizes $\etax, \etay>0$, $\varepsilon_x,\varepsilon_y>0$ batch-sizes $M_x,M_y>0$ outputs an $\epsilon$-approximate Nash equilibrium after a number of iterations that is at most:
    $$\textstyle O\paren{\frac{1}{\epsilon^{\frac{11}{2}}}} \poly\paren{\frac{1}{\minrho},\gamma,\frac{1}{1-\gamma},|\calS|,|\calA|+|\calB|}.$$
    \label{theorem:informal-cmg-nepg-hc}
\end{theorem}

\begin{theorem}
    Consider a two-player zero-sum cMG, $\Gamma$, with  with utility functions that are hidden strongly concave with moduli $\mux, \muy>0$ respectively.
    Additionally, let $\epsilon$ be a desired accuracy $\epsilon>0$. Then, \cref{alg:nepg} run with approprite step-sizes $\etax, \etay>0$, exploration parameters $\varepsilon_x,\varepsilon_y>0$, and batch-sizes $M_x,M_y>0$ outputs an $\epsilon$-approximate Nash equilibrium after a number of iterations that is at most:
    $$\textstyle O\paren{\log{\frac{1}{\epsilon}}}\poly\paren{\frac{1}{\mux}, \frac{1}{\muy},\frac{1}{\minrho},\gamma,\frac{1}{1-\gamma},|\calS|,|\calA|+|\calB|}.$$
    \label{theorem:informal-cmg-nepg-shc}
\end{theorem}

% \subsection{Two-Timescale Policy Gradient Descent Ascent}
% \begin{algorithm}[htb]
%   \caption{\ttpgda: Two-Timescale Policy GDA }
%   \label{alg:ttpg}
% \begin{algorithmic}
%     \INPUT{$(x\at{0},y\at{0})$, step-sizes $\etax, \etay, T>0$.}
%   \FOR{$t=1$ {\bfseries to} $T$}
%         \STATE $x\at{t} \gets \proj_{\calX}\paren{x\at{t-1} - \etax \hat{\nabla}_x U(x\at{t-1},y\at{t-1})}$
%         \STATE $y\at{t} \gets \proj_{\calY}\paren{y\at{t-1} + \etay \hat{\nabla}_y U(x\at{t-1},y\at{t-1}) }$
%   \ENDFOR
%       \STATE Pick $t^\star\in\{1,\dots,T\}$ of the best iterate.
%   \OUTPUT{$x\at{t^\star}$.}
% \end{algorithmic}
% \end{algorithm}

% % \begin{align}
% %     \begin{array}{rl}
% %      x\at{t} =& \proj\braces{x\at{t-1} - \eta_x \nabla_x v\paren{x\at{t-1}, y\at{t-1}} };\\
% %     y\at{t} =&  \proj \braces{y\at{t-1} +
% %       \eta_y \nabla_y v \paren{x\at{t-1}, y\at{t-1} }}.
% %     \end{array}
% %     \label{eqtt}
% % \end{align}

% \begin{theorem}
%     Consider a two-player zero-sum cMG, $\Gamma$ and let $\epsilon$ be $\epsilon>0$. Then, running \cref{alg:ttpg} twice, with step-sizes $\etax,\etay$ and then identically, outputs an $\epsilon$-NE in  
% \end{theorem}
% \begin{proof}
%     \citep{daskalakis2020independent} and \cref{prop:hc_to_graddom}.
% \end{proof}

\subsection{Alternating Policy Gradient Descent Ascent}
\begin{algorithm}[htb]
   \caption{\apgda: Alternating Policy GDA}
   \label{alg:apgda}
\begin{algorithmic}
    \INPUT{$(x\at{0},y\at{0})$, step-sizes $\etax, \etay, T>0$, regul. coeff. $\mu\geq0$}
   \FOR{$t=1$ {\bfseries to} $T$}
        \STATE $x\at{t} \gets \proj_{\calX}\paren{x\at{t-1} - \etax \hat{\nabla}_x U^\mu(x\at{t-1},y\at{t-1})}$
        \STATE $y\at{t} \gets \proj_{\calY}\paren{y\at{t-1} + \etay \hat{\nabla}_y U(x\at{t},y\at{t-1}) }$
   \ENDFOR
   \STATE Pick $t^\star\in\{1,\dots,T\}$ of the best iterate.
   \OUTPUT{$(x\at{t^\star}, y\at{t^\star})$.}
\end{algorithmic}
\end{algorithm}

\begin{theorem}[Informal Version of \cref{theorem:formal-apgda-nc-ppl}]
        Consider a two-player zero-sum cMG, $\Gamma$, and let $\epsilon$ be a desired accuracy $\epsilon>0$. Then, \cref{alg:apgda} run with appropriately step-sizes $\etax, \etay>0$, exploration parameters $\varepsilon_x,\varepsilon_y>0$ batch-sizes $M_x,M_y>0$, and a regularization coefficient $\mu= O(\epsilon)$,  outputs an $\epsilon$-approximate Nash equilibrium after a number of iterations that is at most:
    $$O\textstyle\paren{\frac{1}{\epsilon^6}} \poly\paren{\frac{1}{\minrho},\gamma,\frac{1}{1-\gamma},|\calS|,|\calA|+|\calB|}.$$
\end{theorem}

\begin{theorem}[Informal Version of \cref{theorem:formal-apgda-cmg-ppl-ppl}]
    \label{informal:apgda-hsc}
    Consider a two-player zero-sum cMG, $\Gamma$, with utility functions that are hidden strongly concave with moduli $\mux, \muy>0$ respectively.
    Additionally, let $\epsilon$ be a desired accuracy $\epsilon>0$. Then, \cref{alg:apgda} run with step-sizes $\etax, \etay>0$, and batch-sizes $M_x,M_y>0$ outputs an $\epsilon$-approximate Nash equilibrium after a number of iterations that is at most:
    $$\textstyle O\paren{\log \frac{1}{\epsilon}} \poly\paren{\frac{1}{\mux},\frac{1}{\muy},\frac{1}{\minrho},\gamma,\frac{1}{1-\gamma},|\calS|,|\calA|+|\calB|}.$$
    
\end{theorem}

We deem noteworthy the fact that \cref{informal:apgda-hsc} guarantees (expected) last-iterate convergence for a class of nonconvex games that take place over a constrained domain. 

\section{Numerical Results}\label{sec:experiments}

% \begin{figure}[htb]
%     \centering
%     \includegraphics[width=0.4\textwidth]{figures/alt_rps.png}    \caption{Exploitability decays towards a small, but positive value corresponding to the bias introduced by the regularization coefficient $\mu$. Results are averaged over $100$ trials, each running the algorithm with a different randomly initialized policy profile.}
%     \label{fig:apgda-converges}
% \end{figure}
\begin{figure*}[htb]
    \centering
    \includegraphics[width=\textwidth]{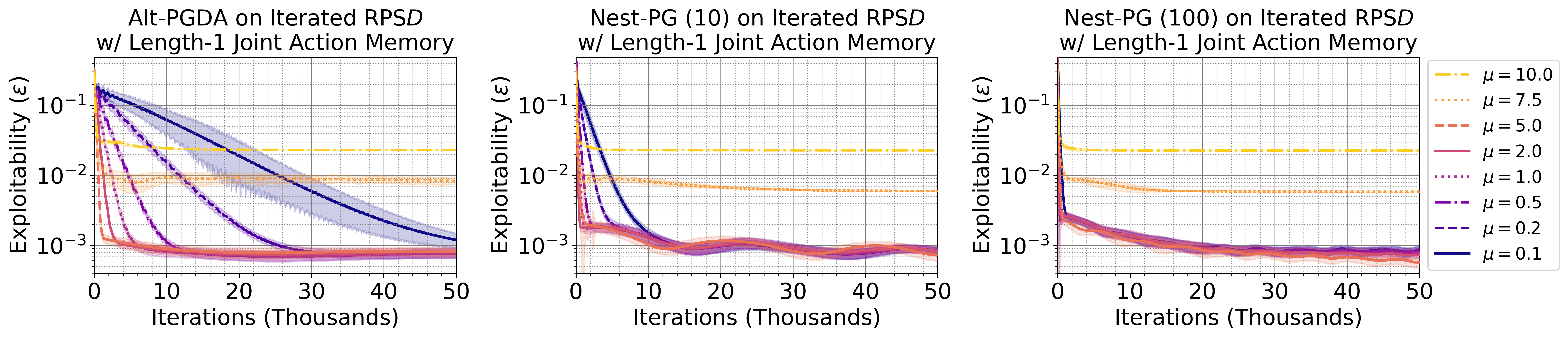}
    \caption{Exploitability decays towards a small, but positive value corresponding to the bias introduced by the regularization coefficient $\mu$. Results are averaged over $100$ trials, each running the algorithm with a different randomly initialized policy profile. The leftmost plot reports results for Algorithm~\ref{alg:apgda}; the right two report results for Algorithm~\ref{alg:nepg} with $T_{\mathrm{in}}=10$ and $100$ respectively.}
    \label{fig:apgda-converges}
\end{figure*}

We demonstrate Algorithms~\ref{alg:nepg} and~\ref{alg:apgda} on an iterated version of rock-paper-scissors-\emph{dummy} where each player remembers the actions selected in the previous round. Hence, the previous joint action constitutes the state in the Markov game. The \emph{dummy} action is dominated by all other actions such that the Nash equilibrium of the stage game is uniform across rock, paper, and scissors with zero mass on \emph{dummy}. We set the step-sizes $\tau_x=\tau_y=0.1$ and vary the regularization coefficient $\mu$ to demonstrate its effect in biasing convergence.

Our experiments suggest clearly that \begin{inparaenum}[(i)]
    \item there is a pronounced trade-off between speed of convergence and the exploitability of the solution of the perturbed problem \item more inner-loop iterations account for more stable exploitability decrease across consecutive iterates. 
\end{inparaenum}
\section{Conclusion and Future Work}\label{sec:future}
%\paragraph{Conclusion}

Convex Markov games (cMGs) unify the domains of convex MDPs and Markov games. Even the rudimentary form of a cMG, a two-player pure competition, fosters a rich mosaic of applications spanning language model alignment, self-driving cars, and creative chess playing, to name a few. 

In this work, we present the first algorithmic solution for computing Nash equilibria answering \cref{central-question}.
To achieve this,
we develop the first guarantees of convergence of alternating descent-ascent to a saddle-point for nonconvex functions that satisfy a one-sided or two-sided proximal Polyak-\L ojasiewicz condition over constrained domains. We utilized these results to design a number of \textit{independent policy gradient algorithms} for convex Markov games that provably \textit{converge} to an approximate Nash equilibrium. In a nutshell, we develop \textit{simple-to-use learning} dynamics that converge to the optimal game solutions even under realistic and challenging conditions where batch learning only allows noisy gradient estimation.

In terms of a message to practitioners, our work highlights that the regularized \eqref{agda} algorithm is exceptionally effective and easy to deploy to tackle min-max problems. Its efficacy is achieved by mere \begin{inparaenum}[(i)] \item\textit{regularization}, \item \textit{alternating updates}, and \item \textit{step-size magnitude separation} \end{inparaenum}. Moreover, the generality of \cref{theorem:informal-agda-ncppl,theorem:informal-agda-2sided-ppl}, coupled with the prevalence of the PL condition in modern machine learning objectives, makes a strong case in favor of \eqref{agda} as the default algorithm for min-max optimization. Furthermore, we believe that it can seamlessly accommodate more elaborate update schemes on top of (i-iii), like adaptive learning rates, preconditioning, and, if necessary, higher-order information. 

From a theoretical standpoint, we anticipate fascinating explorations of multi-player cMG interactions and a deep investigation of an array of game-theoretic solution concepts. We hope to see the experience of the rich multi-agent cMG applications inform theory and give rise to meaningful and challenging problems like equilibrium learning, equilibrium selection, and notions of equilibrium performance. We reiterate the fact that cMGs render the value-iteration   
subroutine useless. The latter lies at the heart of the corresponding algorithmic solutions for equilibrium learning or computation in conventional MGs; its inefficacy in cMGs, more than just posing an additional challenge, can stimulate the search for new algorithmic tools.
\section*{Acknowledgements}
FK gratefully acknowledges Panayotis Mertikopoulos for his mentoring and insightful discussions during the former's 2024 summer research fellowship at Archimedes, Athena Research Center. 

% \section*{Impact Statement}
% This paper presents work whose goal is to advance the field of 
% Machine Learning. There are many potential societal consequences 
% of our work, none which we feel must be specifically highlighted here.

% \section*{Acknowledgements}

\bibliography{main}

\begin{thebibliography}{139}
\providecommand{\natexlab}[1]{#1}
\providecommand{\url}[1]{\texttt{#1}}
\expandafter\ifx\csname urlstyle\endcsname\relax
  \providecommand{\doi}[1]{doi: #1}\else
  \providecommand{\doi}{doi: \begingroup \urlstyle{rm}\Url}\fi

\bibitem[Abel et~al.(2021)Abel, Dabney, Harutyunyan, Ho, Littman, Precup, and Singh]{abel2021expressivity}
David Abel, Will Dabney, Anna Harutyunyan, Mark~K. Ho, Michael Littman, Doina Precup, and Satinder Singh.
\newblock On the expressivity of markov reward.
\newblock In \emph{Advances in Neural Information Processing Systems (NeurIPS)}, 2021.

\bibitem[Afriat(1971)]{afriat1971theory}
SN~Afriat.
\newblock Theory of maxima and the method of lagrange.
\newblock \emph{SIAM Journal on Applied Mathematics}, 20\penalty0 (3):\penalty0 343--357, 1971.

\bibitem[Azizian et~al.(2024)Azizian, Iutzeler, Malick, and Mertikopoulos]{azizian24rate}
Wa\"{\i}ss Azizian, Franck Iutzeler, J\'{e}r\^{o}me Malick, and Panayotis Mertikopoulos.
\newblock The rate of convergence of bregman proximal methods: Local geometry versus regularity versus sharpness.
\newblock \emph{SIAM Journal on Optimization}, 34\penalty0 (3):\penalty0 2440--2471, 2024.
\newblock \doi{10.1137/23M1580218}.
\newblock URL \url{https://doi.org/10.1137/23M1580218}.

\bibitem[Bai and Jin(2020)]{bai2020provable}
Yu~Bai and Chi Jin.
\newblock Provable self-play algorithms for competitive reinforcement learning.
\newblock In \emph{International conference on machine learning}, pages 551--560. PMLR, 2020.

\bibitem[Bailey and Piliouras(2018)]{bailey2018multiplicative}
James~P Bailey and Georgios Piliouras.
\newblock Multiplicative weights update in zero-sum games.
\newblock In \emph{Proceedings of the 2018 ACM Conference on Economics and Computation}, pages 321--338, 2018.

\bibitem[Bailey et~al.(2020)Bailey, Gidel, and Piliouras]{bailey2020finite}
James~P Bailey, Gauthier Gidel, and Georgios Piliouras.
\newblock Finite regret and cycles with fixed step-size via alternating gradient descent-ascent.
\newblock In \emph{Conference on Learning Theory}, pages 391--407. PMLR, 2020.

\bibitem[Bakhtin et~al.(2022)Bakhtin, Wu, Lerer, Gray, Jacob, Farina, Miller, and Brown]{bakhtin2022mastering}
Anton Bakhtin, David~J Wu, Adam Lerer, Jonathan Gray, Athul~Paul Jacob, Gabriele Farina, Alexander~H Miller, and Noam Brown.
\newblock Mastering the game of no-press diplomacy via human-regularized reinforcement learning and planning.
\newblock \emph{arXiv preprint arXiv:2210.05492}, 2022.

\bibitem[Barakat et~al.(2023)Barakat, Fatkhullin, and He]{barakat2023reinforcement}
Anas Barakat, Ilyas Fatkhullin, and Niao He.
\newblock Reinforcement learning with general utilities: Simpler variance reduction and large state-action space.
\newblock In \emph{International Conference on Machine Learning}, pages 1753--1800. PMLR, 2023.

\bibitem[Bauschke and Combettes(2011)]{bauschke2011convex}
Heinz~H. Bauschke and Patrick~L. Combettes.
\newblock \emph{Convex analysis and monotone operator theory in Hilbert spaces}.
\newblock Springer Science \& Business Media, 2011.

\bibitem[Bellman(1958)]{BELLMAN1958228}
Richard Bellman.
\newblock Dynamic programming and stochastic control processes.
\newblock \emph{Information and Control}, 1\penalty0 (3):\penalty0 228--239, 1958.
\newblock ISSN 0019-9958.
\newblock \doi{https://doi.org/10.1016/S0019-9958(58)80003-0}.
\newblock URL \url{https://www.sciencedirect.com/science/article/pii/S0019995858800030}.

\bibitem[Ben-Tal and Teboulle(1996)]{ben1996hidden}
Aharon Ben-Tal and Marc Teboulle.
\newblock Hidden convexity in some nonconvex quadratically constrained quadratic programming.
\newblock \emph{Mathematical Programming}, 72\penalty0 (1):\penalty0 51--63, 1996.

\bibitem[Bernasconi et~al.(2024)Bernasconi, Castiglioni, Celli, and Farina]{bernasconi2024role}
Martino Bernasconi, Matteo Castiglioni, Andrea Celli, and Gabriele Farina.
\newblock On the role of constraints in the complexity of min-max optimization.
\newblock \emph{arXiv preprint arXiv:2411.03248}, 2024.

\bibitem[Berner et~al.(2019)Berner, Brockman, Chan, Cheung, Debiak, Dennison, Farhi, Fischer, Hashme, Hesse, et~al.]{berner2019dota}
Christopher Berner, Greg Brockman, Brooke Chan, Vicki Cheung, Przemys{\l}aw Debiak, Christy Dennison, David Farhi, Quirin Fischer, Shariq Hashme, Chris Hesse, et~al.
\newblock Dota 2 with large scale deep reinforcement learning.
\newblock \emph{arXiv preprint arXiv:1912.06680}, 2019.

\bibitem[Bof et~al.(2018)Bof, Carli, and Schenato]{bof2018lyapunov}
Nicoletta Bof, Ruggero Carli, and Luca Schenato.
\newblock Lyapunov theory for discrete time systems.
\newblock \emph{arXiv preprint arXiv:1809.05289}, 2018.

\bibitem[Bolte et~al.(2010)Bolte, Daniilidis, Ley, and Mazet]{bolte2010characterizations}
J{\'e}r{\^o}me Bolte, Aris Daniilidis, Olivier Ley, and Laurent Mazet.
\newblock Characterizations of {\l}ojasiewicz inequalities: subgradient flows, talweg, convexity.
\newblock \emph{Transactions of the American Mathematical Society}, 362\penalty0 (6):\penalty0 3319--3363, 2010.

\bibitem[Burgard et~al.(2000)Burgard, Moors, Fox, Simmons, and Thrun]{burgard2000collaborative}
Wolfram Burgard, Mark Moors, Dieter Fox, Reid Simmons, and Sebastian Thrun.
\newblock Collaborative multi-robot exploration.
\newblock In \emph{Proceedings 2000 ICRA. Millennium Conference. IEEE International Conference on Robotics and Automation. Symposia Proceedings (Cat. No. 00CH37065)}, volume~1, pages 476--481. IEEE, 2000.

\bibitem[Cai et~al.(2023)Cai, Luo, Wei, and Zheng]{cai2023uncoupled}
Yang Cai, Haipeng Luo, Chen-Yu Wei, and Weiqiang Zheng.
\newblock Uncoupled and convergent learning in two-player zero-sum markov games with bandit feedback.
\newblock \emph{Advances in Neural Information Processing Systems}, 36:\penalty0 36364--36406, 2023.

\bibitem[Cai et~al.(2024)Cai, Luo, Wei, and Zheng]{cai2024near}
Yang Cai, Haipeng Luo, Chen-Yu Wei, and Weiqiang Zheng.
\newblock Near-optimal policy optimization for correlated equilibrium in general-sum markov games.
\newblock In \emph{International Conference on Artificial Intelligence and Statistics}, pages 3889--3897. PMLR, 2024.

\bibitem[Cen et~al.(2022)Cen, Chi, Du, and Xiao]{cen2022faster}
Shicong Cen, Yuejie Chi, Simon~S Du, and Lin Xiao.
\newblock Faster last-iterate convergence of policy optimization in zero-sum markov games.
\newblock \emph{arXiv preprint arXiv:2210.01050}, 2022.

\bibitem[Cevher et~al.(2023)Cevher, Cutkosky, Kavis, Piliouras, Skoulakis, and Viano]{cevher2023alternation}
Volkan Cevher, Ashok Cutkosky, Ali Kavis, Georgios Piliouras, Stratis Skoulakis, and Luca Viano.
\newblock Alternation makes the adversary weaker in two-player games.
\newblock \emph{Advances in Neural Information Processing Systems}, 36:\penalty0 18263--18290, 2023.

\bibitem[Chambolle and Pock(2011)]{Chambolle2011AFP}
A.~Chambolle and Thomas Pock.
\newblock A first-order primal-dual algorithm for convex problems with applications to imaging.
\newblock \emph{Journal of Mathematical Imaging and Vision}, 40:\penalty0 120--145, 2011.
\newblock URL \url{https://api.semanticscholar.org/CorpusID:261281173}.

\bibitem[Chatterji et~al.(2021)Chatterji, Pacchiano, Bartlett, and Jordan]{chatterji2021once}
Niladri Chatterji, Aldo Pacchiano, Peter Bartlett, and Michael Jordan.
\newblock On the theory of reinforcement learning with once-per-episode feedback.
\newblock In \emph{Advances in Neural Information Processing Systems (NeurIPS)}, 2021.

\bibitem[Chavdarova et~al.(2019)Chavdarova, Gidel, Fleuret, and Lacoste-Julien]{chavdarova19reducing}
Tatjana Chavdarova, Gauthier Gidel, Fran\c{c}ois Fleuret, and Simon Lacoste-Julien.
\newblock Reducing noise in gan training with variance reduced extragradient.
\newblock In H.~Wallach, H.~Larochelle, A.~Beygelzimer, F.~d\textquotesingle Alch\'{e}-Buc, E.~Fox, and R.~Garnett, editors, \emph{Advances in Neural Information Processing Systems}, volume~32. Curran Associates, Inc., 2019.
\newblock URL \url{https://proceedings.neurips.cc/paper_files/paper/2019/file/58a2fc6ed39fd083f55d4182bf88826d-Paper.pdf}.

\bibitem[Chavdarova et~al.(2021)Chavdarova, Pagliardini, Stich, Fleuret, and Jaggi]{chavdarova2021tamingganslookaheadminmax}
Tatjana Chavdarova, Matteo Pagliardini, Sebastian~U. Stich, Francois Fleuret, and Martin Jaggi.
\newblock Taming gans with lookahead-minmax, 2021.
\newblock URL \url{https://arxiv.org/abs/2006.14567}.

\bibitem[Chen et~al.(2024)Chen, He, Hu, and Ye]{chen2024efficient}
Xin Chen, Niao He, Yifan Hu, and Zikun Ye.
\newblock Efficient algorithms for a class of stochastic hidden convex optimization and its applications in network revenue management.
\newblock \emph{Operations Research}, 2024.

\bibitem[Cheung(2019{\natexlab{a}})]{cheung2019exploration}
Wang~Chi Cheung.
\newblock Exploration-exploitation trade-off in reinforcement learning on online markov decision processes with global concave rewards.
\newblock \emph{arXiv preprint arXiv:1905.06466}, 2019{\natexlab{a}}.

\bibitem[Cheung(2019{\natexlab{b}})]{cheung2019regret}
Wang~Chi Cheung.
\newblock Regret minimization for reinforcement learning with vectorial feedback and complex objectives.
\newblock In \emph{Advances in Neural Information Processing Systems (NeurIPS)}, 2019{\natexlab{b}}.

\bibitem[Cui and Du(2022)]{cui2022offline}
Qiwen Cui and Simon~S Du.
\newblock When are offline two-player zero-sum markov games solvable?
\newblock \emph{Advances in Neural Information Processing Systems}, 35:\penalty0 25779--25791, 2022.

\bibitem[Cui et~al.(2023)Cui, Zhang, and Du]{cui2023breaking}
Qiwen Cui, Kaiqing Zhang, and Simon Du.
\newblock Breaking the curse of multiagents in a large state space: Rl in markov games with independent linear function approximation.
\newblock In \emph{The Thirty Sixth Annual Conference on Learning Theory}, pages 2651--2652. PMLR, 2023.

\bibitem[Daskalakis et~al.(2020)Daskalakis, Foster, and Golowich]{daskalakis2020independent}
Constantinos Daskalakis, Dylan~J Foster, and Noah Golowich.
\newblock Independent policy gradient methods for competitive reinforcement learning.
\newblock \emph{Advances in neural information processing systems}, 33:\penalty0 5527--5540, 2020.

\bibitem[Daskalakis et~al.(2021)Daskalakis, Skoulakis, and Zampetakis]{daskalakis2021complexity}
Constantinos Daskalakis, Stratis Skoulakis, and Manolis Zampetakis.
\newblock The complexity of constrained min-max optimization.
\newblock In \emph{Proceedings of the 53rd Annual ACM SIGACT Symposium on Theory of Computing}, pages 1466--1478, 2021.

\bibitem[Daskalakis et~al.(2023)Daskalakis, Golowich, and Zhang]{daskalakis2023complexity}
Constantinos Daskalakis, Noah Golowich, and Kaiqing Zhang.
\newblock The complexity of markov equilibrium in stochastic games.
\newblock In \emph{The Thirty Sixth Annual Conference on Learning Theory}, pages 4180--4234. PMLR, 2023.

\bibitem[Deng et~al.(2023)Deng, Li, Mguni, Wang, and Yang]{deng2023complexity}
Xiaotie Deng, Ningyuan Li, David Mguni, Jun Wang, and Yaodong Yang.
\newblock On the complexity of computing markov perfect equilibrium in general-sum stochastic games.
\newblock \emph{National Science Review}, 10\penalty0 (1):\penalty0 nwac256, 2023.

\bibitem[Devolder et~al.(2014)Devolder, Glineur, and Nesterov]{devolder2014first}
Olivier Devolder, Fran{\c{c}}ois Glineur, and Yurii Nesterov.
\newblock First-order methods of smooth convex optimization with inexact oracle.
\newblock \emph{Mathematical Programming}, 146:\penalty0 37--75, 2014.

\bibitem[Diakonikolas et~al.(2021)Diakonikolas, Daskalakis, and Jordan]{diakonikolas2021efficient}
Jelena Diakonikolas, Constantinos Daskalakis, and Michael~I Jordan.
\newblock Efficient methods for structured nonconvex-nonconcave min-max optimization.
\newblock In \emph{International Conference on Artificial Intelligence and Statistics}, pages 2746--2754. PMLR, 2021.

\bibitem[Ding et~al.(2022)Ding, Wei, Zhang, and Jovanovic]{ding2022independent}
Dongsheng Ding, Chen-Yu Wei, Kaiqing Zhang, and Mihailo Jovanovic.
\newblock Independent policy gradient for large-scale markov potential games: Sharper rates, function approximation, and game-agnostic convergence.
\newblock In \emph{International Conference on Machine Learning}, pages 5166--5220. PMLR, 2022.

\bibitem[Drusvyatskiy and Lewis(2018)]{drusvyatskiy2018error}
Dmitriy Drusvyatskiy and Adrian~S Lewis.
\newblock Error bounds, quadratic growth, and linear convergence of proximal methods.
\newblock \emph{Mathematics of Operations Research}, 43\penalty0 (3):\penalty0 919--948, 2018.

\bibitem[Drusvyatskiy and Paquette(2019)]{drusvyatskiy2019efficiency}
Dmitriy Drusvyatskiy and Courtney Paquette.
\newblock Efficiency of minimizing compositions of convex functions and smooth maps.
\newblock \emph{Mathematical Programming}, 178\penalty0 (1):\penalty0 503--558, 2019.

\bibitem[Efroni et~al.(2021)Efroni, Merlis, and Mannor]{efroni2021trajectory}
Yonathan Efroni, Nadav Merlis, and Shie Mannor.
\newblock Reinforcement learning with trajectory feedback.
\newblock In \emph{AAAI Conference on Artificial Intelligence}, 2021.

\bibitem[Erez et~al.(2023)Erez, Lancewicki, Sherman, Koren, and Mansour]{erez2023regret}
Liad Erez, Tal Lancewicki, Uri Sherman, Tomer Koren, and Yishay Mansour.
\newblock Regret minimization and convergence to equilibria in general-sum markov games.
\newblock In \emph{International Conference on Machine Learning}, pages 9343--9373. PMLR, 2023.

\bibitem[Facchinei and Pang(2003)]{facchinei2003finite}
Francisco Facchinei and Jong-Shi Pang.
\newblock \emph{Finite-dimensional variational inequalities and complementarity problems}.
\newblock Springer, 2003.

\bibitem[Facchinei and Pang(2007)]{facchinei2007finite}
Francisco Facchinei and Jong-Shi Pang.
\newblock \emph{Finite-dimensional variational inequalities and complementarity problems}.
\newblock Springer Science \& Business Media, 2007.

\bibitem[Fatkhullin et~al.(2023)Fatkhullin, He, and Hu]{fatkhullin2023stochastic}
Ilyas Fatkhullin, Niao He, and Yifan Hu.
\newblock Stochastic optimization under hidden convexity.
\newblock \emph{arXiv preprint arXiv:2401.00108}, 2023.

\bibitem[Feng and Shanthikumar(2018)]{feng2018supply}
Qi~Feng and J~George Shanthikumar.
\newblock Supply and demand functions in inventory models.
\newblock \emph{Operations Research}, 66\penalty0 (1):\penalty0 77--91, 2018.

\bibitem[Fink(1964)]{fink1964equilibrium}
Arlington~M Fink.
\newblock Equilibrium in a stochastic $ n $-person game.
\newblock \emph{Journal of science of the hiroshima university, series ai (mathematics)}, 28\penalty0 (1):\penalty0 89--93, 1964.

\bibitem[Geist et~al.(2022)Geist, P{\'e}rolat, Lauri{\`e}re, Elie, Perrin, Bachem, Munos, and Pietquin]{geist2022concave}
Matthieu Geist, Julien P{\'e}rolat, Mathieu Lauri{\`e}re, Romuald Elie, Sarah Perrin, Oliver Bachem, R{\'e}mi Munos, and Olivier Pietquin.
\newblock Concave utility reinforcement learning: The mean-field game viewpoint.
\newblock In \emph{International Conference on Autonomous Agents and Multiagent Systems (AAMAS)}, 2022.

\bibitem[Gemp et~al.(2024)Gemp, Haupt, Marris, Liu, and Piliouras]{gemp2024convex}
Ian Gemp, Andreas Haupt, Luke Marris, Siqi Liu, and Georgios Piliouras.
\newblock Convex markov games: A framework for fairness, imitation, and creativity in multi-agent learning.
\newblock \emph{arXiv preprint arXiv:2410.16600}, 2024.

\bibitem[Ghadimi et~al.(2016)Ghadimi, Lan, and Zhang]{ghadimi2016mini}
Saeed Ghadimi, Guanghui Lan, and Hongchao Zhang.
\newblock Mini-batch stochastic approximation methods for nonconvex stochastic composite optimization.
\newblock \emph{Mathematical Programming}, 155\penalty0 (1):\penalty0 267--305, 2016.

\bibitem[Gidel et~al.(2019)Gidel, Hemmat, Pezeshki, Le~Priol, Huang, Lacoste-Julien, and Mitliagkas]{gidel2019negative}
Gauthier Gidel, Reyhane~Askari Hemmat, Mohammad Pezeshki, R{\'e}mi Le~Priol, Gabriel Huang, Simon Lacoste-Julien, and Ioannis Mitliagkas.
\newblock Negative momentum for improved game dynamics.
\newblock In \emph{The 22nd International Conference on Artificial Intelligence and Statistics}, pages 1802--1811. PMLR, 2019.

\bibitem[Goodfellow et~al.(2014)Goodfellow, Pouget-Abadie, Mirza, Xu, Warde-Farley, Ozair, Courville, and Bengio]{goodfellow2014generative}
Ian~J Goodfellow, Jean Pouget-Abadie, Mehdi Mirza, Bing Xu, David Warde-Farley, Sherjil Ozair, Aaron Courville, and Yoshua Bengio.
\newblock Generative adversarial nets.
\newblock \emph{Advances in neural information processing systems}, 27, 2014.

\bibitem[Gronauer and Diepold(2022)]{gronauer2022multi}
Sven Gronauer and Klaus Diepold.
\newblock Multi-agent deep reinforcement learning: a survey.
\newblock \emph{Artificial Intelligence Review}, 55\penalty0 (2):\penalty0 895--943, 2022.

\bibitem[Hazan et~al.(2019)Hazan, Kakade, Singh, and Soest]{hazan2019provably}
Elad Hazan, Sham Kakade, Karan Singh, and Abby~Van Soest.
\newblock Provably efficient maximum entropy exploration.
\newblock In \emph{International Conference on Machine Learning (ICML)}, 2019.

\bibitem[Hughes et~al.(2018)Hughes, Leibo, Phillips, Tuyls, Due{\~n}ez-Guzman, Garc{\'\i}a~Casta{\~n}eda, Dunning, Zhu, McKee, Koster, Roff, and Graepel]{hughes2018inequity}
Edward Hughes, Joel~Z Leibo, Matthew Phillips, Karl Tuyls, Edgar Due{\~n}ez-Guzman, Antonio Garc{\'\i}a~Casta{\~n}eda, Iain Dunning, Tina Zhu, Kevin McKee, Raphael Koster, Heather Roff, and Thore Graepel.
\newblock Inequity aversion improves cooperation in intertemporal social dilemmas.
\newblock \emph{Advances in neural information processing systems}, 31, 2018.

\bibitem[J~Reddi et~al.(2016)J~Reddi, Sra, Poczos, and Smola]{j2016proximal}
Sashank J~Reddi, Suvrit Sra, Barnabas Poczos, and Alexander~J Smola.
\newblock Proximal stochastic methods for nonsmooth nonconvex finite-sum optimization.
\newblock \emph{Advances in neural information processing systems}, 29, 2016.

\bibitem[Jin et~al.(2020)Jin, Netrapalli, and Jordan]{jin2020local}
Chi Jin, Praneeth Netrapalli, and Michael Jordan.
\newblock What is local optimality in nonconvex-nonconcave minimax optimization?
\newblock In \emph{International conference on machine learning}, pages 4880--4889. PMLR, 2020.

\bibitem[Jin et~al.(2021)Jin, Liu, Wang, and Yu]{jin2021v}
Chi Jin, Qinghua Liu, Yuanhao Wang, and Tiancheng Yu.
\newblock V-learning--a simple, efficient, decentralized algorithm for multiagent rl.
\newblock \emph{arXiv preprint arXiv:2110.14555}, 2021.

\bibitem[Jin et~al.(2022)Jin, Muthukumar, and Sidford]{jin2022complexity}
Yujia Jin, Vidya Muthukumar, and Aaron Sidford.
\newblock The complexity of infinite-horizon general-sum stochastic games.
\newblock \emph{arXiv preprint arXiv:2204.04186}, 2022.

\bibitem[Kalogiannis and Panageas(2023)]{kalogiannis23polymatrix}
Fivos Kalogiannis and Ioannis Panageas.
\newblock Zero-sum polymatrix markov games: Equilibrium collapse and efficient computation of nash equilibria.
\newblock In A.~Oh, T.~Naumann, A.~Globerson, K.~Saenko, M.~Hardt, and S.~Levine, editors, \emph{Advances in Neural Information Processing Systems}, volume~36, pages 59996--60020. Curran Associates, Inc., 2023.

\bibitem[Kalogiannis et~al.(2022)Kalogiannis, Anagnostides, Panageas, Vlatakis-Gkaragkounis, Chatziafratis, and Stavroulakis]{kalogiannis2022efficiently}
Fivos Kalogiannis, Ioannis Anagnostides, Ioannis Panageas, Emmanouil-Vasileios Vlatakis-Gkaragkounis, Vaggos Chatziafratis, and Stelios Stavroulakis.
\newblock Efficiently computing nash equilibria in adversarial team markov games.
\newblock \emph{arXiv preprint arXiv:2208.02204}, 2022.

\bibitem[Kalogiannis et~al.(2024)Kalogiannis, Yan, and Panageas]{kalogiannis2024learning}
Fivos Kalogiannis, Jingming Yan, and Ioannis Panageas.
\newblock Learning equilibria in adversarial team markov games: A nonconvex-hidden-concave min-max optimization problem.
\newblock \emph{arXiv preprint arXiv:2410.05673}, 2024.

\bibitem[Karimi et~al.(2016)Karimi, Nutini, and Schmidt]{karimi2016linear}
Hamed Karimi, Julie Nutini, and Mark Schmidt.
\newblock Linear convergence of gradient and proximal-gradient methods under the polyak-{\l}ojasiewicz condition.
\newblock In \emph{Machine Learning and Knowledge Discovery in Databases: European Conference, ECML PKDD 2016, Riva del Garda, Italy, September 19-23, 2016, Proceedings, Part I 16}, pages 795--811. Springer, 2016.

\bibitem[Kinderlehrer and Stampacchia(2000)]{kinderlehrer2000introduction}
David Kinderlehrer and Guido Stampacchia.
\newblock \emph{An introduction to variational inequalities and their applications}.
\newblock SIAM, 2000.

\bibitem[Kobyzev et~al.(2020)Kobyzev, Prince, and Brubaker]{kobyzev2020normalizing}
Ivan Kobyzev, Simon~JD Prince, and Marcus~A Brubaker.
\newblock Normalizing flows: An introduction and review of current methods.
\newblock \emph{IEEE Transactions on Pattern Analysis and Machine Intelligence}, 43\penalty0 (11):\penalty0 3964--3979, 2020.

\bibitem[Korpelevich(1976)]{korpelevich1976extragradient}
G.~M. Korpelevich.
\newblock The extragradient method for finding saddle points and other problems.
\newblock \emph{Ekonomika i Matematicheskie Metody}, 12:\penalty0 747--756, 1976.

\bibitem[Lanctot et~al.(2019)Lanctot, Lockhart, Lespiau, Zambaldi, Upadhyay, P{\'e}rolat, Srinivasan, Timbers, Tuyls, Omidshafiei, et~al.]{lanctot2019openspiel}
Marc Lanctot, Edward Lockhart, Jean-Baptiste Lespiau, Vinicius Zambaldi, Satyaki Upadhyay, Julien P{\'e}rolat, Sriram Srinivasan, Finbarr Timbers, Karl Tuyls, Shayegan Omidshafiei, et~al.
\newblock Openspiel: A framework for reinforcement learning in games.
\newblock \emph{arXiv preprint arXiv:1908.09453}, 2019.

\bibitem[Lee et~al.(2024)Lee, Cho, and Yun]{lee2024fundamental}
Jaewook Lee, Hanseul Cho, and Chulhee Yun.
\newblock Fundamental benefit of alternating updates in minimax optimization.
\newblock \emph{arXiv preprint arXiv:2402.10475}, 2024.

\bibitem[Leonardos et~al.(2021)Leonardos, Overman, Panageas, and Piliouras]{leonardos2021global}
Stefanos Leonardos, Will Overman, Ioannis Panageas, and Georgios Piliouras.
\newblock Global convergence of multi-agent policy gradient in markov potential games.
\newblock \emph{arXiv preprint arXiv:2106.01969}, 2021.

\bibitem[Li et~al.(2005)Li, Wu, Joseph~Lee, Yang, and Zhang]{li2005hidden}
Duan Li, Zhi-You Wu, Heung-Wing Joseph~Lee, Xin-Min Yang, and Lian-Sheng Zhang.
\newblock Hidden convex minimization.
\newblock \emph{Journal of Global Optimization}, 31:\penalty0 211--233, 2005.

\bibitem[Li and Pong(2018)]{li2018calculus}
Guoyin Li and Ting~Kei Pong.
\newblock Calculus of the exponent of kurdyka--{\l}ojasiewicz inequality and its applications to linear convergence of first-order methods.
\newblock \emph{Foundations of computational mathematics}, 18\penalty0 (5):\penalty0 1199--1232, 2018.

\bibitem[Li and Li(2014)]{li2014holder}
Xiao-Bing Li and SJ~Li.
\newblock H{\"o}lder continuity of perturbed solution set for convex optimization problems.
\newblock \emph{Applied Mathematics and Computation}, 232:\penalty0 908--918, 2014.

\bibitem[Liao et~al.(2024)Liao, Ding, and Zheng]{liao2024error}
Feng-Yi Liao, Lijun Ding, and Yang Zheng.
\newblock Error bounds, pl condition, and quadratic growth for weakly convex functions, and linear convergences of proximal point methods.
\newblock In \emph{6th Annual Learning for Dynamics \& Control Conference}, pages 993--1005. PMLR, 2024.

\bibitem[Lin et~al.(2020)Lin, Jin, and Jordan]{lin2020gradient}
Tianyi Lin, Chi Jin, and Michael Jordan.
\newblock On gradient descent ascent for nonconvex-concave minimax problems.
\newblock In \emph{International conference on machine learning}, pages 6083--6093. PMLR, 2020.

\bibitem[Littman(1994)]{littman1994markov}
Michael~L Littman.
\newblock Markov games as a framework for multi-agent reinforcement learning.
\newblock In \emph{Machine learning proceedings 1994}, pages 157--163. Elsevier, 1994.

\bibitem[Liu et~al.(2021)Liu, Rafique, Lin, and Yang]{liu2021first}
Mingrui Liu, Hassan Rafique, Qihang Lin, and Tianbao Yang.
\newblock First-order convergence theory for weakly-convex-weakly-concave min-max problems.
\newblock \emph{Journal of Machine Learning Research}, 22\penalty0 (169):\penalty0 1--34, 2021.

\bibitem[Lu et~al.(2019)Lu, Singh, Chen, Chen, and Hong]{lu2019alternating}
Songtao Lu, Rahul Singh, Xiangyi Chen, Yongxin Chen, and Mingyi Hong.
\newblock Alternating gradient descent ascent for nonconvex min-max problems in robust learning and gans.
\newblock In \emph{2019 53rd Asilomar Conference on Signals, Systems, and Computers}, pages 680--684. IEEE, 2019.

\bibitem[Luo and Tseng(1993)]{luo1993error}
Zhi-Quan Luo and Paul Tseng.
\newblock Error bounds and convergence analysis of feasible descent methods: a general approach.
\newblock \emph{Annals of Operations Research}, 46\penalty0 (1):\penalty0 157--178, 1993.

\bibitem[Madry et~al.(2017)Madry, Makelov, Schmidt, Tsipras, and Vladu]{madry2017towards}
Aleksander Madry, Aleksandar Makelov, Ludwig Schmidt, Dimitris Tsipras, and Adrian Vladu.
\newblock Towards deep learning models resistant to adversarial attacks.
\newblock \emph{arXiv preprint arXiv:1706.06083}, 2017.

\bibitem[Martinet(1970)]{martinet1970regularisation}
B.~Martinet.
\newblock Régularisation d’inéquations variationnelles par approximations successives.
\newblock \emph{ESAIM: Mathematical Modelling and Numerical Analysis - Modélisation Mathématique et Analyse Numérique}, 1970.

\bibitem[Mertikopoulos et~al.(2018)Mertikopoulos, Lecouat, Zenati, Foo, Chandrasekhar, and Piliouras]{mertikopoulos2018optimistic}
Panayotis Mertikopoulos, Bruno Lecouat, Houssam Zenati, Chuan-Sheng Foo, Vijay Chandrasekhar, and Georgios Piliouras.
\newblock Optimistic mirror descent in saddle-point problems: Going the extra (gradient) mile.
\newblock \emph{arXiv preprint arXiv:1807.02629}, 2018.

\bibitem[Mladenovic et~al.(2021)Mladenovic, Sakos, Gidel, and Piliouras]{mladenovic2021generalized}
Andjela Mladenovic, Iosif Sakos, Gauthier Gidel, and Georgios Piliouras.
\newblock Generalized natural gradient flows in hidden convex-concave games and gans.
\newblock In \emph{International Conference on Learning Representations}, 2021.

\bibitem[Mnih et~al.(2013)Mnih, Kavukcuoglu, Silver, Graves, Antonoglou, Wierstra, and Riedmiller]{mnih2013playing}
Volodymyr Mnih, Koray Kavukcuoglu, David Silver, Alex Graves, Ioannis Antonoglou, Daan Wierstra, and Martin Riedmiller.
\newblock Playing atari with deep reinforcement learning.
\newblock \emph{arXiv preprint arXiv:1312.5602}, 2013.

\bibitem[Mulvaney-Kemp et~al.(2022)Mulvaney-Kemp, Park, Jin, and Lavaei]{mulvaney2022dynamic}
Julie Mulvaney-Kemp, SangWoo Park, Ming Jin, and Javad Lavaei.
\newblock Dynamic regret bounds for constrained online nonconvex optimization based on polyak--lojasiewicz regions.
\newblock \emph{IEEE Transactions on Control of Network Systems}, 10\penalty0 (2):\penalty0 599--611, 2022.

\bibitem[Nachum et~al.(2019)Nachum, Dai, Kostrikov, Chow, Li, and Schuurmans]{nachum2019algaedice}
Ofir Nachum, Bo~Dai, Ilya Kostrikov, Yinlam Chow, Lihong Li, and Dale Schuurmans.
\newblock Algaedice: Policy gradient from arbitrary experience.
\newblock \emph{arXiv preprint arXiv:1912.02074}, 2019.

\bibitem[Nemirovski(2004)]{nemirovski2004prox}
Arkadi Nemirovski.
\newblock Prox-method with rate of convergence o(1/t) for variational inequalities with lipschitz continuous monotone operators.
\newblock \emph{SIAM Journal on Optimization}, 15\penalty0 (1):\penalty0 229--251, 2004.

\bibitem[Nemirovski et~al.(2009)Nemirovski, Juditsky, Lan, and Shapiro]{nemirovski2009robust}
Arkadi Nemirovski, Anatoli Juditsky, Guanghui Lan, and Alexander Shapiro.
\newblock Robust stochastic approximation approach to stochastic programming.
\newblock \emph{SIAM Journal on Optimization}, 19\penalty0 (4):\penalty0 1574--1609, 2009.

\bibitem[Nesterov(2005)]{nesterov2005smooth}
Yu~Nesterov.
\newblock Smooth minimization of non-smooth functions.
\newblock \emph{Mathematical programming}, 103:\penalty0 127--152, 2005.

\bibitem[Neu et~al.(2017)Neu, Jonsson, and G{\'o}mez]{neu2017unified}
Gergely Neu, Anders Jonsson, and Vicen{\c{c}} G{\'o}mez.
\newblock A unified view of entropy-regularized markov decision processes.
\newblock \emph{arXiv preprint arXiv:1705.07798}, 2017.

\bibitem[Nouiehed et~al.(2019)Nouiehed, Sanjabi, Huang, Lee, and Razaviyayn]{nouiehed2019solving}
Maher Nouiehed, Maziar Sanjabi, Tianjian Huang, Jason~D Lee, and Meisam Razaviyayn.
\newblock Solving a class of non-convex min-max games using iterative first order methods.
\newblock \emph{Advances in Neural Information Processing Systems}, 32, 2019.

\bibitem[Oikonomidis et~al.(2025)Oikonomidis, Laude, and Patrinos]{oikonomidis2025forward}
Konstantinos Oikonomidis, Emanuel Laude, and Panagiotis Patrinos.
\newblock Forward-backward splitting under the light of generalized convexity.
\newblock \emph{arXiv preprint arXiv:2503.18098}, 2025.

\bibitem[Papadimitriou et~al.(2023)Papadimitriou, Vlatakis-Gkaragkounis, and Zampetakis]{papadimitriou2023computational}
Christos Papadimitriou, Emmanouil-Vasileios Vlatakis-Gkaragkounis, and Manolis Zampetakis.
\newblock The computational complexity of multi-player concave games and kakutani fixed points.
\newblock In \emph{Proceedings of the 24th ACM Conference on Economics and Computation}, 2023.

\bibitem[Park et~al.(2023)Park, Zhang, and Ozdaglar]{park2023multi}
Chanwoo Park, Kaiqing Zhang, and Asuman Ozdaglar.
\newblock Multi-player zero-sum markov games with networked separable interactions.
\newblock \emph{Advances in Neural Information Processing Systems}, 36:\penalty0 37354--37369, 2023.

\bibitem[Peters et~al.(2010)Peters, Mulling, and Altun]{peters2010relative}
Jan Peters, Katharina Mulling, and Yasemin Altun.
\newblock Relative entropy policy search.
\newblock In \emph{Proceedings of the AAAI Conference on Artificial Intelligence}, volume~24, pages 1607--1612, 2010.

\bibitem[Polyak and Juditsky(1992)]{polyak1992acceleration}
Boris~T. Polyak and Anatoli~B. Juditsky.
\newblock Acceleration of stochastic approximation by averaging.
\newblock \emph{SIAM Journal on Control and Optimization}, 30\penalty0 (4):\penalty0 838--855, 1992.

\bibitem[Puterman(2014)]{puterman2014markov}
Martin~L Puterman.
\newblock \emph{Markov decision processes: discrete stochastic dynamic programming}.
\newblock John Wiley \& Sons, 2014.

\bibitem[Ratliff et~al.(2006)Ratliff, Bagnell, and Zinkevich]{ratliff2006maximum}
Nathan~D Ratliff, J~Andrew Bagnell, and Martin~A Zinkevich.
\newblock Maximum margin planning.
\newblock In \emph{Proceedings of the 23rd international conference on Machine learning}, pages 729--736, 2006.

\bibitem[Rebjock and Boumal(2024)]{rebjock2024fast}
Quentin Rebjock and Nicolas Boumal.
\newblock Fast convergence to non-isolated minima: four equivalent conditions for c 2 functions.
\newblock \emph{Mathematical Programming}, pages 1--49, 2024.

\bibitem[Rockafellar(1976)]{rockafellar1976monotone}
R.~Tyrrell Rockafellar.
\newblock Monotone operators and the proximal point algorithm.
\newblock \emph{SIAM Journal on Control and Optimization}, 14\penalty0 (5):\penalty0 877--898, 1976.

\bibitem[Rockafellar and Wets(2009)]{rockafellar2009variational}
R~Tyrrell Rockafellar and Roger J-B Wets.
\newblock \emph{Variational analysis}, volume 317.
\newblock Springer Science \& Business Media, 2009.

\bibitem[Rogers et~al.(2013)Rogers, Nieto-Granda, and Christensen]{rogers2013coordination}
John~G Rogers, Carlos Nieto-Granda, and Henrik~I Christensen.
\newblock Coordination strategies for multi-robot exploration and mapping.
\newblock In \emph{Experimental Robotics: The 13th International Symposium on Experimental Robotics}, pages 231--243. Springer, 2013.

\bibitem[Ruppert(1988)]{ruppert1988efficient}
David Ruppert.
\newblock Efficient estimations from a slowly convergent robbins-monro process.
\newblock \emph{Technical Report, Cornell University}, 1988.

\bibitem[Sakos et~al.(2024)Sakos, Vlatakis-Gkaragkounis, Mertikopoulos, and Piliouras]{sakos2024exploiting}
Iosif Sakos, Emmanouil-Vasileios Vlatakis-Gkaragkounis, Panayotis Mertikopoulos, and Georgios Piliouras.
\newblock Exploiting hidden structures in non-convex games for convergence to nash equilibrium.
\newblock \emph{Advances in Neural Information Processing Systems}, 36, 2024.

\bibitem[Sayin et~al.(2021)Sayin, Zhang, Leslie, Basar, and Ozdaglar]{sayin2021decentralized}
Muhammed Sayin, Kaiqing Zhang, David Leslie, Tamer Basar, and Asuman Ozdaglar.
\newblock Decentralized q-learning in zero-sum markov games.
\newblock \emph{Advances in Neural Information Processing Systems}, 34:\penalty0 18320--18334, 2021.

\bibitem[Schneider and Wagner(1957)]{schneider1957}
S.~Schneider and D.~H. Wagner.
\newblock Error detection in redundant systems.
\newblock In \emph{Papers Presented at the February 26-28, 1957, Western Joint Computer Conference: Techniques for Reliability}, IRE-AIEE-ACM '57 (Western), page 115–121, New York, NY, USA, 1957. Association for Computing Machinery.
\newblock ISBN 9781450378611.
\newblock \doi{10.1145/1455567.1455587}.
\newblock URL \url{https://doi.org/10.1145/1455567.1455587}.

\bibitem[Shalev-Shwartz et~al.(2016)Shalev-Shwartz, Shammah, and Shashua]{shalev2016safe}
Shai Shalev-Shwartz, Shaked Shammah, and Amnon Shashua.
\newblock Safe, multi-agent, reinforcement learning for autonomous driving.
\newblock \emph{arXiv preprint arXiv:1610.03295}, 2016.

\bibitem[Shapley(1953)]{shapley1953stochastic}
Lloyd~S Shapley.
\newblock Stochastic games.
\newblock \emph{Proceedings of the national academy of sciences}, 39\penalty0 (10):\penalty0 1095--1100, 1953.

\bibitem[Silver et~al.(2017)Silver, Schrittwieser, Simonyan, Antonoglou, Huang, Guez, Hubert, Baker, Lai, Bolton, et~al.]{silver2017mastering}
David Silver, Julian Schrittwieser, Karen Simonyan, Ioannis Antonoglou, Aja Huang, Arthur Guez, Thomas Hubert, Lucas Baker, Matthew Lai, Adrian Bolton, et~al.
\newblock Mastering the game of go without human knowledge.
\newblock \emph{nature}, 550\penalty0 (7676):\penalty0 354--359, 2017.

\bibitem[Singh et~al.(2022)Singh, Kumar, and Singh]{singh2022reinforcement}
Bharat Singh, Rajesh Kumar, and Vinay~Pratap Singh.
\newblock Reinforcement learning in robotic applications: a comprehensive survey.
\newblock \emph{Artificial Intelligence Review}, 55\penalty0 (2):\penalty0 945--990, 2022.

\bibitem[Song et~al.(2021)Song, Mei, and Bai]{song2021can}
Ziang Song, Song Mei, and Yu~Bai.
\newblock When can we learn general-sum markov games with a large number of players sample-efficiently?
\newblock \emph{arXiv preprint arXiv:2110.04184}, 2021.

\bibitem[Stella et~al.(2017)Stella, Themelis, and Patrinos]{stella2017forward}
Lorenzo Stella, Andreas Themelis, and Panagiotis Patrinos.
\newblock Forward--backward quasi-newton methods for nonsmooth optimization problems.
\newblock \emph{Computational Optimization and Applications}, 67\penalty0 (3):\penalty0 443--487, 2017.

\bibitem[Syed and Schapire(2007)]{syed2007game}
Umar Syed and Robert~E Schapire.
\newblock A game-theoretic approach to apprenticeship learning.
\newblock \emph{Advances in neural information processing systems}, 20, 2007.

\bibitem[Tammelin et~al.(2015)Tammelin, Burch, Johanson, and Bowling]{tammelin15solving}
Oskari Tammelin, Neil Burch, Michael Johanson, and Michael Bowling.
\newblock Solving heads-up limit texas hold'em.
\newblock In \emph{Proceedings of the 24th International Conference on Artificial Intelligence}, IJCAI'15, page 645–652. AAAI Press, 2015.
\newblock ISBN 9781577357384.

\bibitem[Tan et~al.(2022)Tan, Bejarano, Zhu, Ren, and Nejat]{tan2022deep}
Aaron~Hao Tan, Federico~Pizarro Bejarano, Yuhan Zhu, Richard Ren, and Goldie Nejat.
\newblock Deep reinforcement learning for decentralized multi-robot exploration with macro actions.
\newblock \emph{IEEE Robotics and Automation Letters}, 8\penalty0 (1):\penalty0 272--279, 2022.

\bibitem[v.~Neumann(1928)]{v1928theorie}
J~v.~Neumann.
\newblock Zur theorie der gesellschaftsspiele.
\newblock \emph{Mathematische annalen}, 100\penalty0 (1):\penalty0 295--320, 1928.

\bibitem[Vlatakis-Gkaragkounis et~al.(2021)Vlatakis-Gkaragkounis, Flokas, and Piliouras]{vlatakis2021solving}
Emmanouil-Vasileios Vlatakis-Gkaragkounis, Lampros Flokas, and Georgios Piliouras.
\newblock Solving min-max optimization with hidden structure via gradient descent ascent.
\newblock \emph{Advances in Neural Information Processing Systems}, 34:\penalty0 2373--2386, 2021.

\bibitem[Von~Stengel and Koller(1997)]{von1997team}
Bernhard Von~Stengel and Daphne Koller.
\newblock Team-maxmin equilibria.
\newblock \emph{Games and Economic Behavior}, 21\penalty0 (1-2):\penalty0 309--321, 1997.

\bibitem[Wang et~al.(2021)Wang, Lacotte, and Pilanci]{wang2020hidden}
Yifei Wang, Jonathan Lacotte, and Mert Pilanci.
\newblock The hidden convex optimization landscape of regularized two-layer relu networks: an exact characterization of optimal solutions.
\newblock In \emph{International Conference on Learning Representations}, 2021.

\bibitem[Wang et~al.(2023)Wang, Liu, Bai, and Jin]{wang2023breaking}
Yuanhao Wang, Qinghua Liu, Yu~Bai, and Chi Jin.
\newblock Breaking the curse of multiagency: Provably efficient decentralized multi-agent rl with function approximation.
\newblock In \emph{The Thirty Sixth Annual Conference on Learning Theory}, pages 2793--2848. PMLR, 2023.

\bibitem[Wei et~al.(2021)Wei, Lee, Zhang, and Luo]{wei2021last}
Chen-Yu Wei, Chung-Wei Lee, Mengxiao Zhang, and Haipeng Luo.
\newblock Last-iterate convergence of decentralized optimistic gradient descent/ascent in infinite-horizon competitive markov games.
\newblock In \emph{Conference on learning theory}, pages 4259--4299. PMLR, 2021.

\bibitem[Wibisono et~al.(2022)Wibisono, Tao, and Piliouras]{wibisono2022alternating}
Andre Wibisono, Molei Tao, and Georgios Piliouras.
\newblock Alternating mirror descent for constrained min-max games.
\newblock \emph{Advances in Neural Information Processing Systems}, 35:\penalty0 35201--35212, 2022.

\bibitem[Williams(1992)]{williams1992simple}
Ronald~J Williams.
\newblock Simple statistical gradient-following algorithms for connectionist reinforcement learning.
\newblock \emph{Machine learning}, 8:\penalty0 229--256, 1992.

\bibitem[Wu et~al.(2025)Wu, Viano, Chen, Zhu, Antonakopoulos, Gu, and Cevher]{wu2025multi}
Yongtao Wu, Luca Viano, Yihang Chen, Zhenyu Zhu, Kimon Antonakopoulos, Quanquan Gu, and Volkan Cevher.
\newblock Multi-step alignment as markov games: An optimistic online gradient descent approach with convergence guarantees.
\newblock \emph{arXiv preprint arXiv:2502.12678}, 2025.

\bibitem[Wu et~al.(2007)Wu, Li, Zhang, and Yang]{wu2007peeling}
Zhi-You Wu, Duan Li, Lian-Sheng Zhang, and XM~Yang.
\newblock Peeling off a nonconvex cover of an actual convex problem: hidden convexity.
\newblock \emph{SIAM Journal on Optimization}, 18\penalty0 (2):\penalty0 507--536, 2007.

\bibitem[Xia(2020)]{xia2020survey}
Yong Xia.
\newblock A survey of hidden convex optimization.
\newblock \emph{Journal of the Operations Research Society of China}, 8\penalty0 (1):\penalty0 1--28, 2020.

\bibitem[Yang et~al.(2020)Yang, Kiyavash, and He]{yang2020global}
Junchi Yang, Negar Kiyavash, and Niao He.
\newblock Global convergence and variance reduction for a class of nonconvex-nonconcave minimax problems.
\newblock \emph{Advances in Neural Information Processing Systems}, 33:\penalty0 1153--1165, 2020.

\bibitem[Yang et~al.(2022)Yang, Orvieto, Lucchi, and He]{yang2022faster}
Junchi Yang, Antonio Orvieto, Aurelien Lucchi, and Niao He.
\newblock Faster single-loop algorithms for minimax optimization without strong concavity.
\newblock In \emph{International Conference on Artificial Intelligence and Statistics}, pages 5485--5517. PMLR, 2022.

\bibitem[Yang and Ma(2022)]{yang2022t}
Yuepeng Yang and Cong Ma.
\newblock O(t-1) convergence of optimistic-follow-the-regularized-leader in two-player zero-sum markov games.
\newblock \emph{arXiv preprint arXiv:2209.12430}, 2022.

\bibitem[Ying et~al.(2023)Ying, Guo, Ding, Lavaei, and Shen]{ying2023policy}
Donghao Ying, Mengzi~Amy Guo, Yuhao Ding, Javad Lavaei, and Zuo-Jun Shen.
\newblock Policy-based primal-dual methods for convex constrained markov decision processes.
\newblock \emph{Proceedings of the AAAI Conference on Artificial Intelligence}, 37:\penalty0 10963--10971, 2023.

\bibitem[Zahavy et~al.(2021)Zahavy, O’Donoghue, Desjardins, and Singh]{zahavy2021reward}
Tom Zahavy, Brendan O’Donoghue, Guillaume Desjardins, and Satinder Singh.
\newblock Reward is enough for convex mdps.
\newblock In \emph{Advances in Neural Information Processing Systems (NeurIPS)}, 2021.

\bibitem[Zahavy et~al.(2022)Zahavy, Schroecker, Behbahani, Baumli, Flennerhag, Hou, and Singh]{zahavy2022discovering}
Tom Zahavy, Yannick Schroecker, Feryal Behbahani, Kate Baumli, Sebastian Flennerhag, Shaobo Hou, and Satinder Singh.
\newblock Discovering policies with domino: Diversity optimization maintaining near optimality.
\newblock \emph{arXiv preprint arXiv:2205.13521}, 2022.

\bibitem[Zahavy et~al.(2023)Zahavy, Veeriah, Hou, Waugh, Lai, Leurent, Tomasev, Schut, Hassabis, and Singh]{zahavy2023diversifying}
Tom Zahavy, Vivek Veeriah, Shaobo Hou, Kevin Waugh, Matthew Lai, Edouard Leurent, Nenad Tomasev, Lisa Schut, Demis Hassabis, and Satinder Singh.
\newblock Diversifying ai: Towards creative chess with alphazero.
\newblock \emph{arXiv preprint arXiv:2308.09175}, 2023.

\bibitem[Zeng et~al.(2022)Zeng, Doan, and Romberg]{zeng2022regularized}
Sihan Zeng, Thinh Doan, and Justin Romberg.
\newblock Regularized gradient descent ascent for two-player zero-sum markov games.
\newblock \emph{Advances in Neural Information Processing Systems}, 35:\penalty0 34546--34558, 2022.

\bibitem[Zhang and Yu(2019)]{Zhang2019ConvergenceOG}
Guojun Zhang and Yaoliang Yu.
\newblock Convergence of gradient methods on bilinear zero-sum games.
\newblock In \emph{International Conference on Learning Representations}, 2019.
\newblock URL \url{https://api.semanticscholar.org/CorpusID:211069144}.

\bibitem[Zhang et~al.(2020)Zhang, Koppel, Bedi, Szepesvari, and Wang]{zhang2020variational}
Junyu Zhang, Alec Koppel, Amrit~Singh Bedi, Csaba Szepesvari, and Mengdi Wang.
\newblock Variational policy gradient method for reinforcement learning with general utilities.
\newblock \emph{Advances in Neural Information Processing Systems}, 33:\penalty0 4572--4583, 2020.

\bibitem[Zhang et~al.(2021)Zhang, Ni, Yu, Szepesvari, and Wang]{zhang2021convergence}
Junyu Zhang, Chengzhuo Ni, Zheng Yu, Csaba Szepesvari, and Mengdi Wang.
\newblock On the convergence and sample efficiency of variance-reduced policy gradient method.
\newblock \emph{Advances in Neural Information Processing Systems}, 34:\penalty0 2228--2240, 2021.

\bibitem[Zhang et~al.(2022{\natexlab{a}})Zhang, Liu, Wang, Xiong, Li, and Bai]{zhang2022policy}
Runyu Zhang, Qinghua Liu, Huan Wang, Caiming Xiong, Na~Li, and Yu~Bai.
\newblock Policy optimization for markov games: Unified framework and faster convergence.
\newblock \emph{Advances in Neural Information Processing Systems}, 35:\penalty0 21886--21899, 2022{\natexlab{a}}.

\bibitem[Zhang et~al.(2022{\natexlab{b}})Zhang, Mei, Dai, Schuurmans, and Li]{zhang2022global}
Runyu Zhang, Jincheng Mei, Bo~Dai, Dale Schuurmans, and Na~Li.
\newblock On the global convergence rates of decentralized softmax gradient play in markov potential games.
\newblock \emph{Advances in Neural Information Processing Systems}, 35:\penalty0 1923--1935, 2022{\natexlab{b}}.

\bibitem[Zheng et~al.(2023)Zheng, Zhu, So, Blanchet, and Li]{zheng2023universal}
Taoli Zheng, Linglingzhi Zhu, Anthony Man-Cho So, Jos{\'e} Blanchet, and Jiajin Li.
\newblock Universal gradient descent ascent method for nonconvex-nonconcave minimax optimization.
\newblock \emph{Advances in Neural Information Processing Systems}, 36:\penalty0 54075--54110, 2023.

\bibitem[Ziebart et~al.(2008)Ziebart, Maas, Bagnell, Dey, et~al.]{ziebart2008maximum}
Brian~D Ziebart, Andrew~L Maas, J~Andrew Bagnell, Anind~K Dey, et~al.
\newblock Maximum entropy inverse reinforcement learning.
\newblock In \emph{Aaai}, volume~8, pages 1433--1438. Chicago, IL, USA, 2008.

\bibitem[Zimin and Neu(2013)]{zimin2013online}
Alexander Zimin and Gergely Neu.
\newblock Online learning in episodic markovian decision processes by relative entropy policy search.
\newblock \emph{Advances in neural information processing systems}, 26, 2013.

\end{thebibliography}
% \bibliographystyle{icml2025}

%%%%%%%%%%%%%%%%%%%%%%%%%%%%%%%%%%%%%%%%%%%%%%%%%%%%%%%%%%%%%%%%%%%%%%%%%%%%%%%
%%%%%%%%%%%%%%%%%%%%%%%%%%%%%%%%%%%%%%%%%%%%%%%%%%%%%%%%%%%%%%%%%%%%%%%%%%%%%%%
% APPENDIX
%%%%%%%%%%%%%%%%%%%%%%%%%%%%%%%%%%%%%%%%%%%%%%%%%%%%%%%%%%%%%%%%%%%%%%%%%%%%%%%
%%%%%%%%%%%%%%%%%%%%%%%%%%%%%%%%%%%%%%%%%%%%%%%%%%%%%%%%%%%%%%%%%%%%%%%%%%%%%%%
\newpage
\onecolumn

\appendix
\doparttoc
\faketableofcontents
\part{Appendix} % Start the document part
\parttoc
\section*{Roadmap of the Appendix} Helping the reader navigate, we list a short summary of the topics that are covered in each section.  

\begingroup
\setlist{leftmargin=30pt,rightmargin=30pt,labelindent=20pt,topsep=0pt}
\setlist[enumerate]{wide=10pt, leftmargin=20pt, labelwidth=10pt, align=left}
\begin{itemize}
    \item 
In \cref{sec:ext-related-work}, we attempt to present an overview of key research directions, theoretical landmarks, and applications. Interested readers can find further details in the cited literature therein although an exhaustive survey is infeasible. 
\item In \cref{sec:opt-prelims-appendix}, we establish a series of lemmata concerning gradient descent for smooth functions:
\begin{inparaenum}[(i)]
    \item \cref{lemma:difference-in-gradient,lemma:projected_gradient_descent_ascent,lemma:stochastic-vs-deterministic-grad-mapping} are standard results regarding the iterates of projected gradient descent (with an exact or inexact and stochastic first-order oracle).
    \item We introduce $\mathcal{D}_{X}$ as the primary proxy of stationarity and derive a descent inequality for inexact gradient estimates based on this proxy (\cref{lemma:D-vs-grad-mapping,D-descent-lemma,lemma:reddi-three-point}).
    \item The section concludes with an analysis of the Lipschitz continuity of this proxy in min-max optimization. We also examine how one player's deviation affects the proxy of the other player.
\end{inparaenum}

\item \cref{sec:regularity-conditions}, we study relationship of Hidden Strong Convexity (HSC) and the equivalent conditions of the proximal Polyak-Łojasiewicz (pPL) and the Kurdyka-Łojasiewicz (KL).

\item In \cref{sec:convergence-min-max}, we present key results, including:
\begin{inparaenum}[(i)]
    \item The Lipschitz continuity of maximizers,
    \item Convergence rates for nested and alternating schemes under the general framework of hidden-convex hidden concave min-max optimization.
\end{inparaenum}
    
\item Finally, in \cref{sec:cmgs-convergence-appendix}, we extend our analysis to convex multi-agent reinforcement learning, combining the results from \cref{sec:convergence-min-max} with policy gradient estimators to compute parameter updates.
\end{itemize}
\endgroup

\section{Further Related Work}
\label{sec:ext-related-work}
\subsection{Hidden Convex Optimization}
Hidden convexity has emerged as an important structural property in nonconvex optimization, enabling global convergence results in settings where traditional convex analysis fails---a survey can be found here~\citep{xia2020survey}. \citet{fatkhullin2023stochastic} have decisively proven the convergence of first-order methods even for the nonsmooth and stochastic settings. Moreover, hidden convexity has been extensively studied across diverse applications, including policy optimization in reinforcement learning (RL) and optimal control \citep{hazan2019provably, zhang2020variational, ying2023policy}, generative models \citep{kobyzev2020normalizing}, supply chain and revenue management \citep{feng2018supply, chen2024efficient}, and neural network training \citep{wang2020hidden}. Even earlier, instances of an implicit convex structure have appeared in a number of works  in the past \citep{ben1996hidden,li2005hidden,wu2007peeling}. In a certain sense, hidden convexity falls into the general category of metric regularity conditions~\citep{karimi2016linear,li2018calculus,drusvyatskiy2019efficiency,drusvyatskiy2018error,liao2024error,rebjock2024fast,luo1993error,oikonomidis2025forward} that have been used to prove convergence of iterative gradient-based methods.

Particularly, hidden convexity ensures convergence in nonconvex-nonconcave games. Namely, \citet{vlatakis2021solving} introduced the notion of hidden-convex hidden-concave games and proved global convergence to a NE. Further extending these ideas, \citep{sakos2024exploiting} investigates the impact of hidden structure and show that such properties can be leveraged to enhance the stability of first-order methods. The generalized notion of hidden \textit{monotonicity} \citep{mladenovic2021generalized} guarantees global convergence for multi-player nonconvex games.

\subsection{Min-Max Optimization}
The literature of min-max optimization is long-standing and intimately connected to game theory~\citep{v1928theorie}. In recent years, the exploration of nonconvex-nonconcave min-max problems gained prominence in machine learning, particularly due to developments like generative adversarial networks (GANs)~\citep{goodfellow2014generative} and adversarial learning~\citep{madry2017towards}. Research has focused on defining appropriate solution concepts and developing methods to mitigate oscillatory behaviors in optimization algorithms. A min-max optimization problem is often formulated as a \textit{variational inequality} problem (VIP)~\citep{facchinei2003finite} and the literature has managed to guarantee provable convergence to solutions of the corresponding VIP only under certain assumptions. Namely, convergence to an approximate solution of the VIP point is guaranteed under \begin{inparaenum}[(i)] \item monotonicity, (or, a convex-concave objective) \item  a gradient domination, and \item other regularity conditions.
\end{inparaenum}

In particular, \textit{under monotonicity}, the proximal point methods \citep{martinet1970regularisation, rockafellar1976monotone} for VIP guarantee convergence. When the objective function is Lipschitz and strictly convex-concave, simple forward-backward schemes are known to converge. Moreover, when coupled with Polyak–Ruppert averaging \citep{ruppert1988efficient, polyak1992acceleration, nemirovski2009robust}, these methods achieve an \({O}(1/\epsilon^2)\) complexity without requiring strict convex-concavity \citep{bauschke2011convex}. If additionally the objective function has Lipschitz continuous gradients, the extragradient algorithm \citep{korpelevich1976extragradient} ensures trajectory convergence without strict monotonicity assumptions, while the time-averaged iterates converge at a rate of \({O}(1/\epsilon)\) \citep{nemirovski2004prox}. Furthermore, in strongly convex-concave settings, forward-backward methods compute an \(\epsilon\)-saddle point in ${O}(1/\epsilon)$ steps. If the operator is also Lipschitz continuous, classical results in operator theory establish that simple forward-backward methods are sufficient to achieve linear convergence \citep{facchinei2007finite, bauschke2011convex}.  

Under a single-sided gradient domination condition, convergence to a stationary point of the min-max objective has been proven in \citep{lin2020gradient,nouiehed2019solving,yang2022faster}. Under a two-sided gradient domination condition, \citep{daskalakis2020independent,yang2020global,zheng2023universal} have proven the convergence to a solution of the corresponding VIP (and equivalently to a min-max point). Further, under the Minty condition, \citep{mertikopoulos2018optimistic,liu2021first,diakonikolas2021efficient} show convergence to solution of the VIP even under nonconvexity. Evenmore, \citep{azizian24rate} shows convergence of gradient descent under various conditions of regularity and local optimization geometry.

\subsection{Convex Markov Decision Processes}
The study of convex reinforcement learning emerged as a natural extension of standard RL to handle more expressive, non-linear utility functions. \citet{hazan2019provably} introduced the problem of reward-free exploration of an MDP through state-occupancy entropy maximization. To tackle policy optimization, they proposed a provably efficient algorithm based on the Frank-Wolfe method. Further advances were made by \citet{zhang2020variational,zhang2021convergence}, who studied convex RL under the framework of RL with general utilities. Their key contribution was a policy gradient estimator for general utilities and the identification of a hidden convexity property within the convex RL objective, enabling statistically efficient policy optimization in the infinite-trials setting. More recently,  in \citep{zahavy2021reward,geist2022concave} convex RL is reinterpreted through a game-theoretic perspective. The former work views convex RL as a min-max game between a policy player and a cost player, while the latter positioned convex RL as a subclass of mean-field games.

A related strand of research focuses on the expressivity of scalar (Markovian) rewards. \citet{abel2021expressivity} demonstrated that scalar rewards cannot naturally encode all tasks, such as policy ordering or trajectory ranking. While convex RL extends the expressivity of scalar RL in these aspects, it still has inherent limitations. Specifically, infinite-trial convex RL excels at defining policy orderings but lacks full trajectory ordering capabilities, as it only considers the stationary state distribution. In contrast, the finite-trials convex RL formulation presented in this paper naturally captures trajectory orderings at the cost of reduced expressivity in policy ordering.

Another relevant direction concerns RL with trajectory feedback, where learning occurs through entire sequences rather than scalar rewards. Most prior works in this area assume an underlying scalar reward model, which merely delays feedback until the episode’s end \citep{efroni2021trajectory}. A notable exception is the once-per-episode feedback model studied by \citet{chatterji2021once}. Lastly, related research in multi-objective RL has explored the use of vectorial rewards to encode convex objectives. The works of \citet{cheung2019exploration, cheung2019regret} demonstrated that stationary policies are often suboptimal in such settings, necessitating non-stationary strategies. They provided principled procedures to optimize policies with sub-linear regret, complementing our analysis in infinite-horizon convex RL, where distinctions between finite and infinite trials diminish.

% Markov decision processes (MDPs) are the predominant framework for modeling sequential decision making problems, especially in infinite-horizon settings \cite{puterman2014markov}. The goal of a decision maker in an MDP is typically to maximize a γ- discounted sum of rewards earned throughout the sequen- tial decision process. In the infinite-horizon setting, re- cent research has exploited an alternative, but equivalent view of maximizing the expected reward under the agent’s stationary state-action occupancy measure (Zhang et al., 2020)—the probability of being in a given state and tak- ing a given action. This viewpoint reveals an optimiza- tion problem with a linear objective (maximize return) and linear constraints (valid occupancy measure); from this launchpad, research has generalized to convex objectives that incorporate, for example, the (neg)entropy of the oc- cupancy measure in order to maximize exploration of the MDP (Zahavy et al., 2021) or maximize robustness (Grand-
% Cle ́ment & Petrik, 2022).

\subsection{Markov Games and Multi-Agent Reinforcement Learning}
Markov, or \textit{stochastic}, games, were introduced by Lloyd S. Shapley~\citep{shapley1953stochastic}. Interestingly enough, their introduction coincides with that of single-agent of Markov decision processes \citep{schneider1957,BELLMAN1958228}. In an MG, agents interact with both the environment and each other. Each agent must balance immediate rewards against the potential future benefits of guiding the system to more advantageous states. Since their inception, MGs have served as the canonical model of MARL \citep{littman1994markov} which in turn encompasses numerous applications like autonomous vehicles \cite{shalev2016safe}, multi-agent robotics \citep{singh2022reinforcement,gronauer2022multi}, and general recreational game-playing \cite{mnih2013playing,silver2017mastering,berner2019dota}. 

In recent years, theoretical MG research has focused on developing computationally and statistically efficient algorithms. The literature has experienced a rapid development making an exhaustive review infeasible in the context of this small note. We outline results regarding \begin{inparaenum}[(i)] \item the computational complexity of equilibrium computation in MGs, \item ``no-regret'' learning approaches with convergence to the corresponding equilibrium notion, \item convergence to a NE when the game's structure permits it, and \item offline equilibrium learning from data.
\end{inparaenum} 

A number of works \citep{deng2023complexity,jin2022complexity,daskalakis2023complexity} concurrently proved that (coarse) correlated equilibria are intractable in infinite-horizon MG when the policies are required to be stationary and Markovian; computing them is $\mathsf{PPAD}$-complete. This is a pronounced contrast to normal-form games where they are computable in strongly polynomial time. In terms of ``no-regret'' approaches, the literature has been quite fruitful, as it considers finite-horizon games and policies that are stationary and both Markovian and non-Markovian. The solutions \citep{jin2021v,yang2022t,zhang2022policy,erez2023regret,cai2024near} are diverse and build upon the \textit{follow-the-regularized-leader} and \textit{online-mirror-descent} framework. Nash equilibrium convergence has mainly considered fully competitive two-player zero-sum MGs of infinite horizon, the Markovian counterpart of potential games, or their unification (adversarial team Markov games) and games with certain structural properties on the rewards and the dynamics. In two-player zero-sum games there have been results with a mostly stochastic-optimization approach 
\citep{daskalakis2020independent,wei2021last,sayin2021decentralized,cen2022faster,zeng2022regularized}. In addition, a notable work guarantees convergence to a NE in infinite-horizon games using only bandit feedback~\citep{cai2023uncoupled}. Then, for Markov potential games, global convergence of policy gradient methods has been presented in \citep{zhang2022global,leonardos2021global,ding2022independent}. Moreover, \citep{kalogiannis2022efficiently,kalogiannis2024learning} consider an MG of a ``team'' versus an ``adversary'' similar to \citep{von1997team} and guarantee provable convergence to NE. Lastly, \citep{kalogiannis23polymatrix,park2023multi} generalize zero-sum polymatrix game to their Markovian counterpart with extra assumptions on the dynamics, and prove the tractability of NEs. In terms of learning equilibria from samples, \citep{bai2020provable,cui2022offline} guaranteed sample-efficient learning of NE in zero-sum games and \citep{song2021can} went further to consider general-sum games. Finally, \citep{wang2023breaking,cui2023breaking} consider the sample complexity of learning an equilibrium beyond the tabular setting. 
% For instance,  \cite{bai2020provable} presented the first provably sample-efficient multi-agent reinforcement learning algorithm for two-player zero-sum Markov games.
% In the context of multi-player general-sum Markov games, \cite{liu2021provably} introduced an algorithm with sample complexity dependent on the size of the joint action space.

% Other studies have explored decentralized approaches to overcome the challenges posed by multiple agents \cite{jin2021v}, though some resulting policies may not be Markovian. Notably, \cite{daskalakis2023complexity} developed an algorithm capable of learning Markov coarse correlated equilibria while addressing the complexities introduced by multiple agents.

% Finally, additional research has examined settings with full-information feedback, aiming to establish convergence to various equilibrium concepts or achieve sublinear individual regret \cite{sayin2021decentralized,ding2022independent, zhang2022policy, cen2021fast, yang2022t, erez2023regret, ding2022independent}. The offline learning scenario, where a fixed dataset is provided without further environmental interaction, has also been a subject of recent studies~\cite{yang2021believe,pan2022plan,wang2023offline,shao2023counterfactual}.

\section{Optimization Preliminaries}
In this section we go over some rudimentary optimization lemmata as well as some novel exploration of the min-max optimization landscape for pPL functions. Namely, \cref{subsec:min-max-lemmata} compliments the study of \citep{nouiehed2019solving} for the constrained domain. To our knowledge, we offer the first such investigation of the constrained pPL landscape.
\label{sec:opt-prelims-appendix}
\subsection{Optimization Definitions \& Lemmata}

\begin{lemma}[Smoothness inequality]
    Assume that $f:\calX\to\R$ is an $\ell$-smooth function. Then, for $x,y \in \calX$ it holds that:
    \begin{equation}
        f(y) + \inprod{\nabla f(y)}{y-x} - \frac{\ell}{2}\norm{y-x}^2 \leq f(x) \leq f(y) + \inprod{\nabla f(y)}{y-x} + \frac{\ell}{2}\norm{y-x}^2.
    \end{equation}
    % \fk{get a reference}
\end{lemma}

\begin{definition}
    A function $f$ is said to be \emph{$\ell$-weakly convex} if $f(x) + \frac{\ell}{2}\norm{x}^2$ is convex.
\end{definition}

\subsubsection{Concerning the gradient mapping}
% \fk{define the gradient mapping}

\begin{lemma}    \label{lemma:projected_gradient_descent_ascent}
    Let $\calX \subseteq \R^d$ be a non-empty, closed, and convex set. Denote by $\proj_{\mathcal{X}}: \mathbb{R}^d \to \mathcal{X}$ the Euclidean projection operator onto $\mathcal{X}$. The following inequalities hold, 
    \begin{itemize}
        \item 
    for projected gradient descent:
    \begin{align}
        \inprod{v}{x - \proj_{\calX}\paren{x - \eta v}} \geq \frac{1}{\eta} \norm{x - \proj_{\calX}\paren{x - \eta v}}^2;
    \end{align}
    \item
    for projected gradient ascent:
    \begin{align}
        \inprod{v}{ \proj_{\calX}\paren{x + \eta v} - x} \geq \frac{1}{\eta} \norm{x - \proj_{\calX}\paren{x + \eta v}}^2.
    \end{align}
    \end{itemize}
    \label{lemma:descent-inner-prod-grad-mapping}
\end{lemma}

% \begin{lemma}
%     For projected gradient descent:
%     \begin{align}
%         \inprod{v}{x - \proj_{\calX}\paren{x - \eta v}} \geq \frac{1}{\eta} \norm{x - \proj_{\calX}\paren{x - \eta v}}^2
%     \end{align}
%     for projected ascent:
%     \begin{align}
%         \inprod{v}{ \proj_{\calX}\paren{x + \eta v} - x} \geq \frac{1}{\eta} \norm{x - \proj_{\calX}\paren{x + \eta v}}^2.
%     \end{align}
% \end{lemma}
\begin{proof}
    After writing the Euclidean projection of gradient descent with feedback $v$ as an optimization problem:
    \begin{align}
        \min_{x' \in\calX} \frac{1}{2} \norm{x' - \paren{x - \eta v} }^2,
    \end{align}
    we get the first-order optimality conditions,
    \begin{align}
        \inprod{ \eta v + { \proj_{\calX}\paren{x - \eta v} - x} }{z - \proj_{\calX}\paren{x - \eta v} } \geq 0,\quad \forall z \in \calX.
    \end{align}
    % \begin{align}
        % \inprod{ \eta v + { \proj_{\calX}\paren{x - \eta v} - x} }{x - \proj_{\calX}\paren{x - \eta v} } \geq 0
    % \end{align}
    Setting $z= x$ and re-arranging,
    \begin{align}
        \inprod{ v }{z - \proj_{\calX}\paren{x - \eta v} } \geq \frac{1}{\eta} \norm{ \proj_{\calX}\paren{x - \eta v} - x}^2.
    \end{align}
    Similarly for gradient ascent with feedback $v$ we write,
    \begin{align}
        \min_{x'\in\calX}\frac{1}{2}\norm{ x' - \paren{x + \eta v} }^2
    \end{align}
    The first-order optimality condition for the Euclidean projection, reads,
    \begin{align}
        \inprod{x'  - x - \eta v}{z - x'} \geq 0,~\quad \forall z \in \calX.
    \end{align}
    Plugging-in $x' = \proj_\calX\paren{x + \eta v}$,
    \begin{align}
        \inprod{\proj_{\calX}\paren{x + \eta v}  - x - \eta v}{x - \proj_{\calX}\paren{x + \eta v}} \geq 0.
    \end{align}
    Re-arranging we conclude,
    \begin{align}
        \inprod{ v}{ \proj_{\calX}\paren{x + \eta v} - x } \geq \frac{1}{\eta} \norm{\proj_{\calX}\paren{x + \eta v}  - x  }^2.
    \end{align}

    % \begin{align}
    %     \inprod{v}{x - \proj_{\calX}\paren{x + \eta v} }
    %     \leq 
    %     \frac{1}{\eta}\norm{x - \proj_{\calX}\paren{x + \eta v} }^2.
    % \end{align}
    % then
    % \begin{align}
    %     \inprod{-\eta v }{}
    % \end{align}
    % \begin{align}
    %     \inprod{ v - \eta \paren{ \proj_{\calX}\paren{x + \eta v} - x} }{z - \proj_{\calX}\paren{x + \eta v} } \leq 0,\quad \forall z \in \calX,
    % \end{align}
    % hence,
    % \begin{align}
    %     \inprod{ v  }{ \proj_{\calX}\paren{x + \eta v} - x } \geq \norm{x - \proj_{\calX}\paren{x + \eta v}}.
    % \end{align}
\end{proof}
% \fk{comment the lipschitzness gradient vector}
\begin{lemma}[{\citep[Lemma 2]{ghadimi2016mini}}]
        Let $v_1,v_2$ be vectors in $\R^d$ and $\calX\subseteq \R^d$ be a compact convex set and a scalar $\eta>0$. Also, let points $x_1^+, x_2^+ \in \calX$ such that:
        \begin{align}
            x_1^+ &:= \proj_{\calX} \paren{ x - \eta v_1}; \\
            x_2^+ &:= \proj_{\calX}\paren{x - \eta v_2}.
        \end{align}
        Then, it holds true that:
        \begin{align}
            \norm{x_1^+ - x_2^+}
            \leq {\eta}\norm{v_1 - v_2}.
        \end{align}
        \label{lemma:difference-in-gradient}.
    \end{lemma}

\begin{lemma}[Stoch. vs. Det. Grad. Mapping]
   Assume a stochastic gradient oracle $\hatgradfbx$ for a differentiable function $f$. For the stochastic gradient oracle, it holds that $\E [\hatgradfbx(x)]= \gradfbx(x)$, $\E[\norm{\hatgradfbx}^2]\leq \sigma_x^2$, and $\norm{\gradfbx(x) - \nabla f(x)} <\delta_x$ for all $x\in \calX$. Also, let $x,x^+,\hat{x}^+$ be points in $\calX$ such that:
   \begin{align}
       x^+ &:= \proj_{\calX}\paren{x - \eta \nabla_x f(x)};\\
       \hat{x}^+ &:= \proj_{\calX}\paren{x - \eta \hatgradfbx (x)}.
   \end{align}
   Then, the following inequalities hold:
   \begin{align}
       \norm{x - \hat{x}^+}^2 &\leq 2\norm{x - x^+}^2 + 4\eta^2\sigma_x^2 + 4\eta^2\delta_x^2;
       \\
       \norm{x - x^+}^2 &\leq 2\norm{x - \bar{x}^+}^2 + 4\eta^2\sigma_x^2 + 4\eta^2\delta_x^2.
   \end{align}
   \label{lemma:stochastic-vs-deterministic-grad-mapping}
\end{lemma}
   \begin{proof}
    We will prove the first inequality; the second follows from the same arguments.
       \begin{align}
           \E\norm{x - \hat{x}^+}^2 
                &\leq 2\E\norm{x - x^+}^2 + 2\E\norm{\hat{x}^+ - x^+ }^2\\
                &\leq 2\E\norm{x - x^+}^2 + 2\eta^2\E\norm{\hatgradfbx(x) - \nabla_x f(x) }^2\\
                &\leq 2\E\norm{x - x^+}^2 
                + 4\eta^2\E\norm{\gradfbx(x) - \hatgradfbx(x) }^2
                + 4\eta^2\E\norm{\gradfbx(x) - \nabla_x f(x) }^2
                \\
                &\leq 2\E\norm{x - x^+}^2 
                + 4\eta^2\sigma_x^2
                + 4\eta^2\delta_x^2
       \end{align}
       The first and second inequalities hold from the fact $|a+b|^2\leq 2|a|^2 + 2|b|^2$. The last inequality comes from the properties of the stochastic gradient oracle.
    %   \fk{once show that variance is bounded by the second moment}
   \end{proof}

\subsubsection{\texorpdfstring{$\calD_{\calX}$}{Dx}: An alternative proxy of stationarity}
In unconstrained optimization of differentiable functions, the conventional bounds of optimization algorithms directly guarantee the minimization of $\norm{\nabla f(x)}$, \textit{i.e.}, the norm of the gradient of $f$ at point $x$. In constrained optimization of differentiable functions, guarantees for algorithms like projected gradient descent ensure the minimization of the norm of the \textit{gradient mapping}, \textit{i.e.}, $\frac{1}{\eta}\norm{x- \proj_{\calX}\paren{x - \eta\nabla f(x) } }.$ It can be shown that the gradient mapping is indeed a good proxy of stationarity. For our work, we will consider the quantity $\calD_{\calX}$ to be defined shortly. This quantity is greater than the squared norm of the gradient mapping and turns out to be particularly favorable when handling the constrained optimization of proximal-PL functions. In what follows, we will define $\calD_\calX$ and discuss some of its useful properties.

\begin{definition}
    Let $\calX \subseteq \R^d$ be a compact convex set and $f:\calX \to \R$ be an $\ell$-smooth function. We define
    $\calD_{\calX}(x, \ell)$ to be:
    $$
    \calD_{\calX}(x,\alpha) = -2 \alpha \min_{y \in \R^d } \left\{ \langle \nabla f(x), y - x \rangle + \frac{\alpha}{2} \|y - x\|^2 + I_\calX(y) - I_\calX(x) \right\}.
    $$
    Where $I_\calX$ is the indicator function of the set $\calX$ with $I_\calX(x) = 0$ if $x\in\calX$ and $I_\calX(x)=+\infty$ otherwise.
\end{definition}
Remarkably, this expression can take a closed-form. As we will show, the minimizer of the display inside the brackets is exactly the point returned by one step of projected gradient descent on the argument $x$ with a stepsize equal to $\frac{1}{\alpha}$.
\begin{claim}[Closed form of $\calD_\calX$]
    Let $f$ be an $\ell$-smooth function defined on a compact convex set $\calX \subseteq \R^d$, a point $x\in\calX$, and a scalar $\alpha \geq \ell$. Then the following equation holds for $\calD_\calX(x,a)$:
    \begin{equation}
        \calD_\calX(x, \alpha) = 2\alpha \inprod{\nabla f(x)}{x - \proj_{\calX}\paren{ x- \frac{1}{\alpha} \nabla f(x) }} - \alpha^2 \norm{ x - \proj_{\calX}\paren{ x- \frac{1}{\alpha} \nabla f(x) }}^2.
    \end{equation}
\end{claim}
\begin{proof}
    By first-order optimality conditions, we see that the objective is equivalent to the objective of the Euclidean projection of the point $\paren{x - \frac{1}{\alpha} \nabla f(x)}\in\R^d$ to $\calX$. Then, we just plug the minimizer into the display.
\end{proof}

Next, we see that $\calD(\cdot,\alpha)$ is non-decreasing in the scalar $\alpha$.
\begin{lemma}[{\citep[Lemma 1]{karimi2016linear}}]
    Let a differentiable functions $f:\calX\to \R$ and positive scalars $0 < \alpha_1 \leq \alpha_2$. Then, the following inequality holds true:
    \begin{align}
        \calD_\calX(x,\alpha_1) \leq \calD_\calX(x,\alpha_2).
    \end{align}
    \label{D-non-decreasing-in-alpha}
\end{lemma}

The following Lemma demonstrates the relationship between $\calD_\calX$ and the norm of the gradient mapping.
\fk{make it $\alpha$}
\begin{lemma}[$\calD_\calX$ vs. Gradient Mapping Norm]
    Let $\mathcal{X} \subseteq \mathbb{R}^d$ be a non-empty, closed, and convex set. Consider a differentiable function $f : \calX \to \R$ that has an $\ell$-Lipschitz continuous gradient.
    Define the \emph{gradient mapping} at point $x \in \mathcal{X}$ with step size $1/\ell$ as,
    $$
     \ell^2(x - x^+)
    $$
    where,
    $$
        x^+ := \proj_{\mathcal{X}} \left( x - \frac{1}{\ell} \nabla f(x) \right).
    $$
    Then, the following inequality holds:
    $$
        \mathcal{D}_{\mathcal{X}}(x, \ell) \geq \ell^2 \|x^+ - x\|^2 .
    $$
    \label{lemma:D-vs-grad-mapping}
\end{lemma}
\begin{proof}
    By observing the definition of $\calD_\calX$, %
    % \begin{align}
    %     \calD_{\calX}(x,\ell) = -2 \ell \min_{y} \left\{ \langle \nabla f(x), y - x \rangle + \frac{\ell}{2} \|y - x\|^2 + I_\calX(y) - I_\calX(x) \right\}
    % \end{align}
    we observe that first-order optimality for its inner minimization problem read,
    \begin{align}
        \inprod{\nabla f(x) + \ell \paren{y- x}}{z - y}  &\geq 0,~\quad \forall z\in\calX.
    \end{align}
    On closer inspection, we recognize the first-order optimality of the Euclidean projection of one gradient descent step, \textit{i.e.},
    \begin{align}
        x^+ = \argmin_{y \in \calY } \left\{ \langle \nabla f(x), y - x \rangle + \frac{\ell}{2} \|y - x\|^2 \right\}.
    \end{align}
    Plugging-in $x^+$:
    \begin{align}
        \calD_{\calX}(x, \ell) &= - 2\ell \inprod{ \nabla f(x)}{ x^+ - x} - \ell^2 \|x^+ - x\|^2  \\
        &= 2\ell \inprod{\nabla f(x)}{x - x^+} - \ell^2 \|x^+ - x\|^2  \\
        &\geq \ell^2 \norm{x^+ - x }^2,
    \end{align}
    where the inequality follows from \cref{lemma:descent-inner-prod-grad-mapping}.
\end{proof}

% \fk{move a.3, a.4 here}

\begin{lemma}[{\citep[Lemma 6]{j2016proximal}}]
    Let $f:\calX\to\R$ be an $\ell$-smooth function and a point $x \in \calX \subseteq \R^d$. Also, define the vector $v\in\R^d$ and $y\in\calX$ to be $$ y := \proj_{\calX} \paren{x - \eta v}.$$ 
     Then, the following inequality is true:
    \begin{align}
        f(y) &\leq f(z) + \inprod{\nabla f(x) - v}{y-z} \\ &
        \quad+ \paren{\frac{\ell}{2} - \frac{1}{2\eta}}\norm{y - x}^2
        +\paren{\frac{\ell}{2} + \frac{1}{2\eta}}\norm{z - x}^2 - \frac{1}{2}\norm{y - z}^2.
    \end{align}
    \label{lemma:reddi-three-point}
\end{lemma}

\subsubsection{A descent lemma involving \texorpdfstring{$\calD_\calX$}{Dx}
} 

\begin{lemma}
Let $\calX \subseteq \R^d$ be a closed convex set, and let $f:\calX \to \R$ be an $\ell$-smooth function for some $\ell > 0.$  Suppose $\eta > 0$ with $\eta\leq \frac{1}{5\ell}$. 
For any $x \in \calX$ and any vector $v \in \R^d,$ define  
$
x^+ = \proj_{\calX}\paren{x - \eta v}.$
Then the following inequality holds:
\begin{equation}
f(x^+) \leq f(x) - \frac{\eta}{6}\calD_{\calX}(x,1/\eta)+\frac{\eta}{2}\norm{ \nabla f(x) - v}^2
\end{equation}
\label{D-descent-lemma}
\end{lemma}
\begin{proof}
First, we define $\bar{x}^+ := \proj_{\calX}\paren{x - \frac{1}{\alpha} \nabla f(x) }$.
\begin{itemize}
    \item Invoking $\ell$-smoothness of $f$ for $ x, \bar{x}\at{+}$ and assuming $\alpha>0$ with $\alpha \geq \ell$,
\begin{align}
    f(\bar{x}\at{+}) 
    &\leq f(x) + \inprod{\nabla f(x)}{\bar{x}\at{+} - x} + \frac{\ell}{2}\norm{x\at{+}-x}^2\\
    &\leq f(x) + \inprod{\nabla f(x)}{\bar{x}\at{+} - x} + \frac{\alpha}{2}\norm{x\at{+}-x}^2\\
    &= f(x) -\paren{ \inprod{\nabla f(x)}{ x - \bar{x}\at{+} } - \frac{\alpha}{2}\norm{x\at{+}-x}^2 }\\
    &= f(x) - \frac{1}{2\alpha}\calD_{\calX}(x,\alpha).
    \label{D-descent-ineq1}
\end{align}
    \item Invoking \cref{lemma:reddi-three-point} with $x = x$, $y = \bar{x}\at{+}$, $z = x$, $v = \nabla f(x)$
    \begin{align}
        f(\barx\at{+}) \leq 
        f(x) 
        + \paren{\frac{\ell}{2} - \frac{1}{\alpha}}\norm{\barx\at{+} - x}^2.
        \label{D-descent-ineq2}
    \end{align}

    \item Again, invoking \cref{lemma:reddi-three-point} but with $x = x$, $y = x\at{+}$, $z = \barx\at{+}$, $v $,
    \begin{align}
        f(x\at{+}) &\leq f(\barx\at{+}) + \inprod{\nabla f(x) - v}{x\at{+}-\barx\at{+}} \\ &
        \quad+ \paren{\frac{\ell}{2} - \frac{1}{2\eta}}\norm{x\at{+} - x}^2
        +\paren{\frac{\ell}{2} + \frac{1}{2\eta}}\norm{\barx\at{+} - x}^2 - \frac{1}{2\eta}\norm{x\at{+} - \barx\at{+}}^2. 
        \label{D-descent-ineq3}
    \end{align}

    Adding $1/3\times$\eqref{D-descent-ineq1} and $2/3\times$\eqref{D-descent-ineq2} and letting $1/\alpha = \eta \leq \frac{1}{\ell}$
    \begin{align}
        f(\barx\at{+}) \leq f(x) - \frac{1}{6\eta}\calD_{\calX}(x,1/\eta) + \paren{\frac{\ell}{3} - \frac{2}{3\eta}}\norm{\barx\at{+} - x}^2
    \end{align}

    Adding \eqref{D-descent-ineq3},
    \begin{align}
        f(x\at{+}) &\leq f(x)  - \frac{\eta}{6}\calD_{\calX}(x,1/\eta) + \paren{\frac{\ell}{3} - \frac{2}{3\eta}}\norm{\barx\at{+} - x}^2\\
        &\quad  + 
        \inprod{\nabla f(x) - v}{x\at{+}-\barx\at{+}} 
        \\ &
        \quad+ \paren{\frac{\ell}{2} - \frac{1}{2\eta}}\norm{x\at{+} - x}^2
        +\paren{\frac{\ell}{2} + \frac{1}{2\eta}}\norm{\barx\at{+} - x}^2 
        - \frac{1}{2\eta}\norm{x\at{+} - \barx\at{+}}^2 \\
        &\leq f(x) - \frac{\eta}{6}\calD_{\calX}(x,1/\eta) +  \paren{\frac{5\ell}{6} - \frac{1}{6\eta}}\norm{\barx\at{+} - x}^2 \\
        &\quad + \frac{\rho}{2}\norm{\nabla f(x) - v}^2 + \frac{1}{2\rho} \norm{x\at{+}-\barx\at{+}}^2\\
        &\quad + 
        \paren{\frac{\ell}{2} - \frac{1}{2\eta}}\norm{x\at{+} - x}^2
        - \frac{1}{2\eta}\norm{x\at{+} - \barx\at{+}}^2
        \label{D-descent-peter-paul}
        \\
        &= f(x) - \frac{\eta}{6}\calD_{\calX}(x,1/\eta) 
        + \paren{\frac{5\ell}{6} - \frac{1}{6\eta}}\norm{\barx\at{+} - x}^2 \\
        &\quad + \frac{\eta}{2}\norm{\nabla f(x) - v}^2 
        \\
        &\quad + 
        \paren{\frac{\ell}{2} - \frac{1}{2\eta}}\norm{x\at{+} - x}^2
         \label{D-descent-gather}
        \\
        &\leq f(x) - \frac{\eta}{6}\calD_{\calX}(x,1/\eta) + \frac{\eta}{2}\norm{\nabla f(x) - v}^2 
        \label{D-descent-drop-negatives}
        \\
        % &\leq f(x) - \frac{1}{6\alpha}\calD_{\calX}(x,\alpha)
        % + {\eta}\norm{\nabla f(x)     - \gradfbx(x)}^2
        % + {\eta}\norm{\gradfbx(x)     - \hatgradfbx(x)}^2 
        % \label{ppl-descent-triangle}
        % \\
        % &\quad + 
        % \paren{\frac{\ell}{2} - \frac{1}{2\eta}}\norm{x\at{+} - x}^2
    \end{align}
\end{itemize}
\begin{itemize}
    \item \eqref{D-descent-peter-paul} follows from the application of Young's inequality;
    \item \eqref{D-descent-drop-negatives} follows by dropping the non-positive terms; non-positivity follows from the choice of the step-size, $\eta \leq \frac{1}{5\ell}$.
% \begin{enumerate}
%     \item $0 \geq 5\ell - \frac{1}{\eta} $
%     \item $ \ell \leq 5\ell \leq \frac{1}{\eta}$ then $\ell - \frac{1}{\eta} \leq 0$
% \end{enumerate}
% \item \eqref{ppl-descent-triangle} follows from $\|a + b\|^2\leq 2\|a\|^2 + 2\|b\|^2.$ 
\end{itemize}

% Taking expectations:
% \begin{align}
%     \E f(x\at{+}) &\leq \E f(x) - \frac{1}{6\alpha}\E \calD_{\calX}(x,\alpha) + \eta \delta_x^2 + \eta \sigma_x^2 \\
%     &\leq \E f(x) - \frac{1}{6\alpha}\E \calD_{\calX}(x,\alpha) + \eta \delta_x^2 + \eta \sigma_x^2 ,
%     \label{ppl-descent}
% \end{align}
% where the last inequality follows from the choice of the step size. Rearranging, we get:
% \begin{align}
%     \sum_{t=0}^{T-1}\frac{1}{2} \calD_{\calX}(x,\ell)  &\leq
%     3\ell \sum_{t=0}^{T-1}
%      \paren{\E f(x) -  \E f(x\at{+}) + 3\eta \delta_x^2 + 3\eta \sigma_x^2
%      }
%      \\
%     &\leq 3\ell \paren{\E f(x\at{0}) -  f(x^\star)} + 3\eta \ell \sum_{t=0}^{T-1} \paren{ \delta_x^2 + \sigma_x^2}
% \end{align}

\end{proof}
\subsection{Min-Max Optimization Lemmata}
The following claims are novel to our knowledge. They justify the intuition that the constrained min-max pPL landscape should resemble its unconstrained PL counterpart. For example, the trivial bound $\norm{ \nabla_x f(x,y) - \nabla_x f(x,y')}^2\leq\ell^2 \norm{y - y'}$  of the unconstrained setting takes the form \cref{lemma:continuity-of-D-wrt-dual-var}.
% \begingroup
% \setlist{leftmargin=30pt,rightmargin=30pt,labelindent=20pt,topsep=0pt}
% \setlist[enumerate]{wide=10pt, leftmargin=20pt, labelwidth=10pt, align=left}
% \begin{itemize}
%     \item the trivial bound $\norm{ \nabla_x f(x,y) - \nabla_x f(x,y')}^2\leq\ell^2 \norm{y - y'}$  of the unconstrained setting takes the form \cref{lemma:continuity-of-D-wrt-dual-var},
%     \item similarly, $\norm{\nabla_x f(x,y) }^2 = \norm{\nabla_x f(x^\star(y),y') - \nabla_x f(x^\star(y),y)}^2\leq \ell^2\norm{ y - y'}^2$ takes the form \cref{lemma:D-y-yprime-bound}.
% \end{itemize}
% \endgroup
\label{subsec:min-max-lemmata}
  \begin{claim}
      Let $\calX $ and $\calY$ be convex and compact sets. Suppose that the function $ f: \calX \times \calY \to \R$ is $\ell$-smooth. For given points $y_1, y_2 \in \mathcal{Y}$, define $x_1$ and $x_1^+$ as follows:
        \begin{itemize}
            \item $x_1 := \argmin_{x\in\calX} f(x, y_1),$ and 
            \item $x_1^+ := \proj_\calX \paren{x_1 - \eta \nabla_x f(x_1, y_2) }$.
            % \item $x_1^+ := \argmin_{x\in\calX}\left\{\inprod{\nabla f(x_1, y_2)}{x_1^+ - x_1} + \frac{\eta}{2}\norm{x_1^+ - x_1}^2 \right\}$.
        \end{itemize} Then, then it holds true that:
        \begin{align}
             \frac{1}{\eta} \norm{x_1 - x_1^+} \leq \ell\norm{y_1 - y_2}  .
        \end{align}
        Symmetrically, for given points $x_1, x_2 \in \calX$, if $y_1 := \argmax_{y\in\calY} f(x_1, y_1) $, $y_1^+ := \proj_\calY\paren{y_1 + \eta \nabla_y f(x_2, y_1) },$
        $$ \frac{1}{\eta} \norm{y_1 - y_1^+ }\leq \ell \norm{x_1 - x_2}. $$
        \label{claim:continuity_of_stationarity}
    \end{claim}
    \begin{proof}
    By optimality of $x_1,x_1^+$ for their corresponding optimization problems:
     \begin{align}
        \inprod{\nabla f(x_1,y_1)}{x_1^+ - x_1} &\geq 0 \label{optimality-1}
        \\
        \inprod{\nabla f(x_1,y_2)}{x_1 - x_1^+} &\geq \frac{1}{\eta} \norm{x_1 - x_1^+}^2
        \label{optimality-2}
    \end{align}
    Adding-and-subtracting $\nabla_x f(x_1, y_2)$ to \cref{optimality-1},
    \begin{align}
        0 & \leq \inprod{\nabla f(x_1,y_1)}{x_1^+ - x_1} 
        \\
        &= \inprod{\nabla f(x_1,y_1) - \nabla f(x_1,y_2) + \nabla f(x_1,y_2) }{x_1^+ - x_1} \\
        &= \inprod{\nabla f(x_1,y_1) - \nabla f(x_1,y_2) }{x_1^+ - x_1} + \inprod{\nabla f(x_1,y_2)}{x_1^+ - x_1}.\label{added-and-subtracted-stationarity-cont-claim}
    \end{align}
    Hence, combining \cref{added-and-subtracted-stationarity-cont-claim} with \cref{optimality-2},
    \begin{align}
       \inprod{\nabla f(x_1,y_1) - \nabla f(x_1,y_2) }{x_1^+ - x_1} &\geq \inprod{\nabla f(x_1,y_2)}{x_1 - x_1^+} \\
       &\geq \frac{1}{\eta} \norm{x_1 - x_1^+}^2
    \end{align}
    By Cauchy-Schwarz,
    \begin{align}
        \norm{\nabla f(x_1,y_1) - \nabla f(x_1,y_2)}\norm{x_1^+ - x_1} \geq \frac{1}{\eta} \norm{x_1 - x_1^+}^2
    \end{align}
    Finally by Lispchitz continuity of $\nabla f(x,y)$,
    \begin{align}
        \ell\norm{y_1 - y_2}  \geq \frac{1}{\eta} \norm{x_1 - x_1^+}.
    \end{align}
    This completes the proof for a step of projected gradient descent. The arguments for the projected gradient ascent claim are identical.
    \end{proof}

\begin{lemma}
     Assume $f:\calX\times\calY$ an $\ell$-smooth function, and a scalar $a>0$,  two points $y,y'\in \calY$ and $x = \argmin_{z\in\calX}f(z,y)$. It is true that:
    \begin{align}
        2{\ell^2}\norm{y - y'}^2 \geq \calD_\calX(x,\alpha; y').
    \end{align}
    \label{lemma:D-y-yprime-bound}
\end{lemma}
\begin{proof}
    Let $x^+ := \argmax_{\bar x \in \calX }\left\{ \inprod{\nabla f(x,y')}{x - \bar x} - \frac{\alpha}{2} \norm{x - \bar x}^2 \right\} = \proj_\calX \paren{x - \frac{1}{\alpha}\nabla f(x,y') }.$
    \begin{align}
        \frac{1}{2\alpha}\calD_\calX(x,y',\alpha) &= \inprod{\nabla f(x,y')}{x- x^+} - \frac{\alpha}{2}\norm{x - x^+}^2 \\
        &= \inprod{\nabla f(x,y')  -\nabla f(x,y)}{x- x^+} - \frac{\alpha}{2}\norm{x - x^+}^2  + \inprod{\nabla f(x,y)}{x-x^+}
        \label{D-bound-ineq1}
    \end{align}
    We have assumed that $x = \argmin_{\bar x \in \calX}f(\bar x,y)$,
    therefore,
    \begin{align}
        \inprod{\nabla f(x,y)}{z-x}\geq 0\quad \forall z \in \calX.
    \end{align}
    As such, when $z = x^+$,
    \begin{align}
        \inprod{\nabla f(x,y)}{x-x^+}\leq 0.
    \end{align}
    Hence, since $\frac{-1}{\eta}\norm{x - x^+}^2, \inprod{\nabla f(x,y)}{x-x^+} \leq 0$ plugging in \eqref{D-bound-ineq1},
    \begin{align}
        \inprod{\nabla f(x,y')  -\nabla f(x,y)}{x- x^+} \geq \frac{1}{2\alpha} \calD_\calX (x,\alpha;y').
    \end{align}
    Using Cauchy-Schwarz and $\ell$-Lipschitz continuity of the gradient,
    \begin{align}
        \ell \norm{y - y'}\norm{x - x^+} \geq \frac{1}{2\alpha} \calD_\calX(x,\alpha;y').
    \end{align}
    Finally, using \cref{claim:continuity_of_stationarity} with $\eta = \frac{1}{\alpha}$,
    \begin{align}
        \frac{\ell^2}{ \alpha} \norm{y - y'}^2 \geq \frac{1}{2\alpha}\calD_\calX (x,\alpha;y').
    \end{align}
\end{proof}

% I want to show that $|\calD_{\calY}(y\at{t},\alpha;x\at{+}) - \calD_{\calY}(y\at{t},\alpha;x) |\leq C\norm{x\at{+} - x}^2 $
\begin{lemma}
    Let $f:\calX\times\calY$ be an $\ell$-smooth function, $a>0$,  two points $y,y'\in \calY$, and a point $x\in \calX$. Then, the following inequality holds:
    $$ |\calD_\calX(x,a;y) - \calD_\calX(x,a;y')| \leq 3\ell^2 \norm{y-y'}^2.$$
    \label{lemma:continuity-of-D-wrt-dual-var}
\end{lemma}
\begin{proof}
 We define $\bar{x},\bar{x}' \in \calX$ to be:
    \begin{align}
        \bar{x}&:= \proj_{\calX}\paren{x - \frac{1}{\alpha}\nabla_x f(x,y) };\\
        \bar{x}'&:= \proj_{\calX}\paren{x - \frac{1}{\alpha}\nabla_x f(x,y') }.
    \end{align}

    By the definition of $\calD_{\calX}(x,\alpha;y')$ we write:
    \begin{align}
        \frac{1}{2\alpha}\calD_\calX(x,\alpha;y) &= \inprod{\nabla f(x,y)}{x - \bar{x}} - \frac{\alpha}{2}\norm{x-\bar{x}}^2
        \\
        \frac{1}{2\alpha}\calD_\calX(x,\alpha;y') &= \inprod{\nabla f(x,y')}{x - \bar{x}'} - \frac{\alpha}{2}\norm{x-\bar{x}'}^2.
    \end{align}

    Considering the difference $\calD_\calX(x,\alpha;y) - \calD_\calX(x,\alpha;y')$ we see that:
    \begin{align}
        \frac{1}{2\alpha}| \calD_\calX(x,\alpha;y) - \calD_\calX(x,\alpha;y') | 
        &= 
        \left|
            \inprod{\nabla_x f(x,y) - \nabla_x f(x,y')}{\bar{x}' - \bar{x}} -\frac{\alpha}{2} \paren{\norm{x-\bar{x}}^2 - \norm{x- \bar{x}' }^2 }     
        \right|
        \\
        &\leq 
        \left|
            \inprod{\nabla_x f(x,y) - \nabla_x f(x,y')}{\bar{x}' - \bar{x}}
        \right|
        + \frac{\alpha}{2}
        \left|
        \paren{\norm{x-\bar{x}}^2 - \norm{x- \bar{x}' }^2 }     
        \right|\\
        &\leq 
        \left|
            \inprod{\nabla_x f(x,y) - \nabla_x f(x,y')}{\bar{x}' - \bar{x}}
        \right|
        + \frac{\alpha}{2}
         \norm{\bar{x} - \bar{x}' }^2
     \\
        &\leq 
            \norm{\nabla_x f(x,y) - \nabla_x f(x,y')}\norm{\bar{x}' - \bar{x}}
                + \frac{\alpha}{2}
                 \norm{\bar{x} - \bar{x}' }^2
     \\
        &\leq 
            \frac{1}{\alpha}\norm{\nabla_x f(x,y) - \nabla_x f(x,y')}^2
                + \frac{1}{2\alpha}
                 \norm{\nabla_x f(x,y) - \nabla_x f(x,y') }^2
     \\
        &\leq 
            \frac{3\ell^2}{2\alpha}\norm{y-y'}^2.
    \end{align}
    We note that:
    \begin{itemize}
        \item The first inequality follows from the triangle inequality.
        \item In the second, we apply the reverse triangle inequality. 
        \item The third inequality uses the Cauchy-Schwarz inequality.
        \item The penultimate uses \cref{lemma:difference-in-gradient} and the final $\ell$-Lipschitz continuity of the gradient.
    \end{itemize}
\end{proof}

\begin{claim}
    Assume a function $f:\calX\times\calY$ satisfying the pPL condition with modulus $\mu$ over $f(x,\cdot)$ for any $x\in\calX$. Also, assume that $\Phi(x) := \max_{y\in\calY}f(x,y)$. Additionally, define: $\calD_\calX^\Phi (x,\alpha)= - 2 \alpha \min_{z\in\calX}
            \left\{ 
                \inprod{ \nabla \Phi(x) }{z -x} + \frac{\alpha}{2}\norm{z - x}^2 
            \right\}.$
    % \begin{itemize}
    %     \item $\calD_\calX (x,\alpha;y)= - 2 \alpha \min_{z\in\calX}
    %         \left\{ 
    %             \inprod{ \nabla f(x,y) }{z -x} + \frac{\alpha}{2}\norm{z - x}^2 
    %         \right\}$ and
    %     \item $\calD_\calX^\Phi (x,\alpha)= - 2 \alpha \min_{z\in\calX}
    %         \left\{ 
    %             \inprod{ \nabla \Phi(x) }{z -x} + \frac{\alpha}{2}\norm{z - x}^2 
    %         \right\}$
    %     \item $\calD_\calY (y,\alpha;x)= - 2 \alpha \min_{z\in\calY}
    %         \left\{ 
    %             \inprod{ - \nabla f(x,y)}{z - y} + \frac{\alpha}{2}\norm{z - y}^2 
    %         \right\}$.
    % \end{itemize}
    Then, it holds true that,
    \begin{align}
        \left|\calD_{\calX}^{\Phi}(x, \alpha_1) - \calD_{\calX}(x, \alpha_2; y) \right|\leq 3\ell \kappa^2 \calD_\calY(y,\alpha_2;x)
    \end{align}
    where $\kappa:= \frac{\ell}{\mu}$.
    \label{continuity-of-D-squared-y}
\end{claim}
\begin{proof}
    We define $\bar{x},\bar{x}'$ as,
    \begin{align}
        \bar{x} &:= \proj_{\calX}\paren{x - \frac{1}{\alpha_1}\nabla \Phi(x) };\\
        \bar{x}' &:= \proj_{\calX}\paren{x - \frac{1}{\alpha_1}\nabla f(x,y) }.
    \end{align}
    Then, we see that:
    \begin{align}
        \frac{1}{2\alpha_1}\left|\calD_{\calX}^{\Phi}(x, \alpha_1) - \calD_{\calX}(x, \alpha_2; y) \right|
        &\leq \left|
            \paren{
            \inprod{\nabla \Phi(x)}{x - \bar{x} } 
            - \frac{\alpha_1}{2}\norm{x- \bar{x}}^2
            }
            -
            \paren{
                \inprod{\nabla f(x,y)}{x - \bar{x}'} - \frac{\alpha_1}{2}\norm{x- \bar{x}'}^2
            }
        \right|\\
        &\leq \left|
            \inprod{\nabla \Phi(x) - \nabla f(x,y)}{\bar{x}' - \bar{x}}
        \right|
        +
        \frac{\alpha_1}{2}
        \left|
            \norm{x- \bar{x}}^2 - \norm{x- \bar{x}'}^2 
        \right|\\
        &\leq 
            \norm{\nabla \Phi(x) - \nabla f(x,y)}\norm{\bar{x}' - \bar{x}}
        +
        \frac{\alpha_1}{2}
            \norm{\bar{x}- \bar{x}'}^2\\
        &\leq 
            \frac{3}{2\alpha_1}\norm{\nabla \Phi(x) - \nabla f(x,y)}^2 
            \\
        &\leq \frac{3\ell}{2\alpha_1} \norm{y - y^\star(x)}^2\\
        &\leq \frac{3\ell \kappa^2}{2\alpha_1} \calD_{\calY}(y,\alpha_2;x).
    \end{align}
    Where, the first inequality is due to the definition. The second and the third, follow from the triangle inequality and the reverse triangle inequality respectively. The fourth follows from the fact that the projection operator is contractive ($1$-Lipschitz). The last one is due to the quadratic growth condition and the pPL inequality. 
\end{proof}

\section{Hidden Convexity, PL, KL, and EB}
\label{sec:regularity-conditions}
In this section, we offer a brief exposition to the notions of hidden convexity and other regularity conditions it is related to. We refert the interested reader to \citep{karimi2016linear,li2018calculus,drusvyatskiy2019efficiency,drusvyatskiy2018error,liao2024error,rebjock2024fast} and references therein. We commence by defining hidden convexity in the manner that \cite{fatkhullin2023stochastic} do it.

\begin{definition}[Hidden Convex Problem]
    Consider the optimization problem
    $$
        \min_{x \in \mathcal{X}} f(x) := H(c(x)).
    $$
    This problem is called {hidden convex} with moduli $\Lcinv > 0$ and $\mu_H \geq 0$ (or the function $f$ is hidden convex on $\mathcal{X}$) if the following conditions are satisfied:
    \begin{enumerate}[label=(HC.\arabic*)]
        \item {Convexity of $H$ and Domain:}
        \begin{itemize}
            \item The domain $\calU = c(\mathcal{X})$ is convex.
            \item The function $H : \calU \to \R$ satisfies for all $u, v \in 
            \calU$ and $0\leq\alpha\leq 1$,
            \begin{align}
                H((1-\alpha)u + \alpha v) 
                \leq (1-\alpha)H(u) + \alpha H(v) - \frac{(1-\alpha)\alpha \mu_H}{2} \|u - v\|^2.
            \end{align}
            \item The problem admits a solution $u^\star \in \calU$.
        \end{itemize}
        
        \item {Invertibility and Lipschitz Continuity of $c^{-1}$:}
        \begin{itemize}
            \item The mapping $c : \calX \to \calU$ is invertible.
            \item There exists $\mu_c > 0$ such that for all $x, y \in \mathcal{X}$,
            $$
                \norm{x - y} \leq \mu_c \norm{c(x) - c(y)}.
            $$
        \end{itemize}
    \end{enumerate}
    Furthermore, if $\mu_H > 0$, the problem is referred to as \emph{$(\mu_c, \mu_H)$-hidden strongly convex}.
\end{definition}

\begin{assumption}[Lipschitzness and smoothnes of $c$]
    Let $ c: \mathcal{X} \times \mathcal{Y} \to \mathcal{U} $ be a mapping such that for all $ (x, y), (x', y') \in \mathcal{X} \times \mathcal{Y} $ and all $ u, u' \in \mathcal{U} $, the following conditions hold:
    \begin{align}
        \| c(x, y) - c(x', y') \| &\leq L_c \| (x, y) - (x', y') \|; \label{eq:lipschitz_c} \\
        \| c^{-1}(u) - c^{-1}(u') \| &\leq L_{c^{-1}} \| u - u' \|; \label{eq:lipschitz_c_inv} \\
    \end{align}
    and also,
    \begin{align}
        \| \nabla c(x, y) - \nabla c(x', y') \| &\leq \ell_c \| (x, y) - (x', y') \|; \label{eq:grad_lipschitz_c} \\
        \| \nabla c^{-1}(u) - \nabla c^{-1}(u') \| &\leq \ell_{c^{-1}} \| u - u' \|. \label{eq:grad_lipschitz_c_inv}
    \end{align}
\end{assumption}

\begin{fact}
    A hidden strongly convex function has a unique maximizer.
\end{fact}

% \begin{proposition}
%     Assume a point $\bar{x}\in\calX$ and $\epsilon>0$. If $\bar{x}$ a first-order stationary point of $f$, then $\bar{x}$ is an $\epsilon/\Lc$-approximate global minimum of $f$. 
%     % \begin{align}
%     % \end{align}
% \end{proposition}
% \begin{proof}
%     % Similar to \citep[Proposition 1]{fatkhullin2023stochastic}, we apply the chain rule on function $f$, to get,
%     % \begin{equation}
%     %   \nabla_x f(\bar{x}) =  \nabla_x c(\bar{x}) \nabla_u H(\bar{u}).
%     % \end{equation}
%     % By stationarity of $\epsilon$, we have,
%     % \begin{align}
%     %     \nabla_x f(\bar{x}) \in N(\bar{x}) +  B(0,\epsilon) & \iff   \inprod{ -\nabla f(\bar{x})}{x - \bar{x}} \leq \epsilon \norm{x - \bar{x}},\forall x \in \calX.
%     % \end{align}
%     % But, since we have assumed $c$ to be $\Lc$-Lipschitz continuous
%     % \begin{align}
%     %     \norm{\nabla_x c(\bar{x})}\leq \Lc.
%     % \end{align}
%     % Putting these pieces together, 
%     % \begin{equation}
%     %     \partial_u \paren{ H + I_\calU}(\bar{u}) \in B(0,\epsilon/\Lc).
%     % \end{equation}
%     % The convexity of $H$ implies the conclusion.
% \end{proof}

% \begin{lemma}[Gradient Dominance through HC]
%     \fk{fix}
%     \begin{align}
%         f(x) - f(x^\star) \leq \Lc \max_{x\in\calX} \inprod{-\nabla f(\bar{x})}{x - \bar{x}}
%     \end{align}
% \end{lemma}

\subsection{Equivalence between the three conditions}
Throught this subsection, we will denote $\calX^\star$ to be $\calX^\star = \argmin_{x\in\calX} f(x)$, $x_{p}^\star$ w.r.t. to a point $x\in\calX$ is defined as some element of $\argmin_{x'\in\calX^\star}\norm{x - x'}$. Moreover, we define $F(x):= f(x) + I_{\calX}(x)$ with $I_{\calX}(x) = 0$, if $x\in\calX$, and $I_{\calX}(x) = +\infty$, else. 
% \fk{definition of $x^\star_p, \calX^\star, F$}
\begin{definition}[pPL, KL, pEB]
Let $f:\calX\to \R$ be an $L$-Lipschitz continuous function with $\ell$-Lipschitz continuous gradient. Then,
    \begin{itemize}
        \item \emph{Proximal Error-Bound (pEB):} $f$ is said to satisfy the proximal Error-Bound if $\exists c>0$ s.t.
            \begin{equation}
                \norm{x - x^\star_p} \leq c \norm{x - \proj_{\calX}\paren{x - \frac{1}{\ell} \nabla f(x) } },\quad \forall x \in \calX
            \end{equation}
        \item \emph{Proximal Polyak-Łojasiewicz (pPL):} $f$ is said to satisfy the proximal Polyak-Łojasiewicz condition if $\exists\mu>0$ s.t.
            \begin{equation}
                \frac{1}{2}\calD_{\calX}(x, \ell) \geq \mu \paren{ f(x) - f(x^\star)}
            \end{equation}
        \item \emph{Kurdyaka-Łojasiewicz (KL):} f is said to satisfy if $\exists \bar{\mu}$ s.t.
        \begin{equation}
            \min_{s_x \in \partial(f + I_\calX)(x) } \norm{s_x}^2 \geq 2\bar{\mu} \paren{ f(x) - f(x^\star)},\quad \forall x \in \calX.
        \end{equation}
    \end{itemize}
\end{definition}
\begin{lemma}[{\citep[Appendix G]{karimi2016linear}}]
    \begin{align}
        \text{(pPL)}\equiv \text{(pEB)} \equiv \text{(KL)}.
    \end{align}
\end{lemma}

\begin{remark}
    As we will see, (KL) with a modulus $\bar{\mu}$ implies (pEB) with a modulus $1 + \frac{2\ell}{\bar{\mu}}$. In turn, (pEB) with a modulus $c$, implies (pPL) with a modulus $\frac{\ell}{1+4c^2}$.
\end{remark}

Following, we repeat the calculations of \citep{karimi2016linear} by explicitly tracking the dependence between the constants. We will denote $\calX^\star:= \argmin_{x\in\calX}f(x)$ and $N_{\calX}(x)$ to be the normal cone of $\calX$ at $x$.

\paragraph{KL$\to$pEB}
First, let $F$ be $F(x) := f(x) + I_\calX(x)$. By the $\ell$-smoothness of $f$ and the convexity of $I_\calX$, we observe that $F$ is weakly-convex. By \citep[Theorem 13]{bolte2010characterizations}, for any $x \in \text{dom} F$ there exists a subgradient curve $\chi_x : [0, \infty) \to \text{dom} F$ that satisfies the following:
\begin{align}
\dot{\chi}_x(t) &\in -\partial F(\chi_x(t)),
\\
\chi_x(0) &= x,
\\
\frac{d}{dt} F(\chi_x(t)) &= -\|\dot{\chi}_x(t)\|^2, \label{kl-peb-1}
\end{align}
where $t \mapsto F(\chi_x(t))$ is a non-increasing and Lipschitz continuous function for $t \in [\eta, \infty]$ for any $\eta > 0$. Further, we define the function $r(t) = \sqrt{F(\chi_x(t)) - F^\star}$. By differentiating $r(t)$, we observe that:
\begin{align}
\frac{dr(t)}{dt} &= \frac{\dot{F}(\chi_x(t))}{2\sqrt{F(\chi_x(t)) - F^\star}} \\
&= -\frac{\|\dot{\chi}_x(t)\|^2}{2\sqrt{F(\chi_x(t)) - F^\star}} \\
&\leq -\sqrt{\bar{\mu}/2}\|\dot{\chi}_x(t)\|, \label{kl-peb-2}
\end{align}
\begin{itemize}
    \item In the second line, we plug in the definition of the subgradient curve.
    \item In the third line, we use the KL inequality and the fact that $\dot{\chi}_x(t) \in -\partial F(\chi_x(t))$.
\end{itemize}
Taking the integral, we write:
\begin{align}
r(T) - r(0) &= \int_0^T \frac{dr(t)}{dt} dt \\
&\leq -\sqrt{\bar{\mu}/2} \int_0^T \|\dot{\chi}_x(t)\| dt \label{kl-peb-3-intermed}\\
&\leq -\sqrt{\bar{\mu}/2} \dist(\chi_x(T), \chi_x(0)), \label{kl-peb-3}
\end{align}
where:
\begin{itemize}
    \item \eqref{kl-peb-3-intermed} uses the bound on $r(t)$
    \item \eqref{kl-peb-3} uses the fact that any curve connecting two points has a length greater than their Euclidean distance.
\end{itemize}
Now, we will show that $\lim_{T\to \infty} r(T) = 0$. By \eqref{kl-peb-3}, we get
\begin{align}
r(T) - r(0) &= \int_0^T \frac{dr(t)}{dt} dt \\
&= -\int_0^T \frac{\|\dot{\chi}_x(t)\|^2}{2\sqrt{F(\chi_x(t)) - F^\star}} dt \\
&\leq -\frac{\bar{\mu}}{2} \int_0^T \sqrt{F(\chi_x(t)) - F^\star} dt \\
&\leq -\frac{\bar{\mu}}{2} T r(T), \label{kl-peb-4}
\end{align}
where the first inequality follows from the KL condition, and second inequality uses the fact that $F(\chi_x(t))$ is non-increasing in $t$. Hence, we get a bound on $r(T)$,
\begin{align}    
0 \leq r(T) \leq \frac{2r(0)}{2 + \bar{\mu} T},
\end{align}
By taking the limit of $T \to \infty$, we get $\lim_{T\to \infty} r(T) \to 0$. Now we can focus on the term $r(0) = \sqrt{F(x) - F^\star}$. Proceeding, by \eqref{kl-peb-3} we write,
\begin{equation}
    \sqrt{F(x) - F^\star} \geq \sqrt{\bar{\mu}/2} \dist(x, \calX^\star ), \label{kl-peb-5}
\end{equation}
From \eqref{kl-peb-5} and the KL condition, we see that:
\begin{align}
    \dist(0, \partial F(x)) \geq \bar{\mu} \dist(x, \calX^\star). \label{kl-peb-6}
\end{align}
Now, let define $x^+ = \text{prox}_{\frac{1}{\ell} g} \left(x - \frac{1}{\ell} \nabla f(x)\right)$. By the optimality of $x^+$, we have $-\nabla f(x^+) - \ell(x^+ - x) \in N_{\calX}(x^+)$. As such,
$$
\nabla f(x^+) - \nabla f(x) - \ell(x^+ - x) \in \partial I_{\calX}(x^+).
$$
Letting the particular subgradient of $I_{\calX}(x^+)$ that achieves this be $\xi$, we write,
\begin{align}
\dist(0, \partial F(x^+)) &= \|0 - \xi\| \\
&= \|\nabla f(x^+) - \nabla f(x) - \ell(x^+ - x)\| \\
&\leq \ell\|x^+ - x\| + \|\nabla f(x^+) - \nabla f(x)\| \\
&\leq 2\ell\|x^+ - x\|, \label{kl-peb-7}
\end{align}
where we used the triangle inequality and the $\ell$-Lispchitz continuity of $\nabla f$. Finally, we can derive the proximal-EB condition by
\begin{align}
\dist(x, \calX^\star) &\leq \|x - x^+\| + \text{dist}(x^+, \calX^\star) \\
&\leq \|x - x^+\| + \frac{1}{\bar{\mu}} \dist(0, \partial F(x^+)) \\
&\leq \|x - x^+\| + \frac{2\ell}{\bar{\mu}}\|x^+ - x\| \\
&= \paren{1 + \frac{2\ell}{\bar{\mu}} }\|x - x^+\|. \label{kl-peb-8}
\end{align}
We note that:
\begin{itemize}
    \item the first inequality follows by the triangle inequality;
    \item the second inequality uses the \eqref{kl-peb-6};
    \item the third inequality uses \eqref{kl-peb-7}.
\end{itemize}
Re-iterating:
\begin{align}
    \min_{s_x \in \partial(f + I_\calX)(x) } \norm{s_x}^2 \geq 2\bar{\mu} \paren{ f(x) - f(x^\star)},\quad \forall x \in \calX \quad \implies \quad \dist(x,\calX^\star) \leq \paren{1 + \frac{2\ell}{\bar{\mu}}}\norm{ x - x^+}.
\end{align}

\paragraph{pEB$\to$pPL}
Before proceeding, we define the forward-backward envelope, $F_{\frac{1}{\ell}}$, of $F$~\citep[Definition 2.1]{stella2017forward}, as:
\begin{align}
    F_{\frac{1}{\ell}}:= \min_{y\in \R^d}\left\{f(x) + \inprod{\nabla f(x)}{y-x} + \frac{\ell}{2}\norm{y-x}^2 + I_\calX(y) - I_\calX(x)\right\}.
\end{align}
Additionally, we observe that: $\calD_{\calX}(x,\ell) = 2\ell\paren{F(x)  - F_{\frac{1}{\ell}}(x) }.$
Moving forward we note that by assumption, it holds that:
\begin{align}
    \dist(x,\calX^\star) \leq c  \norm{x - \proj_\calX\paren{x - \frac{1}{\ell}\nabla f(x) }}.
\end{align}
From the latter we write,
\begin{align}
    F(x) - F^\star &= F(x) - F_{\frac{1}{\ell}}(x) + F_{\frac{1}{\ell}}(x)  - F^\star \\
    &\leq F(x) - F_{\frac{1}{\ell}}(x) + 2\ell\norm{x^\star - x}^2 \\
    &\leq F(x) - F_{\frac{1}{\ell}}(x) + 2 \ell c^2  \norm{x - \proj_\calX\paren{x - \frac{1}{\ell}\nabla f(x) }}^2 \\
    &\leq \paren{1 + 4c^2 }\paren{ F(x) - F_{\frac{1}{\ell}}(x)}\\
    &= \paren{1 + 4c^2 }\frac{1}{2\ell}\calD_{\calX}(x,\ell).
\end{align}

Where the second inequality follows from the pEB condition and the two last steps follow from \citep[Proposition 2.2(i)]{stella2017forward}:
$ \frac{\ell}{2}\norm{x - \proj_\calX\paren{x - \frac{1}{\ell}\nabla f(x) }}^2 \leq F(x) - F_{\frac{1}{\ell}}(x)$ and $F(x) - F_{\frac{1}{\ell}}(x) = \frac{1}{2\ell}\calD_{\calX}(x,\ell)$. In short,
\begin{align}
    \frac{\ell}{1 + 4c^2}\paren{F(x) - F^\star} \leq \frac{1}{2}\calD_\calX(x,\ell).
\end{align}

Concluding, we see that KL with modulus $\bar{\mu}$ translates to pEB with constant $c=\paren{1 + \frac{2\ell}{\bar{\mu}}}$. In turn, pEB with constant $c$ translates to pPL with modulus $\frac{\ell}{1+4c^2}$. This means that KL with $\bar{\mu}$ is equivalent to pPL with modulus $\mu$, 
\begin{equation}
    \mu := \frac{\ell}{1 + 4 \paren{1 + \frac{2\ell}{\bar{\mu}} }^2} \geq \frac{\ell \bar{\mu}^2 }{ 9\bar{\mu}^2 +  32\ell^2 }.
\end{equation}

\subsection{pPL from HSC}
\begin{proposition}
    Let $f:\calX\to\R$ be an $\ell$-smooth function that is $(\mu_c,\mu_H)$-hidden strongly convex with $\mu_H>0$. Then, $f$ satisfies the proximal Polyak-\L ojasiewciz condition with a modulus $\mu$,
    $$\mu:= \frac{\ell}{1 + 4\paren{1 + \frac{2\ell}{2\mu_c^2\mu_H } }^2 }$$
\end{proposition}
\begin{proof}
By \citep[Proposition 2(ii)]{fatkhullin2023stochastic}, $(\mu_c,\mu_H)$-hidden strong convexity implies,
\begin{align}
    \min_{s_x\in\partial(f + I_\calX)(x) }\norm{s_x}^2 \geq  {2\mu_c^2\mu_H} \paren{f(x) - f(x^\star)}.
\end{align}
While, \citep[{Appendix G}]{karimi2016linear} and our precise quantification of the constants,
\begin{align}
    \mu := \frac{\ell}{1 + 4\paren{1 + \frac{2\ell}{2\mu_c^2\mu_H } }^2 }.
\end{align}
\end{proof}

\subsection{QG directly from HSC, KL, pPL}
\begin{proposition}[QG from HSC]
    \label{prop:QG_from_HSC}
    Let $ f : \calX \to \R$ be an $\ell$-smooth function that is $(\mu_c, \mu_H)$-hidden strongly convex. Then, for all $x \in \calX$,
    \[
        f(x) - f(x^\star) \geq \frac{\muqg}{2} \|x - x^\star\|^2,
    \]
    with $\mu_{\qg}:= \mu_c^2\mu_H $ where $x^\star$ is the unique minimizer of $f$ over $\mathcal{X}$.
\end{proposition}

    \begin{proof}
    We start by leveraging the strong convexity property of the function $H$. By the definition of strong convexity with modulus $\mu_H$, for any $u, v \in\calU$:
    % and its minimizer $u^\star$, the following inequality holds:
    \begin{align}
        H(u) - H(u^\star) - \inprod{\nabla_u H(u^\star)}{u - u^\star} &\geq \frac{\mu_H}{2} \norm{u - u^\star}^2. \label{eq:strong_convexity}
    \end{align}
    
    When $u^\star$ is a minimizer of $H$, the gradient at $u^\star$ satisfies:
    $$
        \inprod{\nabla_u H(u)}{u - u^\star} \geq 0
    $$
    Leveraging the latter, \eqref{eq:strong_convexity} simplifies into:
    \begin{align}
        H(u) - H(u^\star) &\geq \frac{\mu_H}{2} \|u - u^\star\|^2. \label{eq:strong_convexity_simplified}
    \end{align}
    Because $c^{-1}$ is Lipschitz continuous with modulus $ 1/\mu_c$, we obtain:
    $$
        \norm{x - x^\star} = \norm{c^{-1}(u) - c^{-1}(u^\star)} \leq \frac{1}{\mu_c} \norm{u - u^\star}.
    $$
    % Re-arranging the above inequality provides a lower bound for $\norm{u - u^\star}$ in terms of $\norm{x - x^\star}$:
    % $$
    %     \norm{u - u^\star} \geq\frac{1}{\Lcinv} \norm{x - x^\star}.
    % $$
    Substituting this bound into equation \eqref{eq:strong_convexity_simplified}, we obtain:
    \begin{align}
        H(u) - H(u^\star) &\geq \frac{\mu_H}{2} \norm{u - u^\star}^2 \\
        &= \frac{\mu_c^2 \mu_H }{2} \|x - x^\star\|^2. \label{eq:quadratic_growth_intermediate}
    \end{align}
    Recognizing that \( H(u) = f(x) \) and \( H(u^\star) = f(x^\star) \), we can rewrite the above inequality as:
    \begin{align}
        f(x) - f(x^\star) &\geq \frac{\mu_c^2\mu_H }{2} \|x - x^\star\|^2. \label{eq:quadratic_growth_final}
    \end{align}
    
    This concludes the proof.
\end{proof}

\begin{proposition}[QG from KL]
    Let an $\ell$-smooth function $f:\calX\to\R$ satisfy the Kurdyka-\L ojasiewicz (KL) condition with modulus $\mu>0$, \textit{i.e.}:
    $$\mu \paren{f(x) - f^\star}\leq \frac{1}{2} \inf_{s_x\in \partial F(x) } \norm{s_x}^2, \quad \forall x \in \calX.$$
    Then, it satisfies the quadratic growth (QG) condition with a modulus $\frac{\mu}{2}$, \textit{i.e.,} 
    $$\frac{\mu}{4}\norm{x - x^\star_p}^2 \leq f(x) - f^\star,\quad \forall x \in \calX, $$
    where $x^\star_p \in \argmin_{x^\star \in \calX^\star}\norm{x - x^\star}^2$ and in turn $\calX^\star:= \argmin_{x\in\calX} f(x)$.
    \label{prop:qg-from-kl}
\end{proposition}
\begin{proof}
Before proceeding, we note that the PL condition in \citep{liao2024error} is what we call the KL condition in this manuscript. Let us define $F(\cdot) := f(\cdot) + I_\calX(\cdot)$ where $I_\calX$ is the indicator function of the compact convex set $\calX$. We note that, $I_\calX$ is convex and $f$ is $\ell$-smooth and as such $\ell$-weakly convex. Hence, $F(\cdot) = f(\cdot) + I_\calX(\cdot) $ is also $\ell$-weakly convex. By \citep[Proof of Theorem 3.1 (PL$\to$EB)]{liao2024error}
    \begin{align}
        \norm{x - x^\star_{p}}\leq \frac{1}{\mu} \inf_{s\in \partial F(x)}\norm{s_x},
    \end{align}
    where $x^\star_p \in \argmin_{x^\star \in \calX^\star}\norm{x - x^\star}^2$ and in turn $\calX^\star:= \argmin_{x\in\calX} f(x)$.
Then, by \citep[Proof of Theorem 3.1 (EB$\to$QG)]{liao2024error}, we know that,
\begin{align}
    \frac{1 }{{4}/{\mu}}\norm{x - x^\star_{p}}^2 \leq f(x) - f^*.
\end{align}
    
\end{proof}

\begin{proposition}[QG from pPL;~{\citep[Theorem 2]{mulvaney2022dynamic}}]
    Let an $\ell$-smooth function $f:\calX\to\R$ satisfy the proximal Polyak-\L ojasiewicz (PL) condition with modulus $\mu>0$. Then, it satisfies the quadratic growth (QG) condition with the same modulus ${\mu}$,
    $$\frac{\mu}{4}\norm{x - x^\star_p}^2 \leq f(x) - f^\star,\quad \forall x \in \calX, $$
    where $x^\star_p \in \argmin_{x^\star \in \calX^\star}\norm{x - x^\star}^2$ and in turn $\calX^\star:= \argmin_{x\in\calX} f(x)$.
        \label{prop:qg-from-ppl}
\end{proposition}

\subsection{Warm-up: Stochastic Projected Gradient Descent on a proximal-PL Function}

% \begin{lemma}[{\citep[Lemma 6]{j2016proximal}}]
%     Let $f:\calX\to\R$ be an $\ell$-smooth function and a point $x \in \calX \subseteq \R^d$. Also, define the vector $v\in\R^d$ and $y\in\calX$ to be $$ y := \proj_{\calX} \paren{x - \eta v}.$$ 
%      Then, the following inequality is true:
%     \begin{align}
%         f(y) &\leq f(z) + \inprod{\nabla f(x) - v}{y-z} \\ &
%         \quad+ \paren{\frac{\ell}{2} - \frac{1}{2\eta}}\norm{y - x}^2
%         +\paren{\frac{\ell}{2} + \frac{1}{2\eta}}\norm{z - x}^2 - \frac{1}{2}\norm{y - z}^2.
%     \end{align}
%     \label{lemma:reddi-three-point}
% \end{lemma}

Now consider stochastic projected gradient on a nonconvex pPL $\ell$-smooth function $f$. The analysis closely follows \citep{j2016proximal}.

% \textbf{Reminder}, pPL:
% \begin{align}
%     \frac{1}{2}\calD_{\calX}(x,\ell) \geq\mu \paren{ f(x) - f(x^*)},
% \end{align}
% where
% $$\calD_\calX(x,\alpha) = -2 \alpha \min_{y \in \calX} \left\{  \inprod{\nabla f(x)}{y-x} + \frac{\alpha}{2} \norm{y- x}^2 \right\}.$$

Assume $\E\hatgradfbx(x)= \gradfbx(x)$, $\norm{\hatgradfbx(x)- \gradfbx(x)}\leq \delta$ and $\E \norm{\hatgradfbx(x)}^2 \leq \sigma^2$. Further, define points $x\at{t+1},\bar{x}\at{t+1}$ to be:
    \begin{align}
        x\at{t+1} &= \proj_\calX\left( x\at{t} - \eta \hatgradfbx(x\at{t}) \right);\\
        \bar{x}\at{t+1} &= \proj_{\calX}\paren{x\at{t} - \eta \nabla_x f(x\at{t})}.
    \end{align}
\begin{theorem}
\label{grad-descent-ppl}
Assume an $\ell$-smooth function $f:\calX\to\R$ with an inexact stochastic gradient oracle $\hatg_x$ such that $\E\hatgradfbx(x)= \gradfbx(x)$, $\norm{\hatgradfbx(x)- \gradfbx}\leq \delta$ and $\E \norm{\hatgradfbx}^2 \leq \sigma^2$.
Then, 
\begin{itemize}
    \item 
 running stochastic gradient descent with a $\delta_x$-inexact gradient for $T>0$ iterations, yields the following inequality:
\begin{align}
    \frac{1}{T}\sum_{t=0}^{T-1}\frac{1}{2} \calD_{\calX}(x\at{t+1},\ell) \leq \frac{3\ell \paren{\E f(x\at{0}) -  f(x^\star)}}{T} +  3 \delta_x^2 + 3 \sigma_x^2.
\end{align}
\item Further, if $f$ is pPL with modulus $\mu$, the followinng inequality holds true:
\begin{align}
    \E f(x\at{T+1}) - f(x^\star)\leq \exp\paren{ - \frac{\mu}{3\ell}T}  \paren{ f(x\at{0}) - f^\star} + \frac{3\ell\delta_x^2}{\mu} + \frac{3 \ell\sigma_x^2 }{\mu}.
\end{align}
\label{theorem:stoch-pgd-on-nonconvex-and-on-ppl}
\end{itemize}
\end{theorem}
\begin{proof}
    
\end{proof}
\begin{itemize}
    \item We let $\bar{x}\at{t+1}:= \proj_{\calX}\paren{x\at{t} - \eta \nabla f(x\at{t}) }$. Invoking  $\ell$-smoothness of $f$ for $ x\at{t}, \bar{x}\at{t+1}$ and requiring that $\eta\leq\frac{1}{\ell}$,
\begin{align}
    f(\bar{x}\at{t+1}) &\leq f(x\at{t}) + \inprod{\nabla f(x\at{t})}{\bar{x}\at{t+1} - x\at{t}} + \frac{1}{2\eta}\norm{\bar{x}\at{t+1}-x\at{t}}^2\\
    &= f(x\at{t}) -\paren{ \inprod{\nabla f(x\at{t})}{ x\at{t} - \bar{x}\at{t+1} } - \frac{1}{2\eta}\norm{\bar{x}\at{t+1}-x\at{t}}^2 }\\
    &= f(x\at{t}) - \frac{\eta}{2}\calD_{\calX}(x\at{t},1/\eta)
    \label{ppl-descent-ineq1}
\end{align}
    \item Invoking \cref{lemma:reddi-three-point} with $x = x\at{t}$, $y = \bar{x}\at{t+1}$, $z = x\at{t}$, $v = \nabla f(x\at{t})$
    \begin{align}
        f(\barx\at{t+1}) \leq 
        f(x\at{t}) 
        + \paren{\frac{\ell}{2} - \frac{1}{\eta}}\norm{\barx\at{t+1} - x\at{t}}^2.
        \label{ppl-descent-ineq2}
    \end{align}

    \item Again, invoking \cref{lemma:reddi-three-point} but with $x = x\at{t}$, $y = x\at{t+1}$, $z = \barx\at{t+1}$, $v = \hatg(x\at{t})$,
    \begin{align}
        f(x\at{t+1}) &\leq f(\barx\at{t+1}) + \inprod{\nabla f(x\at{t}) - \hatgradfbx(x\at{t})}{x\at{t+1}-\barx\at{t+1}} \\ &
        \quad+ \paren{\frac{\ell}{2} - \frac{1}{2\eta}}\norm{x\at{t+1} - x\at{t}}^2
        +\paren{\frac{\ell}{2} + \frac{1}{2\eta}}\norm{\barx\at{t+1} - x\at{t}}^2 - \frac{1}{2\eta}\norm{x\at{t+1} - \barx\at{t+1}}^2. 
        \label{ppl-descent-ineq3}
    \end{align}

    Adding $1/3\times$\eqref{ppl-descent-ineq1} and $2/3\times$\eqref{ppl-descent-ineq2}
    \begin{align}
        f(\barx\at{t+1}) \leq f(x\at{t}) - \frac{\eta}{6}\calD_{\calX}(x\at{t},1/\eta) + \paren{\frac{\ell}{3} - \frac{2}{3\eta}}\norm{\barx\at{t+1} - x\at{t}}^2
    \end{align}

    Adding \eqref{ppl-descent-ineq3},
    \begin{align}
        f(x\at{t+1}) &\leq f(x\at{t})  - \frac{\eta}{6}\calD_{\calX}(x\at{t},1/\eta) + \paren{\frac{\ell}{3} - \frac{2}{3\eta}}\norm{\barx\at{t+1} - x\at{t}}^2\\
        &\quad  + 
        \inprod{\nabla f(x\at{t}) - \hatgradfbx(x\at{t})}{x\at{t+1}-\barx\at{t+1}} 
        \\ &
        \quad+ \paren{\frac{\ell}{2} - \frac{1}{2\eta}}\norm{x\at{t+1} - x\at{t}}^2
        +\paren{\frac{\ell}{2} + \frac{1}{2\eta}}\norm{\barx\at{t+1} - x\at{t}}^2 
        - \frac{1}{2\eta}\norm{x\at{t+1} - \barx\at{t+1}}^2 \\
        % &\leq f(x\at{t}) - \frac{\eta}{6}\calD_{\calX}(x\at{t},1/\eta) +  \paren{\frac{5\ell}{6} - \frac{1}{6\eta}}\norm{\barx\at{t+1} - x\at{t}}^2 \\
        % &\quad + \frac{\rho}{2}\norm{\nabla f(x\at{t}) - \hatgradfbx(x\at{t})}^2 + \frac{1}{2\rho} \norm{x\at{t+1}-\barx\at{t+1}}^2\\
        % &\quad + 
        % \paren{\frac{\ell}{2} - \frac{1}{2\eta}}\norm{x\at{t+1} - x\at{t}}^2
        % - \frac{1}{2\eta}\norm{x\at{t+1} - \barx\at{t+1}}^2
        % \label{ppl-descent-peter-paul}
        % \\
        &\leq f(x\at{t}) - \frac{\eta}{6}\calD_{\calX}(x\at{t},1/\eta) 
        + \paren{\frac{5\ell}{6} - \frac{1}{6\eta}}\norm{\barx\at{t+1} - x\at{t}}^2 \\
        &\quad + \frac{\eta}{2}\norm{\nabla f(x\at{t}) - \hatgradfbx(x\at{t})}^2 
        \label{ppl-descent-peter-paul}
        \\
        &\quad + 
        \paren{\frac{\ell}{2} - \frac{1}{2\eta}}\norm{x\at{t+1} - x\at{t}}^2
         \label{ppl-descent-gather}
        \\
        &\leq f(x\at{t}) - \frac{\eta}{6}\calD_{\calX}(x\at{t},1/\eta) + \frac{\eta}{2}\norm{\nabla f(x\at{t}) - \hatgradfbx(x\at{t})}^2 
        \label{ppl-descent-drop-negatives}
        \\
        &\leq f(x\at{t}) - \frac{\eta}{6}\calD_{\calX}(x\at{t},1/\eta)
        + {\eta}\norm{\nabla f(x\at{t})     - \gradfbx(x\at{t})}^2
        + {\eta}\norm{\gradfbx(x\at{t})     - \hatgradfbx(x\at{t})}^2 
        \label{ppl-descent-triangle}
        \\
        % &\quad + 
        % \paren{\frac{\ell}{2} - \frac{1}{2\eta}}\norm{x\at{t+1} - x\at{t}}^2
    \end{align}
\end{itemize}
We note that we have picked $\eta \leq \frac{1}{5\ell}$.
\begin{itemize}
    \item \eqref{ppl-descent-peter-paul} follows from the application of Young's inequality $\inprod{a}{b} \leq \frac{\rho}{2}\norm{a}^2 + \frac{1}{2\rho}\norm{b}^2$ with $\rho = \eta$;
    \item \eqref{ppl-descent-drop-negatives} follows by dropping the non-positive terms; non-positivity follows from the inequalities:
\begin{enumerate}
    \item $0 \geq 5\ell - \frac{1}{\eta} $
    \item $ \ell \leq 5\ell \leq \frac{1}{\eta}$ then $\ell - \frac{1}{\eta} \leq 0$
\end{enumerate}
\item \eqref{ppl-descent-triangle} follows from $\|a + b\|^2\leq 2\|a\|^2 + 2\|b\|^2.$ 
\end{itemize}

Taking expectations:
\begin{align}
    \E f(x\at{t+1}) &\leq \E f(x\at{t}) - \frac{\eta}{6}\calD_{\calX}(x\at{t},1/\eta)+ \eta \delta_x^2 + \eta \sigma_x^2 \label{ppl-descent-alt}
    % &\leq \E f(x\at{t}) - \frac{5\eta}{6}\E \calD_{\calX}(x\at{t},\ell) + \eta \delta_x^2 + \eta \sigma_x^2 ,
    % \label{ppl-descent}
\end{align}
where the last inequality follows from the choice of the step-size. Rearranging \eqref{ppl-descent-alt} and summing over $T$, we get:
\begin{align}
    \frac{1}{T}\sum_{t=0}^{T-1}\frac{1}{2} \calD_{\calX}(x\at{t},1/\eta)  &\leq
    \frac{3}{\eta T} \sum_{t=0}^{T-1}
     \paren{\E f(x\at{t}) -  \E f(x\at{t+1}) + 3 \delta_x^2 + 3 \sigma_x^2
     }
     \\
    &\leq \frac{3 \paren{\E f(x\at{0}) -  f(x^\star)}}{\eta T} +  3\delta_x^2 + 3\sigma_x^2.
\end{align}

When the pPL condition holds and $\eta = \frac{1}{5\ell}$, we get from \eqref{ppl-descent-alt},
\begin{align}
    \E f(x\at{t+1}) - f^\star \leq \paren{1 - \frac{\mu}{3\ell}}  \paren{ \E f(x\at{t}) - f^\star} +  \delta_x^2 +  \sigma_x^2,
\end{align}
and as such,
\begin{align}
    \E f(x\at{T}) - f^\star &\leq \paren{1 - \frac{\mu}{3\ell}}^{T}  \paren{  f(x\at{0}) - f^\star} + \paren{ \delta_x^2 +  \sigma_x^2 } \sum_{t=1}^T \paren{1 - \frac{\mu}{3\ell}}^{t}\\
    &= \paren{1 - \frac{\mu}{3\ell}}^{T}  \paren{ f(x\at{0}) - f^\star} + \paren{ \delta_x^2 +  \sigma_x^2 } \frac{1 -\paren{1 - \frac{\mu}{3\ell}}^{T+1} }{1 - \paren{1 - \frac{\mu}{3\ell}} }\\
    &\leq \paren{1 - \frac{\mu}{3\ell}}^{T}  \paren{  f(x\at{0}) - f^\star} + \frac{3\ell\delta_x^2}{\mu} + \frac{3 \ell\sigma_x^2 }{\mu}\\
    &\leq \exp\paren{ - \frac{\mu}{3\ell}T }  \paren{  f(x\at{0}) - f^\star} + \frac{3\ell\delta_x^2}{\mu} + \frac{3 \ell\sigma_x^2 }{\mu}.
\end{align}
% \section{The Duality Gap of Functions Under Regularity Conditions}
% \begin{lemma}[Exchangeability of Minmax Strategies]
%     Assume a function $f:\calX\times\calY\to\R$ such that $\min_{x\in\calX}\max_{y\in\calY}f(x,y) = \max_{y\in\calY}\min_{x\in\calX} f(x,y) =: f^\star$. Then, if $(x,y)$ and $(\bar{x},\bar{y})$ are two $\epsilon$-NEs, it also the case that $(\bar{x},y)$ and $(x,\bar{y})$ are $\epsilon$-NEs.
% \end{lemma}
% \begin{proof}
%     \begin{align}
%         &f^\star - \epsilon \leq f(x,y) \leq f^\star + \epsilon;\\
%         &f^\star - \epsilon \leq f(\bar{x},\bar{y}) \leq f^\star + \epsilon.
%     \end{align}
    
%     By the definition of an $\epsilon$-NE, $\forall x'\in\calX, y'\in \calY$,
%     \begin{align}
%         &f(x',y) + \epsilon \geq f(x,y) \geq f(x,y') -\epsilon;\\
%         &f(x',\bar{y}) + \epsilon \geq f(\bar{x},\bar{y}) \geq f(\bar{x},y') -\epsilon.
%     \end{align}
    
%     \begin{align}
%         f(x,y) \leq f(\bar{x},y) + \epsilon \leq f(\bar{x},\bar{y}) + 2\epsilon
%     \end{align}
    
%     \begin{align}
%         f(\bar{x},y) + \epsilon\leq f^\star + \epsilon
%     \end{align}
% \end{proof}

% \begin{lemma}
%     Assume a Lipschitz continuous and smooth function $f:\calX\times\calY\to\R$ such that there exist, $\mu_1, \mu_2$, for which 
%     \begin{gather}
%         \max_{x\in\calX}\inprod{\nabla_x f(x,y)}{y} \geq \mu_1 ~\textbf{or}~
%     \end{gather}
% \end{lemma}

\section{Constrained Min-Max Optimization Under the pPL Condition}
\label{sec:convergence-min-max}
\subsection{Key Lemmata}
% \section{Constraind Min-Max Optimization under NC-pPL}
\begin{theorem}[NC-pPL and cont. of maximizers]
    Let function $f:\calX\times\calY\to \R$ with $f(x,\cdot)$ satisfying the proximal-PL condition with parameter $\mu$ for all $x\in\calX$. Then, consider points $x_1,x_2 \in\calX$ and $y^\star(x_1),y^\star(x_2) \in \calY$ with $y^\star(x_1):= \argmax_{y\in\calY} f(x_1, y)$ and $y^\star(x_2):= \argmax_{y\in\calY} f(x_2, y)$, it holds true that:
    \begin{align}
        \norm{y^\star(x_1) - y^\star(x_2)}\leq L_{\star}\norm{x_1 - x_2},
    \end{align}
    where $L_\star := \frac{\ell}{\mu } $.
\end{theorem}

\begin{remark}
    One might compare the last statement to the Robust Berge Maximum Theorem \citep[Th. 3.20]{papadimitriou2023computational} which concerns (non)convex--strongly-concave functions with coupled feasibility sets. Essentially, the former statement illustrates that hidden-strong-concavity is in some aspect a stronger assumption than strong-concavity; in hidden-strong-concavity the feasibility sets are only ``hiddenly'' coupled. This allows us to decouple the constraint sets and view the problem as a constrained nonconvex-pPL problem. Then, it is quite intuitive that Lispchitz continuity of the maximizers holds in light of the analogous result \citep[Lemma A.3]{nouiehed2019solving} which concerns unconstrained min-max optimization over nonconvex-PL functions. Ultimately, this decoupling principle is also the reason why \citet{kalogiannis2024learning} can compute the gradient of the maximum function by only invoking Danskin's Theorem and not the more elaborate Envelope Theorem~\citep{afriat1971theory}.
\end{remark}

\begin{proof}
    % By \cref{claim:continuity_of_stationarity} we know that:
    % \begin{align}
    %     \frac{1}{\eta}\norm{y_1 - \proj_{\calY} (y_1 + \eta \nabla f(x_2, y_1 )} \leq \ell \norm{x_1 - x_2}. 
    % \end{align}
    % Additionally, 
    % \begin{align}
    %     \frac{1}{\eta} \norm{y_1 - \proj_{\calY} (y_1 + \eta \nabla f(x_2, y_1 )}^2  \geq \calD_{\calY}(x_2, y, \eta )
    % \end{align}
    Since we have defined the pPL and QG growth conditions for a minimization problem, let us assume $g(x,y) = - f(x,y)$. Consequentially, $ \argmin_{y\in\calY}g(x,y) = \argmax_{y\in\calY}f(x,y)\ni y^\star(x)$. Finally, we define:
    \begin{equation}
        \calD_{\calY}(y,\alpha;x):= -2\alpha \min_{z\in\calY} \left\{ \inprod{\nabla_y g(x,y)}{z-y} + \frac{\alpha}{2}\norm{z - y}^2 \right\}.
    \end{equation}
    By the proximal-PL condition, it holds that,
    \begin{align}
        \frac{1}{2}\calD_{\calY}\paren{y^\star(x_1), 1/\eta ; x_2 } \geq \mu\paren{   g\paren{x_2, y^\star(x_1) }  - g\paren{x_2, y^\star(x_2)}  }
    \end{align}

    By the QG condition:
    \begin{align}
          g\paren{x_2, y^\star(x_1) } - g\paren{x_2, y^\star(x_2)} \geq \frac{\muqg }{2}\norm{y^\star(x_1) - y_p^\star(x_2)}^2
    \end{align}
    Where, $y^\star_{p}(x_2):= \argmin_{y'\in\calY^\star{x}}\|y^\star(x_1) - y'\|$ and $\calY^\star(x):= \argmin_{y\in\calY}g(x,y)$. Finally, by \cref{lemma:D-y-yprime-bound},
    \begin{align}
        \calD_{\calY}\paren{y^\star(x_1), 1/\eta ; x_2 } \leq \ell^2\norm{x_1 - x_2}^2.
    \end{align}

    Putting these pieces together,
    \begin{align}
         \frac{\muqg }{2}\norm{y^\star(x_1) - y_p^\star(x_2)}^2 \leq \frac{\ell^2}{2\mu} \norm{x_1 - x_2}^2.
    \end{align}
    
    Further, we know that pPL with modulus $\mu$ implies QG with the same modulus (\cref{prop:qg-from-ppl}). This concludes the proof.
    
\end{proof}

\begin{corollary}
Let function $f:\calX\times\calY\to \R$ with $f(x,\cdot)$ satisfying the proximal-PL condition with modulus $\mu$ for all $x\in\calX$. Then, let $\Phi(x):= \max_{y\in\calY}f(x,y)$. For any two points $x_1,x_2\in\calX$ it holds true that:
    \begin{align}
        \norm{\nabla_x \Phi(x_1) - \nabla_x \Phi(x_2) }\leq \ellphi \norm{x_1 - x_2},
    \end{align}
    where $\ellphi := \ell (1 + L_\star) = {\ell} + \frac{\ell^2}{\mu} $.    
\end{corollary}
\begin{proof}
    We write,
    \begin{align}
            \norm{\nabla_x \Phi(x_1) - \nabla_x \Phi(x_2) } &= \norm{\nabla_x f(x_1,y^\star(x_1) ) - \nabla_x f(x_2,y^\star
            (x_2)) } \\
            &\leq \ell \norm{(x_1,y^\star(x_1)) -  (x_2,y^\star
            (x_2)) } \\
            &\leq \ell \norm{x_1 - x_2} + \ell L_\star\norm{x_1 - x_2}.
    \end{align}
    The first equation holds due to Danskin's lemma, and the first inequality is due to $\ell$-Lipschitz continuity of the gradient. The second inequality is due to the triangle inequality and the $L_\star$-Lipschitz continuity of the maximizers.
\end{proof}

\begin{lemma}
    Let $f:\calX\times\calY\to \R$ be an $L$-Lipschitz continuous and $\ell$-smooth function that satisfies the two-sided pPL condition for both $f(\cdot,y)$ and $f(x,\cdot)$, then:
    \begin{align}
        \min_{x\in\calX}\max_{y\in\calY} f(x,y) =\max_{y\in\calY} \min_{x\in\calX} f(x,y) =: \Phi^\star.
    \end{align}
\end{lemma}
\begin{proof}
    We can invoke \citep[Lemma 2.1]{yang2020global} which holds for two-sided pPL functions with minor modifications. 
    % Further, \citep[Lemma 12]{daskalakis2020independent} 
    
    % We begin by observing that the pPL condition implies the gradient dominance condition as, $-\frac{\alpha^2}{2}\norm{x - z}^2 \leq 0$, we can write
    % \begin{align}
    %     2\alpha \inprod{\nabla_x f (x)}{x - z} \geq 2\alpha \inprod{\nabla_x f (x)}{x - z} - \frac{\alpha}{2}\norm{x- z}^2.
    %     % \calD_{\calX}(x,\alpha) = 2\alpha\max_{z\in\calX} \left\{ \inprod{\nabla_x f (x)}{x - z} - \frac{\alpha}{2}\norm{x- z}^2 \right\}
    % \end{align}
    % Maximizing over $z$ in both sides,
    % \begin{align}
    %     2\alpha \max_{z\in\calX} \inprod{\nabla_x f (x)}{x - z} \geq \calD_{\calX}(x,\alpha).
    % \end{align}
    % As such, we can invoke \citep[Lemma 12]{daskalakis2020independent}, to see that 
\end{proof}

\begin{lemma}
    Let the function $f:\calX\times\calY \to \R$ satisfy the pPL condition with moduli $\mu_1, \mu_2 >0$ respectively for $f(\cdot, y),\forall y \in\calY$ and $f(x,\cdot),\forall x \in \calX$. Then, the function $\Phi:\calX \to \R$ with $\Phi(x) = \max_{y'\in\calY} f(x,y')$ satisfies the pPL condition with modulus $\mu_1$.
    \label{lemma:ppl-phi}
\end{lemma}
\begin{proof}
    For the purposes of this proof, we will enhance the notation of $\calD_\calX$ as follows:
    \begin{align}
        \calD_\calX^f(x,\alpha; y) &:= -2\alpha \min_{z \in \calX} \left\{  \inprod{\nabla f(x,y)}{z - x} + \frac{\alpha}{2} \norm{z - x}^2 \right\};
        \\
        \calD_\calX^\Phi(x,\alpha) &:= -2\alpha \min_{z \in \calX} \left\{  \inprod{\nabla \Phi(x)}{z - x} + \frac{\alpha}{2} \norm{z - x}^2 \right\}.
    \end{align}

    \begin{align}
        \calD_\calX^\Phi(x,\ell_\Phi) & = \calD_\calX^f\paren{x,\ell_\Phi; y^\star(x)}  \\
        &\geq 2\mu_1 \paren{ f(x,y^\star(x)) - \min_{x'\in\calX} f(x',y^\star(x))}.
    \end{align}
    Further, 
    \begin{align}
        f(x',y^\star(x)) \leq \max_{y\in\calY} f(x', y),
    \end{align}
    and minimizing on both sides yields,
    \begin{align}
         \min_{x'\in\calX} f(x',y^\star(x)) \leq \min_{x'\in\calX} \max_{y\in\calY} f(x', y) = \Phi^\star.
    \end{align}
    Hence,
    \begin{align}
        \calD_\calX^\Phi(x,\ell_\Phi) \geq 2\mu_1 \paren{ \Phi(x) - \Phi^\star}.
    \end{align}
\end{proof}

    \subsection{Stationarity Proxies and the Gradient Oracle}
    \begin{definition}
    We define $\calD_\calY$ and $\calD_\calX^\Phi$ to be the following quantities:
        \begin{align}
            \calD_\calY (y,\alpha;x)= - 2 \alpha \min_{z\in\calY}
            \left\{ 
                \inprod{ - \nabla f(x,y)}{z - y} + \frac{\alpha}{2}\norm{z - y}^2 
            \right\} ,
        \end{align}
    and correspondigly,
        \begin{align}
            \calD_\calX^\Phi (x,\alpha)= - 2 \alpha \min_{z\in\calX}
            \left\{ 
                \inprod{ \nabla \Phi(x) }{z -x} + \frac{\alpha}{2}\norm{z - x}^2 
            \right\} .
        \end{align}
    \end{definition}
    
    % \fk{move elsewhere}
    \begin{definition}
    We define the deterministic and stochastic gradient mapping at point $x\at{t}$ and $y\at{t}$ to be:    
    \begin{itemize}
        \item $\resx{}(x) :=\frac{1}{\etax}\paren{ x - \proj_\calX\paren{x - \etax {v} }}$ and $\resx{t} = \resx{}(x\at{t})$;
        \item $\sresx{t}(x\at{t}) :=\frac{1}{\etax}\paren{ x - \proj_\calX\paren{x - \etax \hatgradfbx(x\at{t},y\at{t}) }}$ and $\sresx{t} = \sresx{}(x\at{t})$,
    \end{itemize}
    and respectively:
    \begin{itemize}
        \item $\resy{}(y) :=\frac{1}{\etay}\paren{ \proj_\calY\paren{y + \etay {v} }  -  y}$ and $\resy{t} = \resy{t}(y)$;
        \item $\sresy{}(y) :=\frac{1}{\etay}\paren{ \proj_\calY\paren{y + \etay \hat{v} }  -  y}$ and $\sresy{t} = \sresy{t}(y)$.
    \end{itemize}
    \end{definition}
\allowdisplaybreaks
    \begin{assumption}[Unbiased Inexact Gradient Estimators and Bounded Second Moments]
For all iterations \( t \), the gradient estimators \( \hat{g}_x(x\at{t}, y\at{t}) \) and \( \hat{g}_y(x\at{t}, y\at{t}) \) satisfy
\[
\mathbb{E}\left[ \hat{g}_x(x\at{t}, y\at{t}) \right] = g(x\at{t}, y\at{t}),
\]
\[
\mathbb{E}\left[ \hat{g}_y(x\at{t}, y\at{t}) \right] =  g(x\at{t}, y\at{t}),
\]
and
\[
\mathbb{E}\left[ \left\| \hat{g}_x(x\at{t}, y\at{t}) \right\|^2 \right] \leq \sigma_x^2,
\]
\[
\mathbb{E}\left[ \left\| \hat{g}_y(x\at{t}, y\at{t}) \right\|^2 \right] \leq \sigma_y^2.
\]

In turn, $\norm{ g_x(x\at{t}, y\at{t}) - \nabla_x f(x\at{t}, y\at{t})}\leq \delta_x$, $\norm{ g_y(x\at{t}, y\at{t}) - \nabla_y f(x\at{t}, y\at{t})} \leq \delta_y $.
\label{assumption:stoch-grad-oracle}
\end{assumption}

\begin{remark}[Bound on Second Moment instead of Variance] At first, it might appear a slightly stronger assumption to place a bound on the second moment of the gradient estimator. Yet, in most relevant works the variance of the relevant estimators is bounded only after bounding the second moment~\citep{daskalakis2020independent,zhang2020variational,zhang2021convergence}. As such, this assumption is reasonable and rather standard to satisfy for the aforementioned applications.
\end{remark}

\subsection{Convergence of Nested Gradient Iterations}
We can formulate the nested gradient iterations algorithm using the following template,
    \begin{align}
        \begin{cases}
        &y\at{t+1} \gets \mathtt{ARGMAX}(f(x\at{t}, \cdot), \epsilon_y);\\
        &x\at{t+1} \gets \proj_{\calX}\paren{x\at{t} - \eta \nabla f(x\at{t}, y\at{t+1})}
        \end{cases} \label{gdmax}
    \end{align}
    where, $\mathtt{ARGMAX}(h,\epsilon)$ returns an $\epsilon$-approximate maximize of function $h$. As a function, it can be implemented efficiently by projected gradient ascent. 

    Finally, the outer loop of the process implements projected gradient descent with a stochastic and inexact gradient feedback on $\Phi(\cdot)$. %

\begin{theorem}[NC-pPL]
Let $f:\calX\times \calY\to \R$ be an $\ell$-smooth function  satisfying the pPL condition with a modulus $\mu$ for $f(x,\cdot)$. Then, after $T$ iterations of \eqref{gdmax} it holds true that:
\begin{align}
    \frac{1}{T}\sum_{t=0}^{T-1}\norm{\sresx{\Phi,t}}^2 &\leq \frac{6 { L\diam{\calX} }}{\etax T} +  6\delta_x^2 + 6\sigma_x^2.
\end{align}
\label{theorem:formal-nestedgit-nc-ppl}
\end{theorem}
\begin{proof}
By \cref{theorem:stoch-pgd-on-nonconvex-and-on-ppl} and a tuning of step-size $\etax \leq \frac{1}{5\ellphi}$, we get that,
\begin{align}
    \frac{1}{T}\sum_{t=0}^{T-1}\frac{1}{2} \calD_{\calX}^\Phi(x\at{t+1},1/\etax) &\leq \frac{3\ell \paren{\E \Phi (x\at{0}) -  \Phi (x^\star)}}{\etax T} +  3\delta_x^2 + 3\sigma_x^2.
\end{align}
In this context, the inexactness of the gradient oracle $\delta_x$ is: 
\begin{equation}
    \delta_x = \max_{0 \leq t\leq T-1}  \E \norm{\nabla\Phi(x\at{t})- \nabla f(x\at{t},y\at{t}) } = \max_{0 \leq t\leq T-1} \ell \E \norm{y\at{t} - y^\star(x\at{t}) },
\end{equation} 
while by the quadratic growth condition of $f(x,\cdot)$ we know that,
\begin{align}
    \frac{\muy}{2} \norm{y\at{t} - y^\star(x\at{t}) } &\leq \E \Phi(x\at{t}) - \E f(x\at{t},y\at{t})  \\
    &\leq \epsilon_y
\end{align}
where $\epsilon_y$ is the accuracy of the inner loop which we will defer tuning. We, invoke \cref{lemma:D-vs-grad-mapping} to see that the sum of $\calD_{\calX}^\Phi(x\at{t},1/\etax)$ upper bounds the sum of $\norm{\sresx{\Phi,t}}^2$:
\begin{align}
    \frac{1}{T}\sum_{t=0}^{T-1}\norm{\sresx{\Phi,t}}^2 &\leq \frac{6 \paren{\E \Phi (x\at{0}) -  \Phi (x^\star)}}{\etax T} +  6\delta_x^2 + 6\sigma_x^2.
\end{align}
Now, we see that $\Phi(x\at{0}) - \Phi(x^\star)$ is bounded by $L\diam{\calX}$ due to the $L$-Lipschitz continuity of $\Phi$ and the bounded diameter of $\calX$. Finally, we can tune $\etax$ and $\epsilon_y$.

\begin{itemize}
    \item $\etax =\frac{1}{5\ellphi}$
    \item $T = \frac{18 L \diam{\calX}}{\epsilon^2} $
    \item $M_x = \frac{\sigma_x^2}{18\epsilon^2}$
    \item $\epsilon_y = \frac{\epsilon}{\sqrt{18}}$.
\end{itemize}

\end{proof}

\begin{theorem}[pPL-pPL]
Let $f:\calX\times\calY\to \R$ be an $L$-Lipschitz, $\ell$-smooth function and $\calX,\calY$ be two compact convex sets with Euclidean diameters $\diam{\calX},\diam{\calY}$ respectively. Further, assume that $f(\cdot,y)$ satisfies the pPL condition with modulus $\mux$ for all $y\in\calY$ while $f(x,\cdot)$ satisfies the pPL condition with a modulus $\muy$ for any $x\in\calX$. Additionally, let $\hatgradfbx$ be an inexact stochastic gradient oracle such that $\E\hatgradfbx(x) = \gradfbx(x)$, $\E\norm{\hatgradfbx(x) - \gradfbx(x)}\leq \sigma_x^2$, and $\norm{\gradfbx(x) } \leq \delta_x$. Then, after $T$ iterations of \eqref{gdmax} with a tuning of step-sizes $\etax =\frac{1}{5\ellphi}$:
\begin{align}
    \E \Phi(x\at{T+1}) - \Phi(x^\star)\leq \exp\paren{ - \frac{\mu}{3\ellphi}T} L \diam{\calX} + \frac{3\ellphi\delta_x^2}{\mux} + \frac{3 \ellphi\sigma_x^2 }{\mux}.
\end{align}
\label{theorem:formal-nestedgit-2sided-ppl}
\end{theorem}
\begin{proof}
By \cref{theorem:stoch-pgd-on-nonconvex-and-on-ppl} with $\etax = \frac{1}{5\ellphi}$,
\begin{align}
        \E \Phi(x\at{T+1}) - \Phi(x^\star)\leq \exp\paren{ - \frac{\mu}{3\ellphi}T}  \paren{ \Phi(x\at{0}) - \Phi^\star} + \frac{3\ellphi\delta_x^2}{\mux} + \frac{3 \ellphi\sigma_x^2 }{\mux}.
\end{align}
First, we repeat the fact that $L\diam{\calX} \geq\Phi(x\at{0}) - \Phi^\star$ due to Lipschitz continuity. Then, we note that $\delta_x$ has to be tuned as $\sqrt{\frac{\epsilon \mux }{9\ellphi}}$ and the batch-size needs to be $M_x = \left\lceil\frac{9 \ellphi \sigma_x^2}{\epsilon\mux}\right\rceil$ where $\epsilon>0$ is the desired accuracy. Finally, $T =\left\lceil\frac{3\ellphi}{\mu} \log\paren{\frac{3 L\diam{\calX}}{\epsilon}}\right\rceil$.

\end{proof}

   \subsection{Convergence of Stochastic Alternating Gradient Descent-Ascent}
    In what follows, we will analyze the convergence of projected alternating gradient descent-ascent for nonconvex-pPL and two-sided pPL functions. The convergence proofs closely follow those of \citep{yang2022faster} and \citep{yang2020global} respectively after carefully modifying the arguments to make them work for the constrained setting. Convergence is proven by showing that an appropriate Lyapunov function diminshes along the trajectories of the algorithm's iterates~\citep{bof2018lyapunov}.

    In both scenarios we face, proving that the corresponding Lyapunov function diminishes is proven by first:
    \begin{itemize}
        \item lower bounding the descent on the maximum function $\Phi(\cdot)$ for every update on $x$;
        \item lower bounding the ascent on $f(x,\cdot)$ for every update on $y$;
        \item \emph{upper} bounding the descent on $f(\cdot,y)$ for every update on $x$.
    \end{itemize}
    
As a reminder, the iteration scheme of alternating gradient descent-ascent is the following,
    \begin{align}
        x\at{t+1} &= \proj_\calX \paren{x\at{t}-\etax \hatgradfbx(x\at{t},y\at{t}) };\\
        y\at{t+1} &= \proj_\calY \paren{y\at{t} + \etay \hatgradfby(x\at{t+1},y\at{t}) }.
    \end{align}
    
\subsubsection{NC-pPL}
\begin{theorem}
Let $f:\calX\times\calY\to \R$ be an $L$-Lipschitz, $\ell$-smooth function and $\calX,\calY$ be two compact convex sets with Euclidean diameters $\diam{\calX},\diam{\calY}$ respectively. Further, assume that $f(x,\cdot)$ satisfies the pPL condition with modulus $\mu$. Additionally, let $(\hatgradfbx, \hatgradfby)$ be an inexact stochastic gradient oracle satisfying \cref{assumption:stoch-grad-oracle}. Then, after $T$ iterations of \eqref{agda} with a tuning of step-sizes $\etax = \frac{1}{500\ell\kappa^2}$ and $\etay = \frac{1}{5\ell}$, it holds true that
\begin{align}
    \frac{1}{T}\sum_{t=1}^{T-1} 
    \paren{\E \norm{\resx{t}}^2 +  \kappa^2\E \calD_{\calY}(y\at{t},\ell;x\at{t})} \leq \frac{ \kappa^2\ell L (\diam{\calX} + \diam{\calY}) }{ T} +  \frac{c_2\sigma_x^2}{M_x}+ c_2\delta_x^2 + \frac{c_3\kappa^2\sigma_y^2}{M_y} + c_3 \kappa^2 \delta_y^2,
\end{align}
% $c_1 = 360000$, $c_2 = 473400$, and $c_3 = 1440$
where, $c_1, c_2, c_3 \in O(1)$.
\label{theorem:formal-agda-ncppl}
\end{theorem}
\begin{proof}

\item 
\paragraph{Descent bound on $\Phi$.}
$$\resx{\Phi,t} := \frac{1}{\etax}\paren{x\at{t} - \proj_{x} \paren{x\at{t} - \etax \nabla_x\Phi(x\at{t}) }}$$

Due to $\ell$-smoothness and the fact that $\frac{1}{\etax} \geq 5\ellphi$ we can use \cref{ppl-descent-alt} to get:
\begin{align}
    \E\Phi(x\at{t+1})
    &\leq \E \Phi(x\at{t}) 
    - \frac{\etax}{6} \E \calD_{\calX}^{\Phi}(x\at{t},1/\etax) 
    + \etax \E \norm{\nabla_x \Phi(x\at{t}) - \nabla_x f(x\at{t},y\at{t}) }^2
    + \etax \E\norm{\nabla_x f(x\at{t},y\at{t}) - \hatgradfbx(x\at{t},y\at{t}) }^2
    \\
    &
    \leq \E \Phi(x\at{t})
    - \frac{\etax}{6} \E \norm{\resx{\Phi,t}}^2 
    + \etax \E \norm{\nabla_x \Phi(x\at{t}) - \nabla_x f(x\at{t},y\at{t}) }^2 
    + \etax\E  \norm{\nabla_x f(x\at{t},y\at{t}) - \hatgradfbx(x\at{t},y\at{t}) }^2
 \\
    &
    \leq \E \Phi(x\at{t})
    - \frac{\etax}{6} \E \norm{\resx{\Phi,t}}^2 
    + \etax \E \norm{\nabla_x \Phi(x\at{t}) - \nabla_x f(x\at{t},y\at{t}) }^2 
    + 2\etax \sigma_x^2 + 2 \etax \delta_x^2
    % \E  \norm{\nabla_x f(x\at{t},y\at{t}) - \hatgradfbx(x\at{t},y\at{t}) }^2
\end{align}

\item
\paragraph{Ascent bound on $f(x\at{t},\cdot)$}
For a choice of $\etay \leq \frac{1}{5\ell}\leq \frac{1}{\ell}$,
\begin{align}
    \E f(x\at{t+1},y\at{t+1})\geq \E f(x\at{t+1},y\at{t}) + \frac{\etay}{6}\E\calD_\calY(y\at{t},1/\etay;x\at{t+1}) - \etay\delta_y^2 - \etay\sigma_y^2.
\end{align}

\item
    \paragraph{Upper Bound on the descent: $ f(x\at{t}, y\at{t}) - f(x\at{t+1}, y\at{t})$.}
    \begin{align}
        f(x\at{t+1}, y\at{t}) 
        &\geq f(x\at{t}, y\at{t})
        + \inprod{\nabla_x f(x\at{t},y\at{t})}{x\at{t+1} - x\at{t}} 
        - \frac{\ell}{2}\norm{x\at{t+1} - x\at{t}}^2 \\%%%%%%%%%%%%%%%%%%%%%%%%%%%%%%%%%%%%%%%%%%%%%%
        &= f(x\at{t}, y\at{t})
        - \etax \inprod{\gradfbx (x\at{t},y\at{t})}{\sresx{t}} 
        - \etax \inprod{\nabla_x f(x\at{t},y\at{t}) - \gradfbx (x\at{t},y\at{t})}{\sresx{t}} \\
        &\quad
        - \frac{\ell \etax^2}{2}\norm{\sresx{t}}^2 \\
        % %%%%%%%%%%%%%%%%%%%%%%%%%%%%%%%%%%%%%%%%%%%%%%
        &= f(x\at{t}, y\at{t})
        - \etax \inprod{\hatgradfbx (x\at{t},y\at{t})}{\sresx{t}} 
        - \etax \inprod{\gradfbx (x\at{t},y\at{t}) - \hatgradfbx (x\at{t},y\at{t})}{\sresx{t}} 
        \\&\quad
        - \etax \inprod{\nabla_x f(x\at{t},y\at{t}) - \gradfbx (x\at{t},y\at{t})}{\sresx{t}} \\
        &\quad
        - \frac{\ell \etax^2}{2}\norm{\sresx{t}}^2 \\
        %%%%%%%%%%%%%%%%%%%%%%%%%%%%%%%%%%%%%%%%%%%%%%
        &= f(x\at{t}, y\at{t})
        - \etax \inprod{\hatgradfbx (x\at{t},y\at{t})}{\sresx{t}} 
        \\
        &\quad
        - \etax \inprod{\gradfbx (x\at{t},y\at{t}) - \hatgradfbx (x\at{t},y\at{t})}{\resx{t}}
        - \etax \inprod{\gradfbx (x\at{t},y\at{t}) - \hatgradfbx (x\at{t},y\at{t})}{\sresx{t} - \resx{t}} 
        \\
        &\quad
        - \etax \inprod{\nabla_x f(x\at{t},y\at{t}) - \gradfbx (x\at{t},y\at{t})}{\sresx{t}} 
        \\
        &\quad
        - \frac{\ell \etax^2}{2}\norm{\sresx{t}}^2 \\
        %%%%%%%%%%%%%%%%%%%%%%%%%%%%%%%%%%%%%%%%%%%%%%
        &\geq f(x\at{t}, y\at{t})
        - \frac{\etax}{2} \norm{\hatgradfbx (x\at{t},y\at{t})}^2 - \frac{\etax}{2} \norm{\sresx{t}}^2 
        \label{young-s-inequality-descent}
        \\
        &\quad
        - \etax \inprod{\gradfbx (x\at{t},y\at{t}) - \hatgradfbx (x\at{t},y\at{t})}{\resx{t}}
        - \etax \inprod{\gradfbx (x\at{t},y\at{t}) - \hatgradfbx (x\at{t},y\at{t})}{\sresx{t} - \resx{t}} 
        \\
        &\quad
        - \frac{\etax}{2} \norm{\nabla_x f(x\at{t},y\at{t}) - \gradfbx (x\at{t},y\at{t})}^2 - \frac{\etax}{2}\norm{\sresx{t}} \\
        &\quad
        - \frac{\ell \etax^2}{2}\norm{\sresx{t}}^2 \\
        %%%%%%%%%%%%%%%%%%%%%%%%%%%%%%%%%%%%%%
        &\geq f(x\at{t}, y\at{t})
        - \paren{ \etax + \frac{\ell\etax^2}{2} } \norm{\sresx{t}}^2 - \frac{\etax}{2} \norm{\hatgradfbx (x\at{t},y\at{t})}^2  
        \\
        &\quad
        - \etax \inprod{\gradfbx (x\at{t},y\at{t}) - \hatgradfbx (x\at{t},y\at{t})}{\resx{t}}
        - \etax \inprod{\gradfbx (x\at{t},y\at{t}) - \hatgradfbx (x\at{t},y\at{t})}{\sresx{t} - \resx{t}} 
        \\
        &\quad
        - \frac{\etax \delta_x^2}{2} 
        \label{gather-and-inexact-descent}
        \\
        %%%%%%%%%%%%%%%%%%%%%%%%%%%%%%%%%%%%%%
        &\geq f(x\at{t}, y\at{t})
        - \paren{ \etax + \frac{\ell\etax^2}{2} } \norm{\sresx{t}}^2 - \frac{\etax}{2} \norm{\hatgradfbx (x\at{t},y\at{t})}^2  
        \\
        &\quad
        - \etax \inprod{\gradfbx (x\at{t},y\at{t}) - \hatgradfbx (x\at{t},y\at{t})}{\resx{t}}
        - \etax \norm{\gradfbx (x\at{t},y\at{t}) - \hatgradfbx (x\at{t},y\at{t})}^2
        - \frac{\etax \delta_x^2}{2} .
        \label{cs-and-nonexpansiveness-of-proj-descent}
    \end{align}
    \begin{itemize}
        \item the initial equations are mere additions-subtractions of terms and plugging-in of definitions;
        \item \eqref{young-s-inequality-descent} follows from Young's inequality;
        \item \eqref{gather-and-inexact-descent} follows from gathering terms and using the bound on the gradient inexactness error;
        \item \eqref{cs-and-nonexpansiveness-of-proj-descent} follows from the Cauchy-Schwarz inequality and the non-expansiveness of the projection.
    \end{itemize}
    Taking expectations again:
    \begin{align}
        \E f(x\at{t+1}, y\at{t}) 
        &\geq \E f(x\at{t}, y\at{t})
        - \paren{ \etax + \frac{\ell\etax^2}{2} } \E \norm{\sresx{t}}^2 - \frac{\etax \sigma_x^2 }{2}  
        \\
        &\quad
        - \etax \E \inprod{\gradfbx (x\at{t},y\at{t}) - \hatgradfbx (x\at{t},y\at{t})}{\resx{t}}
        - \etax \sigma_x^2
        - \frac{\etax \delta_x^2}{2} 
        \\
        %%%%%%%%%%%%%%%%%%%%%%%%%%%%%%%%%
        &= \E f(x\at{t}, y\at{t})
        - \paren{ \etax + \frac{\ell\etax^2}{2} } \E \norm{\sresx{t}}^2 - \frac{3\etax \sigma_x^2 }{2}
        - \frac{\etax \delta_x^2}{2}
        \label{gather-and-unbiasedness}
        \\
        %%%%%%%%%%%%%%%%%%%%%%%%%%%%%%%%%
        &\geq \E f(x\at{t}, y\at{t})
        - \frac{3\etax}{4} \E \norm{\sresx{t}}^2 - \frac{3\etax \sigma_x^2 }{2}
        - \frac{\etax \delta_x^2}{2} \\
        &\geq \E f(x\at{t}, y\at{t})
        - \frac{3\etax}{2} \E \norm{\resx{t}}^2 
         - \frac{9\etax\sigma_x^2}{2} \sigma_x^2
         -\frac{7\etax\delta_x^2}{2}
        % - \frac{6\etax \sigma_x^2 }{2}
        % -4 \etax \delta_x^2
        \label{eq:upper-bound-descent}
    \end{align}
    Where, we used that $\etax \leq \frac{1}{2\ell}$, which means that 
    $
    \frac{1}{2}
    \left( {\etax} + {\ell \etax^2}  \right) 
     \leq \frac{3\etax}{4}.
    $
    % \fk{the following had a typo, where was it used?}
    % \begin{align}
    %     \norm{x\at{t+1} - x\at{t}}^2 &\leq 2 \norm{\bar{x}\at{t+1} - x\at{t}}^2 + 2\norm{\bar{x}\at{t+1} - x\at{t+1}}^2\\
    %     &\leq 2 \norm{\bar{x}\at{t+1} - x\at{t}}^2 + 2{\eta^2} \norm{\nabla f(x\at{t}) - \hatgradfbx(x\at{t})}^2\\
    %     &\leq 2 \norm{\bar{x}\at{t+1} - x\at{t}}^2 + {4}{\eta^2} \delta_x^2 + {4}{\eta^2}
    % \end{align}

\item
    \paragraph{Bounding AGDA iterate difference $\E f(x\at{t+1},y\at{t+1})- \E f(x\at{t},y\at{t})$.}

    \begin{align}
        \E f(x\at{t+1},y\at{t+1})- \E f(x\at{t},y\at{t}) &\geq -\frac{3\etax}{2}\E\norm{\resx{t}} - \frac{9\etax\sigma_x^2}{2} - \frac{7\etax}{2}\delta_x^2\\
        &\quad + \frac{\etay}{6}\E\calD_\calY(y\at{t},1/\etay;x\at{t+1}) - \etay\delta_y^2 - \etay\sigma_y^2
    \end{align}

\item
\paragraph{The Lyapunov function.}
We will define the Lyapunov function,
\begin{align}
    V(x\at{t},y\at{t}) := \Phi(x\at{t}) + \alpha \Bigparen{\Phi(x\at{t}) - f(x\at{t},y\at{t})} = \paren{1+\alpha} \Phi(x\at{t}) - \alpha f(x\at{t},y\at{t}),  
\end{align}
with $\alpha >0$ to be tuned at the end. Also, we will denote $V\at{t} := V(x\at{t},y\at{t})$.

    \begin{align}
        \E V\at{t} - \E V\at{t+1} = \paren{1 + \alpha}\paren{\E\Phi(x\at{t}) - \E\Phi(x\at{t+1}) } + \alpha \paren{\E f(x\at{t+1},y\at{t+1}) - f(x\at{t},y\at{t}) }
    \end{align}

    \begin{align}
        \E V\at{t} - \E V\at{t+1} 
        &\geq (1+\alpha) \paren{
                \frac{\etax}{6}\E \norm{\resx{\Phi,t}}^2
                - \etax \E \norm{\nabla_x \Phi(x\at{t}) 
                - \nabla_x f(x\at{t},y\at{t}) }^2 
                - 2 \etax \sigma_x^2 - 2 \etax \delta_x^2}\\
        &\quad 
        + \alpha\paren{
        -\frac{3\etax}{2}\E\norm{\resx{t}}^2 
        - \frac{9\etax\sigma_x^2}{2} 
        - \frac{7\etax}{2}\delta_x^2
        }
        \\
        % - \frac{6\etax\sigma_x^2}{2}
        % %
        % - 2\etax\delta_x^2}\\
        &\quad 
        + \alpha \paren{
        \frac{\etay}{6}\E\calD_\calY(y\at{t},1/\etay;x\at{t+1}) 
        - \etay\delta_y^2 
        - \etay\sigma_y^2} \\
        %%%%%%%%%%%%%%%%%%%%%%%%%%%%%
        &\geq \frac{(1+\alpha)\etax}{6}\E \norm{\resx{\Phi,t}}^2
        -
        (1+\alpha)\etax \E \norm{\nabla_x \Phi(x\at{t}) 
                - \nabla_x f(x\at{t},y\at{t}) }^2 \\
        &\quad 
        -\frac{3\etax \alpha}{2}\E\norm{\resx{t}}^2 
        + \frac{\alpha\etay }{6}\E\calD_\calY(y\at{t},1/\etay;x\at{t+1})\\
        &\quad
        +\paren{ -2\etax(1+\alpha) - \frac{9\etax}{2}\alpha }\sigma_x^2
        + \paren{-\alpha \etay} \sigma_y^2\\
        &\quad
        +\paren{ -2\etax(1+\alpha) - \frac{7\etax}{2} \alpha }\delta_x^2
        + \paren{-\alpha \etay} \delta_y^2\\
        %%%%%%%%%%%%%%%%%%%%%%%%%%%%%%%%%%%%%%%%%%%%%%%%%%%%%%%%%%%%
        &\geq \frac{(1+\alpha)\etax}{6}\E \norm{\resx{\Phi,t}}^2
        -
        (1+\alpha)\etax \E \norm{\nabla_x \Phi(x\at{t}) 
                - \nabla_x f(x\at{t},y\at{t}) }^2 \\
        &\quad 
        -\frac{3\etax \alpha}{2}\E\norm{\resx{t}}^2 
        + \frac{\alpha\etay }{6}\E\calD_\calY(y\at{t},1/\etay;x\at{t}) - \frac{3 \alpha \etay \ell^2}{6} \E \norm{x\at{t}- x\at{t+1}}^2 \\
        &\quad
        +\paren{ -2\etax(1+\alpha) - \frac{9\etax}{2}\alpha }\sigma_x^2
        + \paren{-\alpha \etay} \sigma_y^2\\
        &\quad
        +\paren{ -2\etax(1+\alpha) - \frac{7\etax}{2} \alpha }\delta_x^2
        + \paren{-\alpha \etay} \delta_y^2
        \label{potential-use-lemma-continuity-of-D} \\
        %%%%%%%%%%%%%%%%%%%%%%%%%%%%%%%%%%%%%%%%%%%%%%%%%%%
        &\geq \frac{(1+\alpha)\etax}{6}\E \norm{\resx{\Phi,t}}^2
        -
        (1+\alpha)\etax \E \norm{\nabla_x \Phi(x\at{t}) 
                - \nabla_x f(x\at{t},y\at{t}) }^2 \\
        &\quad 
        -\frac{3\etax \alpha}{2}\E\norm{\resx{t}}^2 
        %\\
        + \frac{\alpha\etay }{6}\E\calD_\calY(y\at{t},\ell;x\at{t}) 
        \\
        &\quad
        - \frac{ \alpha \etax^2\etay \ell^2}{2} \E \norm{\resx{t}}^2 
        - {2 \alpha\etax^2 \etay \ell^2} \sigma_x^2 - 2\alpha\etax^2 \etay\ell^2\delta_x^2 \\
        &\quad
        +\paren{ -2\etax(1+\alpha) - \frac{9\etax}{2}\alpha }\sigma_x^2
        + \paren{-\alpha \etay} \sigma_y^2\\
        &\quad
        +\paren{ -2\etax(1+\alpha) - \frac{7\etax}{2} \alpha  }\delta_x^2
        + \paren{-\alpha \etay} \delta_y^2
        \label{potential-stochastic-versus-deterministic-grad-mapping} 
        \\
%%%%%%%%%%%%%%%%%%%%%%%%%%%%%%%%%%%%%%%%%%%%%%%%%%%
        &= \frac{(1+\alpha)\etax}{6}\E \norm{\resx{\Phi,t}}^2
        -
        (1+\alpha)\etax \E \norm{\nabla_x \Phi(x\at{t}) 
                - \nabla_x f(x\at{t},y\at{t}) }^2 \\
        &\quad 
        +
        \paren{
            -\frac{3\etax \alpha}{2} 
            - \alpha\etax^2\etay\ell^2
        }\E\norm{\resx{t}}^2 
        + \frac{\alpha\etay }{6}\E\calD_\calY(y\at{t},\ell;x\at{t}) 
        \\
        &\quad
        +\paren{ -2\etax(1+\alpha) - \frac{9\etax}{2}\alpha - {2 \alpha\etax^2\etay\ell^2} }\sigma_x^2
        + \paren{-\alpha \etay} \sigma_y^2\\
        &\quad
        +\paren{
            -2\etax(1+\alpha)
            - \frac{7\etax}{2} \alpha 
            - {2 \alpha\etax^2\etay \ell^2} }\delta_x^2
        + \paren{-\alpha \etay} \delta_y^2
        \\
%%%%%%%%%%%%%%%%%%%%%%%%%%%%%%%%%%%%%%%%%%%%%%%%%%%
        &\geq \paren{\frac{(1+\alpha)\etax }{6} - 3\etax\alpha  -2 \alpha\etax^2\etay \ell^2}\E \norm{\resx{\Phi,t}}^2
        -
        (1+\alpha)\etax \E \norm{\nabla_x \Phi(x\at{t}) 
                - \nabla_x f(x\at{t},y\at{t}) }^2 \\
        &\quad
        +\paren{ \frac{\alpha\etay }{6}- 3\etax \alpha\kappa^2 - 2 \alpha\etax^2\etay\ell^2\kappa^2 }\E\calD_\calY(y\at{t},\ell;x\at{t}) 
        \\
        &\quad
        +\paren{ -2\etax(1+\alpha) - \frac{9\etax}{2}\alpha - {2 \alpha\etax^2\etay\ell^2} }\sigma_x^2
        + \paren{-\alpha \etay} \sigma_y^2\\
        &\quad
        +\paren{
            -2\etax(1+\alpha)
            - \frac{7\etax}{2} \alpha 
            - {2 \alpha\etax^2\etay \ell^2} }\delta_x^2
        + \paren{-\alpha \etay} \delta_y^2
        \label{plugin-youngs-phi-gradient-mapping-and-ppl}
        \\
    %%%%%%%%%%%%%%%%%%%%%%%%%%%%%%%%%
        &\geq \paren{\frac{(1+\alpha)\etax }{6} - 3\etax\alpha  -2 \alpha\etax^2\etay \ell^2}\E \norm{\resx{\Phi,t}}^2
        -
        (1+\alpha)\etax \kappa^2 \calD_{\calY}(y\at{t},\ell;x\at{t}) \\
        &\quad 
        +\paren{ \frac{\alpha\etay }{6}- 3\etax \alpha\kappa^2 - 2 \alpha\etax^2\etay\ell^2\kappa^2 }\E\calD_\calY(y\at{t},\ell;x\at{t})  
        \\
        &\quad
        +\paren{ -2\etax(1+\alpha) - \frac{9\etax}{2}\alpha - {2 \alpha\etax^2\etay\ell^2} }\sigma_x^2
        + \paren{-\alpha \etay} \sigma_y^2\\
        &\quad
        +\paren{
            -2\etax(1+\alpha)
            - \frac{7\etax}{2} \alpha 
            - {2 \alpha\etax^2\etay \ell^2} }\delta_x^2
        + \paren{-\alpha \etay} \delta_y^2
        \label{ppl-again}
        \\
        %%%%%%%%%%%%%%%%%%%%%
        &=\paren{\frac{(1+\alpha)\etax }{6} - 3\etax\alpha  -2 \alpha\etax^2\etay \ell^2}\E \norm{\resx{\Phi,t}}^2
        % -
        % (1+\alpha)\etax \kappa^2 \calD_{\calY}(y\at{t},\ell;x\at{t}) 
        \\
        &\quad 
        +\paren{\frac{\alpha\etay }{6} - 3\etax \alpha\kappa^2 - 2 \alpha\etax^2\etay\ell^2\kappa^2-(1+\alpha)\etax\kappa^2 }\E\calD_\calY(y\at{t},\ell;x\at{t})  
        \\
        &\quad
        +\paren{ -2\etax(1+\alpha) - \frac{9\etax}{2}\alpha - {2 \alpha\etax^2\etay\ell^2} }\sigma_x^2
        + \paren{-\alpha \etay} \sigma_y^2\\
        &\quad
        +\paren{
            -2\etax(1+\alpha)
            - \frac{7\etax}{2} \alpha 
            - {2 \alpha\etax^2\etay \ell^2} }\delta_x^2
        + \paren{-\alpha \etay} \delta_y^2
    \end{align}

    \begin{itemize}
        \item \eqref{potential-use-lemma-continuity-of-D} uses \cref{lemma:continuity-of-D-wrt-dual-var} and the fact that $|a-b|<c$ implies $a > b -c$, \textit{i.e.,}
        $$\calD_\calY(y\at{t+1},\alpha;x\at{t}) \geq \calD_\calY(y\at{t},\alpha;x\at{t}) - 3\ell^2\norm{x\at{t}- x\at{t+1}}^2,$$
        \item \eqref{potential-stochastic-versus-deterministic-grad-mapping} uses the definition of $\sresx{t}, \resx{t}$ and \cref{lemma:stochastic-vs-deterministic-grad-mapping} to replace the term: $\norm{x\at{t}- x\at{t+1}}^2$, \textit{i.e,}:
        \begin{align}
            \norm{x\at{t}-x\at{t+1}}^2 
            &\leq 2\norm{x\at{t}-\bar{x}\at{t+1}}^2 + 2\norm{x\at{t+1}-\bar{x}\at{t+1}}^2 \\
            &\leq 2\etax^2 \norm{\resx{t}}^2 + 2\etax^2\norm{ \hatgradfbx(x\at{t},y\at{t}) - \nabla_x f(x\at{t},y\at{t}) }^2\\
            &\leq 2\etax^2 \norm{\resx{t}}^2 + 4\etax^2\sigma_x^2 + 4\etax^2\delta_x^2.
        \end{align}
        \item By the fact that $\|c - d \|^2 = \|c\|^2 + \|d\|^2 - 2\inprod{c}{d}$ and Young's inequality, we can write $\|a\|^2=\|(a-b) + b\|^2 \leq (1 + 1/\rho)\|a\|^2 + (1 + \rho)\|b\|^2$. Then we plug in $a\gets \resx{t}$, $b\gets \resx{\Phi,t}$, 
        \begin{align}
            \norm{\resx{t}}^2 &\leq \paren{1 + \rho }\norm{\resx{\Phi,t}}^2 + \paren{1 + \frac{1}{\rho}} \norm{\frac{1}{\etax}\paren{\proj_{\calX}\paren{x-\etax \nabla f(x\at{t},y\at{t}) } - \proj_{\calX}\paren{x - \etax \nabla\Phi(x\at{t}) }  } }^2\\
            &\leq \paren{1  +\rho}\norm{\resx{\Phi,t}}^2 
            + \paren{1 + \frac{1}{\rho}} \norm{ \nabla f(x\at{t},y\at{t}) -  \nabla\Phi(x\at{t}) }^2\\
            &\leq \paren{1 + \rho }\norm{\resx{\Phi,t}}^2 
            + \paren{1 + \frac{1}{\rho}}\ell^2 \norm{ y^{\star}(x\at{t}) - y\at{t} }^2\\
            &\leq \paren{1  +\rho }\norm{\resx{\Phi,t}}^2 
            + \paren{1  +\frac{1}{\rho} }\frac{\ell^2}{\muqg} \paren{\Phi(x\at{t}) - f(x\at{t},y\at{t}) }\\
            &\leq \paren{1  +\rho }\norm{\resx{\Phi,t}}^2 
            + \paren{1  +\frac{1}{\rho} }\kappa^2 \calD_{\calY}(y\at{t},\ell;x\at{t})
        \end{align}
        Where the third to last inequality follows from Danskin's lemma, the penultimate is due to the quadratic growth property, and the last is due to the proximal-PL property. We have set $\kappa:= \frac{\ell}{\sqrt{\muqg\mu}}$.
        \item \eqref{ppl-again} follows from the pPL property.
        \begin{align}
        \norm{\resx{t}}^2 
        &\leq 2 \norm{\resx{\Phi,t}}^2 + 2\norm{\resx{t} - \resx{\Phi,t} } ^2\\
        &\leq 2 \norm{\resx{\Phi,t}}^2 + 2\norm{\nabla f(x\at{t},y\at{t}) - \nabla \Phi(x\at{t}) } ^2\\
        &\leq 2 \norm{\resx{\Phi,t}}^2 + 2\ell \norm{y\at{t} - y^\star(x\at{t}) } ^2\\
        &\leq 2 \norm{\resx{\Phi,t}}^2 + 2\ell \kappa^2 \calD_\calY(y\at{t},\ell;x\at{t})
    \end{align}
    \end{itemize}

    Re-iterating,
    \begin{align}
        \E V\at{t}- \E V\at{t+1}
        &\geq 
        \paren{\frac{(1+\alpha)\etax }{6} - 3\etax\alpha  -2 \alpha\etax^2\etay \ell^2}\E \norm{\resx{\Phi,t}}^2
        \\
        &\quad 
        +\paren{\frac{\alpha\etay }{6} - 3\etax \alpha\kappa^2 - 2 \alpha\etax^2\etay\ell^2\kappa^2-(1+\alpha)\etax\kappa^2 }\E\calD_\calY(y\at{t},\ell;x\at{t})
        \\
        &\quad
        +\paren{ -2\etax(1+\alpha) - \frac{9\etax}{2}\alpha - {2 \alpha\etay\ell^2} }\sigma_x^2
        + \paren{-\alpha \etax^2 \etay} \sigma_y^2\\
        &\quad
        +\paren{
            -2\etax(1+\alpha)
            - \frac{7\etax}{2} \alpha 
            - {2 \alpha\etax^2\etay \ell^2} }\delta_x^2
        + \paren{-\alpha \etay} \delta_y^2
    \end{align}
    What is left to do is to ensure that the coefficients of $\norm{\resx{\Phi,t}}^2,\calD_\calY(y\at{t},\ell;x\at{t})$
    in the last display are positive. We will show that this is indeed the case for a correct tuning of $\etax,\etay$ and $\alpha$.
    \begin{itemize}
        \item \textbf{Coefficient of {$\norm{\resx{\Phi,t}}^2$}.}
    We assume \textit{a priori} that, $\etax,\etay \leq \frac{1}{\ell}$. Hence we can write:
    \begin{align}
        {\frac{(1+\alpha)\etax }{6} - 3\etax\alpha  -2 \alpha\etax^2\etay \ell^2}
       \geq
        \frac{1+\alpha - 18\alpha - 12\alpha }{6}\etax 
        \geq \frac{1 - 29\alpha}{6}\etax\geq \frac{\etax}{180}
    \end{align}
    Where we require that $\alpha \leq 1/30 < 1/29.$
    \item \textbf{Coefficient of $\calD_\calY$.} We require that
    $\etax\leq \frac{\etay}{c\kappa^2}$ for some $c>1$.
    \begin{align}
        {\frac{\alpha\etay }{6} - 3\etax \alpha\kappa^2 - 2 \alpha\etax^2\etay\ell^2\kappa^2-(1+\alpha)\etax\kappa^2 } 
        &\geq \frac{\alpha  - 18 \alpha/c - 12 \alpha/c^2 - 6(1+\alpha)/c }{6} \etay\\
        % &=\frac{c^2\alpha - 18c\alpha - 12\alpha - 6c - 6c\alpha }{6c^2}\etay\\
        &=\frac{(c^2 - 24c - 12)\alpha - 6c }{6c^2}\etay\\
        &\geq \frac{\etay}{10}
    \end{align}
    Where in the inequality we assume $c=100, \alpha=1/30$.
    \item  \textbf{Coefficients of $\sigma_x^2, \sigma_y^2, \delta_x, \delta_y$.}
    \begin{itemize}
        \item For $\sigma_x^2, \delta_x^2$
    
        \begin{align}
        2\etax(1+\alpha) + \frac{7\etax}{2}\alpha + {2 \alpha\etax^2\etay\ell^2} 
        &\leq
            { 2\etax(1+\alpha) + \frac{9\etax}{2}\alpha + {2 \alpha\etax^2\etay\ell^2} }
        \\&= { \etax\frac{60 + 2 + 135}{30} + \frac{200\etax^3\kappa^2\ell^2}{30} } & \quad \paren{\etax = \frac{\etay}{100\kappa^2}}
        \\&= { \etax\frac{60 + 2 + 135}{30} + \frac{200\etax^2\ell}{30\cdot 500} } & \quad \paren{\etax = \frac{1}{500\kappa^2\ell}}
        \\&\leq { \etax\frac{60 + 2 + 135}{30} + \frac{200\etax}{30\cdot 500} } \leq 8\etax & \quad \paren{\etax \leq 1/\ell \implies \ell\leq 1/\etax}.
        \end{align}
    
        % where we assume that $\ell\geq 1$. We also use the fact that $\kappa \geq 1$ and $a+b\leq 2ab$ when $a,b\geq 1$.
        \item For $\sigma_y^2, \delta_y^2$:

        \begin{align}
            \alpha\etay \leq \frac{100\kappa^2}{30}\etax\leq 4\kappa^2\etax
        \end{align}
        \end{itemize}
    \end{itemize}

\begin{align}
    \norm{\resx{t}}^2  \leq 2\norm{\resx{\Phi,t}}^2 
            + 2\kappa^2 \calD_{\calY}(y\at{t},\ell;x\at{t})
\end{align}
Summarizing:
\begin{align}
    \frac{\etax}{360} \norm{\resx{t}}^2 + 9\kappa^2\etax\E \calD_{\calY}(y\at{t},\ell;x\at{t}) &\leq 
    \frac{\etax}{360} \norm{\resx{t}}^2 +\paren{ \frac{\etay}{10} - \frac{2\kappa^2\etax}{360}}\E \calD_{\calY}(y\at{t},\ell;x\at{t})\\
    &\leq \E V\at{t} - \E V\at{t+1} + 8\etax\sigma_x^2 + 8\etax\delta_x^2 + 4\kappa^2\etax \sigma_y^2 + 4\kappa^2\etax  \delta_y^2. 
\end{align}
Summing over $t$ and dividing by $T$:
% \begin{align}
%     \frac{1}{T}\sum_{t=1}^{T-1} 
%     \paren{\norm{\resx{t}}^2 +  \kappa^2\E \calD_{\calY}(y\at{t},\ell;x\at{t})} \leq c_1 \frac{\max_{x,y} V(x,y)- \min_{x,y}V(x,y)}{\etax T} + c_2 {\kappa^2\ell\paren{\sigma_x^2+ \delta_x^2 }} + c_3{\kappa^2\paren{\sigma_y^2 + \delta_y^2}}
% \end{align}
% where $c_1,c_2,c_3\in O(1)$. Let $\Delta_V := \max_{x,y}V(x,y) - \min_{x,y}V(x,y)\leq2 L (\diam{\calX} + \diam{\calY})$
\begin{align}
    \frac{1}{T}\sum_{t=1}^{T-1} 
    \paren{\norm{\resx{t}}^2 +  \kappa^2\E \calD_{\calY}(y\at{t},\ell;x\at{t})} \leq \frac{720 L (\diam{\calX} + \diam{\calY}) }{\etax T} + 2880 {\paren{\sigma_x^2+ \delta_x^2 }} + 1440{\kappa^2\paren{\sigma_y^2 + \delta_y^2}}.
\end{align}

% \fk{the following are leaving}
% Tuning:
% \begin{align}
%     \frac{1}{T}\sum_{t=1}^{T-1} 
%     \paren{\norm{\resx{t}}^2 +  \kappa^2\E \calD_{\calY}(y\at{t},\ell;x\at{t})} \leq \frac{720 L (\diam{\calX} + \diam{\calY}) }{\etax T} + 2880 {\paren{\sigma_x^2+ \delta_x^2 }} + 1440{\kappa^2\paren{\sigma_y^2 + \delta_y^2}} = \frac{\epsilon^2}{2\mu\kappa^2}
% \end{align}

% % $$\epsilon\geq \Phi(x) - f(x,y)$$

% % $$\epsilon \mu \geq \mu(\Phi(x) - f(x,y))$$

% % $$\frac{1}{2}\calD \leq  \epsilon \mu$$

% % $$\kappa^2\calD \leq  2\kappa^2\mu\epsilon = 2 \kappa \ell\epsilon$$

% Desired accuracy: $$\paren{\norm{\resx{t}}^2 +  \kappa^2\E \calD_{\calY}(y\at{t},\ell;x\at{t})} < \epsilon $$

% $$2 \kappa\ell\geq 1$$

% \begin{itemize}
%     \item $$\etax T = \frac{2\cdot 5 \cdot 720 L \paren{\diam{\calX} + \diam{\calY}} }{\epsilon^2}$$
%     \item $$ M_x = ?? $$
%     \item $$\delta_x = ?? $$
%     \item $$M_y = \frac{2\cdot 5\cdot 1440 \kappa^2\sigma_y^2}{\epsilon^2} $$
%     \item $$\delta_y = \frac{\epsilon}{\sqrt{2\cdot 5\cdot 1440}\kappa}$$
%     \item $$\etax = \frac{1}{500\ell\kappa^2}\leq \frac{1}{5\ellphi}$$
%     \item $$\etay = \frac{1}{5\ell}$$
% \end{itemize}

% \begin{align}
%     \frac{1}{T}\sum_{t=1}^{T-1} 
%     \paren{\norm{\resx{t}}^2 +  \kappa^2\E \calD_{\calY}(y\at{t},\ell;x\at{t})} \leq \frac{720 \cdot 500 \kappa^2\ell L (\diam{\calX} + \diam{\calY}) }{ T} + 4\epsilon^2/5
% \end{align}

% Then, $T = \frac{ 720 \cdot 500 \kappa^2\ell L \paren{\diam{\calX}  +\diam{\calY} }}{\epsilon^2/5}.$
\end{proof}

% \newpage
\subsubsection{Two-Sided pPL Min-Max Optimization}
\begin{theorem}\label{theorem:duality-gap:Alt-GDA:PL-PL}
Let $f:\calX\times\calY\to \R$ be an $L$-Lipschitz, $\ell$-smooth function and $\calX,\calY$ be two compact convex sets with Euclidean diameters $\diam{\calX},\diam{\calY}$ respectively. Further, assume that $f(\cdot,y)$ satisfies the pPL condition with modulus $\mux$ for all $y\in\calY$ while $f(x,\cdot)$ satisfies the pPL condition with a modulus $\muy$ for any $x\in\calX$. Additionally, let $(\hatgradfbx, \hatgradfby)$ be an inexact stochastic gradient oracle satisfying \cref{assumption:stoch-grad-oracle}. Then, after $T$ iterations of \eqref{agda} with a tuning of step-sizes $\etax = \frac{\muy^2}{160\ell^3}$ and $\etay = \frac{1}{5\ell}$, it holds true that:
\begin{align}
    \E \Phi(x\at{T}) - \Phi^\star + 1/10\paren{\E\Phi(x\at{T}) - \E f(x\at{T},y\at{T}) }\leq \exp\paren{- \frac{\mux\muy^2}{160\ell^3} T} L(\diam{\calX}+\diam{\calY}) + \frac{c_1\sigma_x^2}{\mux} + \frac{c_1 \delta_x^2}{\mux} 
   + \frac{c_2\ell^2 \sigma_y^2}{\mux\muy^2} + \frac{c_2\ell^2\delta_y^2}{\mux\muy^2},
\end{align}
where, $c_1,c_2 \in O(1)$.
\label{theorem:formal-agda-2sided-ppl}
\end{theorem}
\begin{proof}
Our goal is to ultimately demonstrate that there exists a Lyapunov function, $V\at{t}$ whose value decreases along any trajectory of the algorithm's iterates. Not only so, but its value contracts, \textit{i.e.}, there exists $0<\varpi<1$ such that $V\at{t+1} \leq \varpi V\at{t},~\forall t$. To demonstrate this, we first need to lower bound the descent on $\Phi(\cdot)$, lower bound the ascent on $f(x\at{t},\cdot)$, and finally \emph{upper bound} the descent on $f(\cdot,y\at{t})$.
\\
\item 
\paragraph{Descent Bound on \texorpdfstring{$\Phi(\cdot)$}{Φ(·)}.}
By \cref{D-descent-lemma} and \cref{lemma:ppl-phi} we write,
\begin{align}
    \E\Phi(x\at{t+1})
    &\leq \E \Phi(x\at{t}) 
    - \frac{\etax}{6} \E \calD_{\calX}^{\Phi}(x\at{t},1/\etax) 
    + \etax \E \norm{\nabla_x \Phi(x\at{t}) - \nabla_x f(x\at{t},y\at{t}) }^2 \\&\quad
    + 2\etax \sigma_x^2 + 2 \etax \delta_x^2\\
    &\leq \E \Phi(x\at{t}) 
    - \frac{\etax \mux}{3}\paren{\Phi(x\at{t}) - \Phi^\star}
    + \etax \E \norm{\nabla_x \Phi(x\at{t}) - \nabla_x f(x\at{t},y\at{t}) }^2 \\&\quad
    + 2\etax \sigma_x^2 + 2 \etax \delta_x^2
\end{align}
Hence,
\begin{align}
    \E\Phi(x\at{t+1}) - \Phi^\star &\leq \paren{1
    - \frac{\etax \mux}{3} }\paren{\Phi(x\at{t}) - \Phi^\star}
    + \etax \ell^2 \norm{y^\star - y\at{t} }^2 \\&\quad
    + 2\etax \sigma_x^2 + 2 \etax \delta_x^2\\
    &\leq \paren{1
    - \frac{\etax \mux}{3} }\paren{\Phi(x\at{t}) - \Phi^\star}
    + \etax \ell^2 \muqg/2 \paren{\Phi(x\at{t}) - f(x\at{t},y\at{t})} \\&\quad
    + 2\etax \sigma_x^2 + 2 \etax \delta_x^2,
    \label{2sided-pl-ineq1}
\end{align}
or, we can also write,
\begin{align}
\E\Phi(x\at{t+1}) - \E\Phi(x\at{t}) &\leq - \frac{\etax \mux}{3}\paren{\Phi(x\at{t}) - \Phi^\star}
    + \etax \ell^2 \muqg/2 \paren{\Phi(x\at{t}) - f(x\at{t},y\at{t})}
    \\&\quad
    + 2\etax \sigma_x^2 + 2 \etax \delta_x^2  
    \label{phi-plus-phi}
\end{align}

\item \paragraph{Ascent Bound on \texorpdfstring{$f(x\at{t},\cdot)$}{f(x\at{t},·)}.}  A simple application of \cref{D-descent-lemma} yields,
\begin{align}
    \E f(x\at{t+1},y\at{t+1})\geq \E f(x\at{t+1},y\at{t}) + \frac{\etay}{6}\E\calD_\calY(y\at{t},1/\etay;x\at{t+1}) - \etay\delta_y^2 - \etay\sigma_y^2.
\end{align}

From which, we can also write,
\begin{align}
    &\E \Phi(x\at{t+1}) - \E f(x\at{t+1},y\at{t+1})\\
    \quad&\leq \E\Phi(x\at{t+1}) -  \E f(x\at{t+1},y\at{t}) - \frac{\etay}{6}\E\calD_\calY(y\at{t},1/\etay;x\at{t+1}) + \etay\delta_y^2 + \etay\sigma_y^2
    \\
    &\leq \paren{1 - \frac{\muy\etay}{6}}\paren{\E\Phi(x\at{t+1}) -  \E f(x\at{t+1},y\at{t})} + \etay\delta_y^2 + \etay\sigma_y^2
    \\
    &= \paren{1 - \frac{\muy\etay}{6}}\paren{
            \E\Phi(x\at{t})
            - \E f(x\at{t},y\at{t}) 
            + \E f(x\at{t},y\at{t})
            -  \E f(x\at{t+1},y\at{t}) 
            +  \E\Phi(x\at{t+1}) 
            - \E\Phi(x\at{t})}
        \\ &\quad + \etay\delta_y^2 + \etay\sigma_y^2
    \label{2ppl-phi-minus-f-intermediate}
\end{align}

\item \paragraph{Upper bound on the descent on \texorpdfstring{$f(\cdot,y\at{t})$}{f(·,y\at{t})}}
From \eqref{eq:upper-bound-descent} we write,
\begin{align}
    \E f(x\at{t+1}, y\at{t}) &\geq \E f(x\at{t}, y\at{t})
        - \frac{3\etax}{2} \E \norm{\resx{t}}^2 
         - \frac{9\etax}{2} \sigma_x^2
         -\frac{7\etax}{2}\delta_x^2\\
    &\geq \E f(x\at{t}, y\at{t})
        - \frac{3\etax}{2} \calD_{\calX}(x\at{t},1/\etax;y\at{t})
         - \frac{9\etax}{2} \sigma_x^2
         -\frac{7\etax}{2}\delta_x^2
     \label{2ppl-f-descent-upper-bound}
\end{align}
The second inequality follows \cref{lemma:D-vs-grad-mapping}. Now, \eqref{2ppl-phi-minus-f-intermediate} by plugging-in \eqref{phi-plus-phi} and \eqref{2ppl-f-descent-upper-bound} reads,
\begin{align}
    &\E \Phi(x\at{t+1}) - \E f(x\at{t+1},y\at{t+1})
    \leq \paren{1 - \frac{\muy\etay}{6}}\paren{\E\Phi(x\at{t})- \E f(x\at{t},y\at{t})} 
    \\
    &\quad+ \paren{1 - \frac{\muy\etay}{6}}\paren{ \E f(x\at{t},y\at{t}) -  \E f(x\at{t+1},y\at{t})} \\
    &\quad+ \paren{1 - \frac{\muy\etay}{6}}\paren{  \E\Phi(x\at{t+1}) - \E\Phi(x\at{t})}\\ &\quad + \etay\delta_y^2 + \etay\sigma_y^2\\
    \quad
    %%%%%%%%%%%%%%%%%%%%%%%%%%%%%%%%%%
    &\leq \paren{1 - \frac{\muy\etay}{6}}\paren{\E\Phi(x\at{t})- \E f(x\at{t},y\at{t})} \\
    &\quad + \paren{1 - \frac{\muy\etay}{6}}\paren{ \frac{3\etax}{2} \calD_{\calX}(x\at{t},1/\etax;y\at{t})
        + \frac{9\etax}{2} \sigma_x^2
        + \frac{7\etax}{2}\delta_x^2 } \\
    &\quad+ \paren{1 - \frac{\muy\etay}{6}}\paren{ 
        - \frac{\etax}{6} \E \calD_{\calX}^{\Phi}(x\at{t},1/\etax) 
        + \etax \E \norm{\nabla_x \Phi(x\at{t}) - \nabla_x f(x\at{t},y\at{t}) }^2 
        + 2\etax \sigma_x^2 + 2 \etax \delta_x^2
    }
    \\ &\quad
    + \etay\sigma_y^2 + \etay\delta_y^2 
    \\
    %%%%%%%%%%%%%%%%%%%%%%%%%%%%%%%%%%
    &\leq \paren{1 - \frac{\muy\etay}{6}}\paren{\E\Phi(x\at{t})- \E f(x\at{t},y\at{t})} \\
    &\quad + \paren{1 - \frac{\muy\etay}{6}} { \frac{3\etax}{2} \calD_{\calX}(x\at{t},1/\etax;y\at{t})
         } \\
    &\quad+ \paren{1 - \frac{\muy\etay}{6}}\paren{ 
        - \frac{\etax}{6} \E \calD_{\calX}^{\Phi}(x\at{t},1/\etax) 
        + \etax \E \norm{\nabla_x \Phi(x\at{t}) - \nabla_x f(x\at{t},y\at{t}) }^2 
    }
    \\ &\quad
    + 7 \paren{1 - \frac{\muy\etay}{6}} \etax \sigma_x^2 + 7 \paren{1 - \frac{\muy\etay}{6}} \etax \delta_x^2 +\etay\sigma_y^2 + \etay\delta_y^2 
 \label{2sided-pl-ineq2}
\end{align}
\item \paragraph{The Lyapunov Function.}
We will consider the following Lyapunov function where $\alpha>0$ is to be defined later along the proof:
\begin{align}
    V(x\at{t}, y\at{t}):= (\Phi(x\at{t}) - \Phi^\star) + \alpha \paren{\Phi(x\at{t}) - f(x\at{t},y\at{t})}.
\end{align}
For convenience, with $U\at{t}, W\at{t}$ we will denote the quantities:
\begin{align}
    U\at{t} &:=  \Phi(x\at{t}) - \Phi^\star\\
    W\at{t} &:=  \Phi(x\at{t}) -  f(x\at{t},y\at{t}),
\end{align}
and we can then write $V\at{t} = U\at{t} + \alpha W\at{t}$. Piecing together \eqref{2sided-pl-ineq1}, and \eqref{2sided-pl-ineq2}
\begin{align}
    U\at{t+1} + \alpha W\at{t+1}
    &\leq U\at{t}  
    - \frac{\etax}{6} \E \calD_{\calX}^{\Phi}(x\at{t},1/\etax) 
    + \etax \E \norm{\nabla_x \Phi(x\at{t}) - \nabla_x f(x\at{t},y\at{t}) }^2 \\
    &\quad + 2\etax \sigma_x^2 + 2 \etax \delta_x^2\\
    &\quad
    +\alpha{
        \paren{1 - \frac{\muy\etay}{6}}W\at{t}
        % \paren{\E\Phi(x\at{t})- \E f(x\at{t},y\at{t})}
         }
    \\
    &\quad+ \alpha { \paren{1 - \frac{\muy\etay}{6}} { \frac{3\etax}{2} \calD_{\calX}(x\at{t},1/\etax;y\at{t})
             } 
    }\\
    &\quad +\alpha 
    {
     \paren{1 - \frac{\muy\etay}{6}}\paren{ 
        - \frac{\etax}{6} \E \calD_{\calX}^{\Phi}(x\at{t},1/\etax) 
        + \etax \E \norm{\nabla_x \Phi(x\at{t}) - \nabla_x f(x\at{t},y\at{t}) }^2 
    }
    }\\
    &\quad +\alpha \paren{7 \paren{1 - \frac{\muy\etay}{6}} \etax \sigma_x^2 + 7 \paren{1 - \frac{\muy\etay}{6}} \etax \delta_x^2 +\etay\sigma_y^2 + \etay\delta_y^2 }
    \\
    %%%%%%%%%%%%%%%%%%%%%%%%%%%%%%%%%%%%%%
    &\leq U\at{t} +  
    \paren{ - \frac{\etax}{6} -\frac{\alpha\etax}{6} \paren{1 - \frac{\muy\etay}{6}} + \alpha \paren{1 - \frac{\muy\etay}{6}}\frac{3\etax}{2}
    } \E \calD_{\calX}^{\Phi}(x\at{t},1/\etax)  \\
    &\quad 
    + \paren{\etax + \alpha \etax \paren{1 - \frac{\muy\etay}{6}} } \E \norm{\nabla_x \Phi(x\at{t}) - \nabla_x f(x\at{t},y\at{t}) }^2 \\
    &\quad 
    +\alpha
        \paren{1 - \frac{\muy\etay}{6}}W\at{t}
    \\
    &\quad + \alpha { \paren{1 - \frac{\muy\etay}{6}} { \frac{3\etax}{2}
        {        \left|\calD_\calX(x\at{t},1/\etax;y\at{t}) - \calD_\calX^\Phi(x\at{t},1/\etax)\right|
        }
             } 
    }\\
    &\quad +\paren{2 + 7\alpha  \paren{1 - \frac{\muy\etay}{6}}} \etax \sigma_x^2 + \paren{2 + 7\alpha  \paren{1 - \frac{\muy\etay}{6}}} \etax \delta_x^2 +\alpha\etay\sigma_y^2 + \alpha\etay\delta_y^2 
    \label{2ppl-potential-difference-Phi-f-and-triangle}
    \\
%%%%%%%%%%%%%%%%%%%%%%%%%%%%%%%%%%%%%%
    &\leq U\at{t} +  
    \paren{ - \frac{\etax}{6} -\frac{\alpha\etax}{6} \paren{1 - \frac{\muy\etay}{6}} + \alpha \paren{1 - \frac{\muy\etay}{6}}\frac{3\etax}{2}
    } \E \calD_{\calX}^{\Phi}(x\at{t},1/\etax)  \\
    &\quad 
    + \paren{\etax + \alpha \etax \paren{1 - \frac{\muy\etay}{6}} } \frac{2\ell^2}{\muqg}  \paren{\E \Phi(x\at{t}) - \E f(x\at{t},y\at{t}) }\\
    &\quad
    +\alpha
        \paren{1 - \frac{\muy\etay}{6}}W\at{t}
    \\
    &\quad+ \alpha { \paren{1 - \frac{\muy\etay}{6}} \frac{9\etax\ell^2}{\muqg}
        \paren{\E \Phi(x\at{t}) - \E f(x\at{t},y\at{t}) }
    }\\
    &\quad +\paren{2 + 7\alpha  \paren{1 - \frac{\muy\etay}{6}}} \etax \sigma_x^2 + \paren{2 + 7\alpha  \paren{1 - \frac{\muy\etay}{6}}} \etax \delta_x^2 +\alpha\etay\sigma_y^2 + \alpha\etay\delta_y^2 
    \label{2ppl-potential-qg}
    \\
%%%%%%%%%%%%%%%%%%%%%%%%%%%%%%%%%%%%%%
    &\leq 
    \paren{1  - \frac{\mux}{2} \paren{\frac{\etax}{6} -\frac{\alpha\etax}{6} \paren{1 - \frac{\muy\etay}{6}} + \alpha \paren{1 - \frac{\muy\etay}{6}}\frac{3\etax}{2}
    } } U\at{t}  \\
    &\quad 
    + \paren{\etax + \alpha \etax \paren{1 - \frac{\muy\etay}{6}} } \frac{2\ell^2}{\muqg} W\at{t} \\
    &\quad
    +\alpha
        \paren{1 - \frac{\muy\etay}{6}}W\at{t}
    \\
    &\quad+ \alpha { \paren{1 - \frac{\muy\etay}{6}} \frac{9\etax\ell^2}{\muqg}
        W\at{t}
    }\\
    &\quad +\paren{2 + 7\alpha  \paren{1 - \frac{\muy\etay}{6}}} \etax \sigma_x^2 + \paren{2 + 7\alpha  \paren{1 - \frac{\muy\etay}{6}}} \etax \delta_x^2 +\alpha\etay\sigma_y^2 + \alpha\etay\delta_y^2 
    \label{use-ppl-of-phi}
\end{align}

\begin{itemize}
    \item In \eqref{2ppl-potential-difference-Phi-f-and-triangle} we use the fact $a \leq |a - b| + b$ with $a = \calD_\calX(x\at{t},1/\etax;y\at{t}), b = \calD_\calX^\Phi(x\at{t},1/\etax)$ in order to remove the term $ \calD_\calX(x\at{t},1/\etax;y\at{t})$ and introduce the terms $\left|\calD_\calX(x\at{t},1/\etax;y\at{t}) - \calD_\calX^\Phi(x\at{t},1/\etax) \right|$ and $\calD_\calX^\Phi(x\at{t},1/\etax)$;
    \item In \eqref{2ppl-potential-qg}, we use the following facts:
    \begin{enumerate}[]
        \item $\left|\calD_\calX(x\at{t},1/\etax;y\at{t}) - \calD_\calX^\Phi(x\at{t},1/\etax)\right| \leq 3\ell^2 \norm{y\at{t}- y^\star(x\at{t})}^2$, \cref{lemma:continuity-of-D-wrt-dual-var};
        \item $\norm{\nabla_x \Phi(x\at{t}) - \nabla_x f(x\at{t},y\at{t}) }^2\leq \ell^2 \norm{y\at{t} - y^\star(x\at{t})}^2$;
        \item $\frac{\muqg}{2}\norm{y\at{t}-y^\star(x\at{t})}^2 \leq \Phi(x\at{t}) - f(x\at{t},y\at{t})$, quadratic growth condition of pPL functions \cref{prop:qg-from-ppl}.
    \end{enumerate}
    \item In \eqref{use-ppl-of-phi} we use the fact that $\frac{1}{2}\calD_\calX^{\Phi}(x,1/\etax) \geq \mux U\at{t}$ (\cref{lemma:ppl-phi}) 
    and the negativity of the coefficient of 
    $\E \calD_{\calX}^{\Phi}(x\at{t},1/\etax) $,  \textit{i.e.},
    the display $\paren{ - \frac{\etax}{6} -\frac{\alpha\etax}{6} \paren{1 - \frac{\muy\etay}{6}} + \alpha \paren{1 - \frac{\muy\etay}{6}}\frac{3\etax}{2}
    }.$
    To ensure negativity, we require that $\alpha <1$ and utilize the fact that $\mux \leq \ell$ and hence $\etax \mux \leq 1$ by the choice of step-size.
    % \begin{enumerate}
    %     \item 
    % \item Inequality \eqref{use-ppl-of-phi} uses the pPL property of function $\Phi$ with modulus $\mux$ .
    % \end{enumerate}
\end{itemize}
Summarizing:
\begin{align}
    U\at{t+1} + \alpha W\at{t+1} &\leq
    \underbrace{\paren{1 - \frac{\mux\etax}{3} + \frac{\alpha \mux\etax}{3} \left(1 - \frac{\muy \etay}{6}\right) }}_{\varpi_1} U\at{t}\\
    &\quad+ \alpha 
    \underbrace{\paren{
    1 + \frac{11 \ell^2 \etax}{\muqg} - \frac{11 \muy \ell^2 \etax\etay}{6\muqg} - \frac{\muy \etay}{6} 
    + \frac{2 \ell^2 \etax}{\alpha\muqg}}}_{\varpi_2} W\at{t}\\
    &\quad +\underbrace{\paren{2 + 7\alpha  \paren{1 - \frac{\muy\etay}{6}}} \etax \sigma_x^2 + \paren{2 + 7\alpha  \paren{1 - \frac{\muy\etay}{6}}} \etax \delta_x^2 +\alpha\etay\sigma_y^2 + \alpha\etay\delta_y^2}_{\xi} 
\end{align}
\item \paragraph{Tuning the parameters:} We need to ensure that $0< \varpi_1, \varpi_2 <1$ and then, we can show that the value of the Lyapunov function is contracting, \textit{i.e.} $U\at{t+1}+\alpha W\at{t+1} \leq \max\{\varpi_1, \varpi_2\}\paren{U\at{t} + \alpha W\at{t}} + \xi$. We wiil now upper-bound $\varpi_1, \varpi_2$. For $\varpi_1$ we see that:
\begin{align}
    1 - \frac{\mux\etax}{3} + \frac{\alpha\mux\etax}{3}\paren{1 - \frac{\muy\etay}{6}} \leq 1 - \frac{\etax\mux}{3} + \frac{\alpha\mux\etax}{3} =1 - \frac{29\mux \etax}{30} \leq 1 - {\mux\etax}.
\end{align}
For $\varpi_2$,
% \begin{align}
% &\paren{\etax + \alpha \etax \paren{1 - \frac{\muy\etay}{6}} } \frac{2\ell^2}{\muqg} W\at{t}
%     +\alpha
%         \paren{1 - \frac{\muy\etay}{6}}W\at{t}
%     + \alpha  \paren{1 - \frac{\muy\etay}{6}} \frac{9\etax\ell^2}{\muqg}
%         W\at{t}
%         \\
%         &=\quad\paren{
%         \paren{\etax + \alpha \etax \paren{1 - \frac{\muy\etay}{6}} } \frac{2\ell^2}{\muqg} 
%         +\alpha
%         \paren{1 - \frac{\muy\etay}{6}}
%     + \alpha  \paren{1 - \frac{\muy\etay}{6}} \frac{9\etax\ell^2}{\muqg}
%         }W\at{t} \\
%         %%%%%%%%%%%%%
%         &=\quad \alpha \paren{1 - \frac{\muy\etay}{6} + \frac{2\etax\ell^2}{\alpha \muqg} + 
%         {  \etax \paren{1 - \frac{\muy\etay}{6}} } \frac{2\ell^2}{\muqg} 
%     +  \paren{1 - \frac{\muy\etay}{6}} \frac{9\etax\ell^2}{\muqg}
%         }W\at{t}
% \end{align}
\begin{align}
    &1 - \frac{\muy\etay}{6} + \frac{2\etax\ell^2}{\alpha \muqg} + 
        {  \etax \paren{1 - \frac{\muy\etay}{6}} } \frac{2\ell^2}{\muqg} 
    +  \paren{1 - \frac{\muy\etay}{6}} \frac{9\etax\ell^2}{\muqg}\\
    &=1 - \frac{\ell^2\etax}{\muy}\paren{\frac{\muy^2\etay}{\etax\ell^2} - \frac{2\muy}{\alpha}
    -\paren{1 - \frac{\muy\etay}{6}}11
    } \\
    &\leq 
    1 - \frac{\ell^2\etax}{\muy}\paren{\frac{\muy^2\etay}{\etax\ell^2} - \frac{2}{\alpha}
    -11
    }\\
    &\leq 
    1 - \frac{\ell^2\etax}{\muy}\paren{\frac{\muy^2\etay}{\etax\ell^2} - 31
    }\\
    &\leq 
    1 - \frac{\ell^2\etax}{\muy},
\end{align}
where we have used the fact that $\muqg$ is equal to $\muy$ by \cref{prop:qg-from-ppl} and we require that $\frac{\muy^2\etay}{\etax\ell^2} = 32$. As such, we $\etax = \frac{\muy^2\etay}{32\ell^2}$ and $\etay=\frac{1}{5\ell}$. Further, We observe that $1 - \mux \etax = \max\left\{ 1- \frac{\ell^2\etax}{\muy} , 1- \mux\etax \right\} $ since $\ell\geq \mux,\muy$. Combining the pieces:
\begin{align}
    \E V\at{t+1}\leq \paren{ 1 - \mux \etax } \E V\at{t} + 3\etax \sigma_x^2 + 3 \etax \delta_x^2 + \etay \sigma_y^2 + \etay\delta_y^2
\end{align}
Applying the inequality recursively and using the formula for the sum of the geometric sequence, we get,
% \begin{align}
%   \E V\at{T}\leq \paren{ 1 - \mux \etax }^T V\at{0} + \frac{3\sigma_x^2}{\mux} + \frac{3 \delta_x^2}{\mux} + \frac{\etay \sigma_y^2}{\etax\mux} + \frac{\etay\delta_y^2}{\mux\etax}.
% \end{align}
% $\frac{\etay}{\etax} = \frac{32\ell^2}{\muy^2} = 32\kappa^2$
\begin{align}
   \E V\at{T}&\leq \paren{ 1 - \mux \etax }^T V\at{0} + \frac{3\sigma_x^2}{\mux} + \frac{3 \delta_x^2}{\mux} + + \frac{32\ell^2 \sigma_y^2}{\mux\muy^2} + \frac{32\ell^2\delta_y^2}{\mux\muy^2} \\
   &\leq \exp\paren{-\mux\etax T} V\at{0} 
   + \frac{3\sigma_x^2}{\mux} + \frac{3 \delta_x^2}{\mux} 
   + \frac{32\ell^2 \sigma_y^2}{\mux\muy^2} + \frac{32\ell^2\delta_y^2}{\mux\muy^2} 
\end{align}
Further, we note that since $f,\Phi$ are $L$-Lipschitz continuous and their domains have bounded diameters, we can bound $V\at{0}$ as $V\at{0}\leq L \paren{\diam{\calX}+ \diam{\calY}}.$ Finally, we see that by the choice of step-sizes $\mux \etax = \frac{\mux\muy^2}{160\ell^3}$.
% $\etax = \frac{\muy^2}{160\ell^3}$

% \begin{align}
%     \max\left\{ 1 - {\mux\etax},  1 - \frac{\ell^2\etax}{\muy} \right\}
% \end{align}
% \begin{align}
%     &\etax\leq\frac{1}{5\ell} \\
%     &\ell \geq \mux\\
%     &\implies \mux\etax\leq \frac{1}{5}<1
% \end{align}

% % \begin{align}
% %     &\ell \geq \mux, \muy\\
% %     &1 - \frac{}{} \leq \frac{\ell^2}{\muy}\etax \geq \ell\etax
% % \end{align}

\end{proof}

\newpage
\section{Convex Markov Games}
\label{sec:cmgs-convergence-appendix}
\begin{lemma}[Continuity of the occupancy measure]
Let $\lambda \in \Delta(\calS\times\calA\times\calB)$ be the occupancy measure in a Markov game. Then $\lambda$ is $\Llambda$-Lipschitz continuous and $\elllambda$-smooth with respect to the policy pair $(x,y)\in\calX\times\calY$. Specifically, for all $(x,y)$ and $(x',y')$,
\begin{align}
\norm{\lambda(x,y) - \lambda(x',y')}
&\leq\Llambda\norm{(x,y) - (x',y')};\\
\norm{\nabla \lambda(x,y) - \nabla \lambda(x',y')}
&\leq \elllambda\norm{(x,y) - (x',y')},
\end{align}
where $\Llambda:=\frac{{|\calS|^{\frac{1}{2}}}\paren{|\calA|+|\calB|}}{(1-\gamma)^2}$, and $\elllambda:=  \frac{2\gamma{|\calS|^{\frac{1}{2}}}\paren{|\calA|+|\calB|}^{\frac{3}{2}}}{(1-\gamma)^3}$.

Moreover, consider the functions $\lambda_1^{-1}:\Delta(\calS\times\calA)\to\calX$ and $\lambda_2^{-1}:\Delta(\calS\times\calB)\to\calY$, such that:
\begin{align}
 \lambda_{1}^{-1}\paren{\lambda_1(x,y)} &= x;\\
  \lambda_{2}^{-1}\paren{\lambda_2(x,y)} &= y.
\end{align}
For any fixed $y$ (respectively, $x$), $\lambda^{-1}$ is $\Llambdainv$-Lipschitz continuous with respect to $\lambda_1$ (respectively, $\lambda_2$), \textit{i.e.}, for all $\lambda_1(x,y), \lambda_1(x',y)$---respectively, $\lambda_2(x,y), \lambda_2(x,y')$,
\begin{align}
\norm{x-x'}&\leq\Llambdainv \norm{\lambda_1^{-1}\paren{\lambda_1(x,y) - \lambda_1(x',y)}};\\
\norm{y-y'}&\leq\Llambdainv \norm{\lambda_2^{-1}\paren{\lambda_2(x,y) - \lambda_2(x,y')}},
\end{align}
with $\Llambdainv:=  \frac{2}{\min_s \varrho(s)(1-\gamma)}$.
% \label{continuity}
\end{lemma}
\begin{proof}
    \citep[Lemmata C.2 \& C.3]{kalogiannis2024learning}.
\end{proof}

Throughout, we will assume $\varepsilon$-greedy parametrization of the policies, \textit{i.e.,} the players play according to policies:
\begin{gather}
    \pi_x  = (1-\varepsilon) x + \frac{\varepsilon \bm{1} }{|\calA|}, \quad \text{and}\quad \pi_y = (1-\varepsilon)y + \frac{\varepsilon \bm{1}}{|\calB|},
\end{gather}
where $\bm{1}$ is the all-ones vector of appropriate dimension.

\newpage

\subsection{Properties of the cMG Utility Fucntions}
We reinstate that the utility function $L_U$-Lipschitz continuous and $\ell_U$-smooth with,
\begin{gather}
    L_U = \frac{L_F|\calS|^{\frac{1}{2}}\paren{|\calA|+ |\calB| } }{(1-\gamma)^2};
    \quad
    \ell_U = \frac{2\ell_F\gamma |\calS|^{\frac{1}{2}} \paren{|\calA|+ |\calB|}^{\frac{3}{2}} }{ (1-\gamma)^3 }.
\end{gather}

\newcommand{\mur}{{\mu_{\mathrm{reg.}}}}

\subsubsection{Properties of cMGs with Convex Utilities}
For the two-player zero-sum cMG with merely concave utilites, we consider the regularized utility function $U^{\mu}(x,y):= U(x,y) + \frac{\mu}{2}\norm{\lambda_2(x,y)}^2$.
We note that that an upper bound to the Lipschitz smoothness moduli of $U^\mu, \nabla U^{\mu}$ is given by $ L_U^{\mu},  \ell_U^{\mu}$:
\begin{gather}
    L_U^{\mu} = 2L_U L_\lambda^2 = \frac{L_F|\calS|^{\frac{3}{2}} \paren{|\calA|+ |\calB|}^{3} }{(1-\gamma)^6} ; \quad \ell_U^{\mu} = \frac{4\ell_F \gamma^2|\calS|^{\frac{3}{2}} \paren{|\calA|+|\calB|}^{4} }{(1-\gamma)^8 }.
\end{gather}
Let us now compute the moduli of the quadratic growth and pPL condition,
\begin{gather}
    \muqg = \frac{(\minrho)^2(1-\gamma)^2 \mu}{4};
    \quad
    \mupl = \frac{(\minrho)^4(1-\gamma)^{12} \mu^2 }{4\ell_F \gamma^2 |\calS|^{\frac{3}{2}} \paren{|\calA|+ |\calB|}^4 }
\end{gather}
Since the quantity $\muqg\mupl$ will frequently appear in our calculations, we write:
\begin{align}
    \muqg \mupl= \frac{(\minrho)^6(1-\gamma)^{14} \mu^3 }{4\ell_F \gamma^2 |\calS|^{\frac{3}{2}} \paren{|\calA|+ |\calB|}^4 }
\end{align}
The condition number $\kappa$ is defined as $\kappa:= \frac{\ell_U^\mu }{\sqrt{\muqg\mupl}}$ and is equal to:
\begin{align}
    \kappa = \frac{16\ell_F^{\frac{3}{2}} \gamma^3 |\calS|^{\frac{9}{4}} \paren{|\calA|+ |\calB|}^6 }{(\minrho)^3 (1-\gamma)^{15}\mu^{\frac{3}{2}}}
\end{align}
Finally, we define function $\Phi^\mu:= \max_{y\in\calY} U^\mu(x,y)$ and observe for its Lipschitz modulus, $L_{\Phi}^\mu$, and its gradient's Lipschitz modulus, $\ell_{\Phi}^\mu$:
\begin{gather}
    L_{\Phi}^{\mu} = L_U^\mu = \frac{L_F|\calS|^{\frac{3}{2}} \paren{|\calA|+ |\calB|}^{3} }{(1-\gamma)^6}; \quad  \ell_{\Phi}^\mu = 2\ell_U^{\mu}\kappa = \frac{512\ell_F^{\frac{9}{2}} \gamma^{8} |\calS|^{\frac{23}{4}} \paren{|\calA|+|\calB|}^{\frac{23}{2}} }{(1-\gamma)^{23} (\minrho)^3 \mu^{\frac{3}{2}}}.
\end{gather}

% \begin{lemma}
%     $\Phi^\mu$ is $L_U = L_F \Llambda = \frac{L_F |\calS|^{\frac{1}{2}}\paren{|\calA| + |\calB|} }{ (1-\gamma)^2 }$, is $\ell_U= \ell_F\elllambda= \frac{2\ell_F\gamma|\calS|^{\frac{1}{2}}\paren{|\calA|+ |\calB| }^{\frac{3}{2}} }{(1-\gamma)^3}$
%     $\kappa_{\Phi,\mu} = \frac{\ell_U}{\sqrt{\mu \muqg }}$
% \end{lemma}
% \begin{itemize}
%     % \item $ \mu := \frac{\ell}{1 + 4 \paren{1 + \frac{2\ell}{\mur} }^2} \geq \frac{\ell \mur^2 }{ 9\mur^2 +  32\ell^2 }\geq\frac{\mur^2}{41\ell}.$
%     % \item $\mu\muqg \geq \frac{\ell_U \mur^3 }{9\mur^2 + 32\ell_U^2 }\geq\frac{\ell_U \mur^3}{41\ell^2_U}= \frac{\mur^3}{41\ell_U} $
%     % \item $\frac{1}{\mu\muqg}\leq \frac{9\mur^2 + 32\ell_U^2 }{\ell_U \mur^3 } \leq \frac{ 41\ell_U }{ \mur^3 } $
%     \item $\muqg = \mu_c^2 \mur = \frac{(1-\gamma)^2 \paren{\min \varrho(s)}^2 \mur }{8} $
%     \item $    \mu := \frac{\ell}{1 + 4\paren{1 + \frac{2\ell}{2\mu_c^2\mu_H } }^2 } \geq \frac{\mu_c^4 \mur^2}{41\ell_U\elllambda }.$
%     % \item $\kappa = \frac{\ell_U}{\sqrt{\mu \muqg }} = \frac{41\ell_U^{\frac{3}{2}}}{ \mur^{\frac{3}{2}} }$
% \end{itemize}
% \begin{align}
% \end{align}

% \newcommand{\mupl}{\frac{\mur^2}{41\ell_U\elllambda }}

\subsubsection{Properties of cMGs with Strongly Concave Utilities}
We carry over the same calculations for the cMG with strongly concave utilities. Since we do not perturb the utilities, the relevant Lipschitz moduli remain the same. This is not the case for $\muqg,\mupl$. We write:
\begin{gather}
    \muqg = 
    \frac{(\minrho)^2(1-\gamma)^2 \mu}{4};
    \quad
    \mupl =
    \frac{(\minrho)^4(1-\gamma)^{7} \mu^2 }{4\ell_F \gamma |\calS|^{\frac{1}{2}} \paren{|\calA|+ |\calB|}^{\frac{3}{2}} }
\end{gather}
Further, $\kappa:= \frac{\ell_U}{\sqrt{\muqg\mupl}}$ which is equal to:
\begin{align}
    \kappa 
    = 
    \frac{4\sqrt{2} \ell_F^{\frac{3}{2}} \gamma^{\frac{3}{2}} |\calS|^{\frac{3}{4}} \paren{|\calA| + |\calB|}^{\frac{9}{4}} }{ (\minrho)^{3}(1-\gamma)^{\frac{9}{2}} \mu^{\frac{3}{2}}}.
\end{align}
Finally, the Lipschitz modulus of $\Phi$ and its gradient will be:
\begin{align}
    L_{\Phi} = L_U
    =\frac{L_F|\calS|^{\frac{1}{2}}\paren{|\calA|+ |\calB| } }{(1-\gamma)^2}; \quad
    \ellphi = \frac{16\sqrt{2} \ell_F^{\frac{5}{2}} \gamma^{\frac{5}{2}} |\calS|^{\frac{5}{4}} \paren{|\calA|+|\calB|}^{\frac{15}{4}} }{(\minrho)^3 (1-\gamma)^{\frac{15}{2}} \mu^{\frac{3}{2}} }.
\end{align}
\newpage

\subsection{Stochastic Estimators}
Using a trajectory $\xi := \paren{s^{(0)}, a^{(0)}, s^{(1)}, a^{(1)},\dots }$ and $\xi := \paren{s^{(0)}, b^{(0)}, s^{(1)}, b^{(1)},\dots }$ respectively, we define the estimates of $\lambda_1, \lambda_2$,
\begin{align}
    \hat{\lambda}^{s,a}_{1,\xi} &:= \sum_{h=1}^{H_{\xi}} \gamma^h  \mathbbm{1}\{ s^{(h)}=s,a^{(h)}=a \};\\
    \hat{\lambda}^{s,b}_{2,\xi} &:= \sum_{h=1}^{H_{\xi}} \gamma^h  \mathbbm{1}\{ s^{(h)}=s,b^{(h)}=b \}.
\end{align}

Additionally, for a pseudo-reward vector $z\in\R^{|\calS|\times|\calA|}$ and $\xi := \paren{s^{(0)}, a^{(0)}, s^{(1)}, a^{(1)},\dots }$---or $z\in\R^{|\calS|\times|\calB|}$ correspondingly and $\xi := \paren{s^{(0)}, b^{(0)}, s^{(1)}, b^{(1)},\dots }$---we define gradient estimates $\hat{\nabla}_{x,\xi}^{t},\hat{\nabla}_{y,\xi}^{t}$:
\begin{align}
\hatgradfbx(\xi|x\at{t},y\at{t},z) &= \sum_{h=0}^{H-1} \left[ \gamma^h z \paren{s^{(h)}, a^{(h)}} \cdot \left( \sum_{h'=0}^{h} \nabla_x \log x \paren{a^{(h')} | s^{(h')}} \right) \right]
\\
\hatgradfby(\xi|x\at{t}, y\at{t},z) &= \sum_{h=0}^{H-1} \left[ \gamma^h z \paren{s^{(h)}, b^{(h)}} \cdot \left( \sum_{h'=0}^{h} \nabla_y \log y \paren{b^{(h')} | s^{(h')}} \right) \right]
\end{align}

At every iteration $t$, each agent constructs $\hatr\at{x,t} \gets \nabla_{\lambda_1} F(\hat{\lambda}_{1,t})$ describe how to construct $\hatr_x,\hatr_y$ batch versions of the latter estimators will be,
\begin{gather}
    \hat{\nabla}_{x}^{t} := \sum_{i=1}^{M_x}\hatgradfbx(\xi_i|x\at{t},y\at{t},\hatr\at{x,t}) ,  \quad \text{and} ~\quad  \hat{\nabla}_{y}^{t} := \sum_{i=1}^{M_y}\hatgradfby(\xi_i|x\at{t},y\at{t},\hatr\at{y,t}) .
    \label{policy-gradient-estimator}
\end{gather}

Now, for each one of the two cases we let $z \gets \nabla_{\lambda_1}F\paren{\hat{\lambda}_{1,t} ; y\at{t}}$ and  $z\gets \nabla_{\lambda_1}F\paren{\hat{\lambda}_{1,t} ; y\at{t}}$ correspondingly.
% \begin{gather}
%     \hat{\lambda}\at{1,t} =\frac{1}{M_x} \sum_{i=1}^{M_x} \hat{\lambda}_{1,t,\xi_i};
%     & \hat{r}_x = \nabla_{\lambda_1} F\paren{\hat{\lambda}_{1,t};y}.
% % \begin{align}
% %     \begin{array}{cc}
% %     \hat{\lambda}\at{1,t} =\frac{1}{M_x} \sum_{i=1}^{M_x} \hat{\lambda}_{1,t,\xi_i};
% %     & \hat{r}_x = \nabla_{\lambda_1} F\paren{\hat{\lambda}_{1,t};y}.
% % \end{array}
% % \end{align}
% \\
%     \hat{\nabla}_x^{t} = \sum_{h=0}^{H-1} \left[ \gamma^h \hat{r}_x \paren{s^{(h)}, a^{(h)}} \cdot \left( \sum_{h'=0}^{h} \nabla_x \log x \paren{a^{(h')} | s^{(h')}} \right) \right].
% \end{gather}

% \begin{equation}
%     \begin{array}{ccc}
%     \hat{\lambda}\at{2,t} =\frac{1}{M_y} \sum_{i=1}^{M_y} \hat{\lambda}_{2,t,\xi_i}; 
%     & \hat{r}_y = \nabla_{\lambda_2} F\paren{\hat{\lambda}_{2,t};x};
%     & \hat{\nabla}_y^{t} = \sum_{h=0}^{H-1} \left[ \gamma^h \hat{r}_y \paren{s^{(h)}, b^{(h)}} \cdot \left( \sum_{h'=0}^{h} \nabla_y \log y \paren{b^{(h')} | s^{(h')}} \right) \right].

% \end{array}
% \end{equation}
\begin{remark}
  This stochastic gradient estimator assumes access to a gradient oracle for $\nabla_{\lambda_1}F$. Access to an oracle for $F$ or $\nabla_\lambda F$ is an assumption made in virtually all the references that we have encountered concerning policy gradient methods for MDPs with general utilities~\citep{zhang2020variational,zhang2021convergence,barakat2023reinforcement}. Designing a gradient estimator that entirely relies on samples lies well-beyond the scope of our paper. Nevertheless, this assumption is not a strong one as even with a complete knowledge of $F$, agents cannot control the function's arguments except for affecting them implicitly.
\end{remark}
\paragraph{The occupancy measure estimator}

Let a trajectory sample $\xi$ of length $H_\xi$, we define the occupancy estimator for player $1$ to be,
\begin{align}
    \hat{\lambda}^{s,a}_{1}(\xi) &:= \sum_{h=1}^{H_{\xi}} \gamma^h  \mathbbm{1}\{ s^{(h)}=s,a^{(h)}=a \}\\
    \hat{\lambda}^{s,b}_{2}(\xi) &:= \sum_{h=1}^{H_{\xi}} \gamma^h  \mathbbm{1}\{ s^{(h)}=s,b^{(h)}=b \}
\end{align}

Assume policies $x\at{t},y\at{t}$ and corresponding sampled trajectories $\{\xi_{i,t}\}_{i=1}^{M_x}$ and $\{\xi'_{i,t}\}_{i=1}^{M_x}$ for player 1 and 2 correspondingly. The occupancy batch estimators $\hat{\lambda}\at{1,t},\hat{\lambda}\at{2,t}$
\begin{gather}
    \hat{\lambda}_{1,t} := \frac{1}{M_x}\sum_{i=1}^{M_x} \hat{\lambda}_1 (\xi_{i,t}), \qquad \hat{\lambda}_{2,t} := \frac{1}{M_x}\sum_{i=1}^{M_x} \hat{\lambda}_2 (\xi'_{i,t})
\end{gather}

\begin{align}
    \lambda^{s,a}_{1,H}(x,y) &:= \E_{x,y} \left[ \sum_{h=1}^H \gamma^h  \mathbbm{1}\{ s^{(h)}=s,a^{(h)}=a \} \pr\paren{s^{(h+1)}|a^{(h)},b^{(h)},s^{(h)} }\left|s^{(0)}\sim \varrho\right.\right]\\
    \lambda^{s,b}_{2,H}(x,y) &:= \E_{x,y} \left[ \sum_{h=1}^H \gamma^h  \mathbbm{1} \{ s^{(h)}=s,b^{(h)}=b \} \pr\paren{s^{(h+1)}|a^{(h)},b^{(h)},s^{(h)} }\left|s^{(0)}\sim \varrho\right.\right]
\end{align}

\begin{lemma}
    Let an empirical estimate $\hat{\lambda}_1$ of the truncated-horizon occupancy measure $\lambda_{1,H}$. Then, the variance of the estimate can be bounded as:
    \begin{equation}
            \E \left[ \norm{ \hat{\lambda}_{1} - \lambda_{1,H} }^2 \right] \leq \frac{1}{M_x(1-\gamma)^2}.
    \end{equation}
\end{lemma}

\paragraph{The policy gradient estimator}
Assuming $\varepsilon$-greedy parametrization to control the variance of our estimators, we can go forward and state our main lemmata regarding the gradient policy gradient estimator's variances.
\begin{lemma}\label{lemma:variance:policy-gradient-estimator}
    Let $\hat{\nabla}_{x}^t$ be defined as in \eqref{policy-gradient-estimator}. The estimators variance can be bounded as:
    \begin{align}
        \E \left[ \norm{\hat{\nabla}_x^{t} - \nabla_x F(\lambda_{1,H}(x\at{t},y\at{t}); y\at{t} ) }^2 \right] 
        \leq \frac{27  L_F^2}{M_x(1-\gamma)^6 \varepsilon^2}.
    \end{align}
\end{lemma}
\begin{proof}
We essentially carry over the same computation as Step 2 in the proof of \citep{zhang2021convergence}[Lemma F.2] using our notation and assumptions. To simplify the excessively busy notation and improve readability, let us introduce some shortcuts, \begin{itemize}
    \item $\bar{\lambda}_1 := \lambda_{1,H}(x\at{t},y\at{t})$,
    \item $\hat{\lambda}_1 := \hat{\lambda}_{1,t}$,
    \item $r_* := \nabla_{\lambda_1} F\paren{ \lambda_{1,H}(x\at{t},y\at{t}) ;y\at{t}}  =  \nabla_{\lambda_1} F\paren{\bar{\lambda}_1 ;y\at{t}}$,
    \item $\mathbf{J}:= \nabla_x \lambda_{1,H}(x\at{t},y\at{t}) = \nabla_x \bar{\lambda}_1 $,
    \item $\hatr := \nabla_{\lambda_1} F\paren{ \hat{\lambda}_1 ;y\at{t}}  $,
    \item $\gradfb := \nabla_x F\paren{{\lambda}_{1,H}(x\at{t},y\at{t});y\at{t}} = \nabla_x F\paren{\bar{\lambda}_{1};y\at{t}}$,
    \item $\hatgradfb_{i}:= \hatgradfbx(\xi_i |x\at{t},y\at{t},\hatr)$,
    \item $\hatgradfb_{i,*}:= \hatgradfbx(\xi_i |x\at{t},y\at{t}, r_* )$.
    % \item $\hatgradfbx_* \hatgradfbx(\xi_i |x\at{t},y\at{t}, r_*)$
\end{itemize}
We can finally write:
\begin{align}
    \mathbb{E} \left[ \|\hat{\nabla}_{x}^{t} - \mathbf{J}^\top \hat{r} \|^2 \right] 
    &= \mathbb{E} \Bigg[ \Bigg\| \frac{1}{M_x} \sum_i^{M_x} \hatgradfb_{i,*} + \frac{1}{M_x} \sum_i^{M_x} \hatgradfb_{i,*} - \gradfb 
    + \gradfb - \mathbf{J}^\top \hat{r} \Bigg\|^2 \Bigg] \\
    &\leq 3 \mathbb{E} \Bigg[ \Bigg\| \frac{1}{M_x} \sum_i^{M_x} \Big(\hatgradfb_{i} -\hatgradfb_{i,*} \Big) \Bigg\|^2 \Bigg] 
    + 3 \mathbb{E} \Bigg[ \Bigg\| \gradfb - \mathbf{J}^\top \hat{r}\at{t} \Bigg\|^2 \Bigg]  + 3 \mathbb{E} \Bigg[ \Bigg\| \frac{1}{M_x} \sum_i^{M_x} \hatgradfb_{i,*} - \gradfb \Bigg\|^2 \Bigg].
\end{align}

For the first term we write,
\begin{align}
    \mathbb{E} \Bigg[ \Bigg\| \frac{1}{M_x} \sum_i^{M_x} \Big(\hatgradfb_{i} - \hatgradfb_{i,*} \Big) \Bigg\|^2 \Bigg] 
    &\leq \frac{1}{M_x} \sum_i^{M_x} \mathbb{E} \Bigg[ \Big\| \hatgradfb_{i} - \hatgradfb_{i,*} \Big\|^2 \Bigg] \\
    &\leq \frac{4 }{(1-\gamma)^4 \varepsilon^2 } \cdot \mathbb{E} \Big[ \|\hatr - r_*\|_\infty^2 \Big] \\
    &\leq \frac{4  L_F^2}{(1-\gamma)^4\varepsilon^2 }  \mathbb{E} \Big[ \norm{\hat{\lambda}_1 - \bar{\lambda}_1 }^2 \Big] \label{grad-est-var-bound}  \\
    &\leq \frac{4  L_F^2}{M_x(1-\gamma)^6 \varepsilon^2} .
\end{align}

Where \eqref{grad-est-var-bound} follows from the following fact:
\begin{align}
    \|\hat{g}(\xi | x, r_1) - \hat{g}(\xi | x, r_2)\| 
    &= \left\| \sum_{t=0}^{H-1} \gamma^t \cdot \big(r_1(s_t, a_t) - r_2(s_t, a_t)\big)  
    \left( \sum_{t' = 0}^t \nabla_x \log \pi_x(a_{t'} | s_{t'}) \right) \right\| \\
    &\leq \sum_{t=0}^{H-1} \frac{1}{\varepsilon}  \gamma^t (t+1)  \|r_1 - r_2\|_\infty \\
    &\leq \frac{2  \|r_1 - r_2\|}{  (1-\gamma)^2\varepsilon }.
\end{align}

The second and third terms can be similarly bounded as:
\begin{gather}
    \mathbb{E} \Bigg[ \Bigg\| \gradfb^{t} - \mathbf{J}^\top \hat{r}\at{t} \Bigg\|^2 \Bigg]\leq \frac{ L_F^2}{M_x (1-\gamma)^4};\\
    \text{and}\\
    \mathbb{E} \Bigg[ \Bigg\| \frac{1}{M_x} \sum_i^{M_x} \hatgradfb_{i,*} - \gradfb \Bigg\|^2 \Bigg] \leq  \frac{4L_F^2}{M_x (1-\gamma)^6\varepsilon^2}.
\end{gather}
\end{proof}

\begin{lemma} \label{lemma:bias:policy-gradient-estimator}
     For any $H > \frac{1}{\ln(1\sqrt{\gamma})}$ and a greedy exploration parameter $\varepsilon$, the following inequality holds true for the bias of the policy gradient estimator defined in  \cref{policy-gradient-estimator}:
\begin{align}
    \|\nabla_x F(\lambda_H(x,y);y) - \nabla_x F(\lambda_1(x,y),y)\|^2 
    & \leq \frac{256 L_F }{(1-\gamma)^6\varepsilon^2}  \exp\Bigl( -(1-\gamma) (H-1) \Bigr).
\end{align}
\end{lemma}
\begin{proof}
By \citep[Lemma E.3]{zhang2021convergence}, we have that,
\begin{align}
    \|\nabla_x F(\lambda_H(x,y);y) - \nabla_x F(\lambda_1(x,y),y)\|^2 
    &\leq\paren{ \frac{8  L_F^2}{(1-\gamma)^6\varepsilon^2 } 
    + 16 \frac{L_F^2}{\varepsilon^2} 
    \left( \frac{(H+1)^2}{(1-\gamma)^2} + \frac{1}{(1-\gamma)^4} \right) } \gamma^{2H}.
\end{align}
Further, we can bound the RHS as:
\begin{align}
    \paren{ \frac{8  L_F^2}{(1-\gamma)^6\varepsilon^2 } 
    + 16 \frac{L_F^2}{\varepsilon^2} 
    \left( \frac{(H+1)^2}{(1-\gamma)^2} + \frac{1}{(1-\gamma)^4} \right) } \gamma^{2H}&\leq \frac{256 L_F (H+1)^2 }{(1-\gamma)^6\varepsilon^2} \gamma^{2H}.
\end{align}
Now, in order to simplify the display, compute an $H>0$ such that:
\begin{align}
    \frac{(H+1)^2}{(1/\gamma)^{H+1} }\leq 1
\end{align}
equivalently,
\begin{align}
    2\log_{1/\gamma}(H+1) \leq H+1
\end{align}
changing the base of the logarithm, we get,
\begin{align}
    \frac{1}{\ln(1/\sqrt{\gamma})} \ln (H+1) \leq H+1,
\end{align}
but the latter is true for any $H+1>1/\ln(1/\sqrt{\gamma})$.
Hence, since $\gamma <1$ and noting that $\gamma = 1-(1-\gamma) $,
\begin{align}
    \frac{256 L_F (H+1)^2 }{(1-\gamma)^6\varepsilon^2} \gamma^{2H} \leq \frac{256 L_F }{(1-\gamma)^6\varepsilon^2} \gamma^{H-1}\leq \frac{256 L_F }{(1-\gamma)^6\varepsilon^2}  \exp\Bigl( -(1-\gamma) (H-1) \Bigr).
\end{align}
\end{proof}

\subsection{Additional Supporting Claims}
In this subsection we bound the errors incurred due to the $\varepsilon$-greedy parametrization, and the regularization. 

\paragraph{Error due to \texorpdfstring{$\varepsilon$}{ε}-greedy parametrization.}
\begin{claim}
    Assume an $L$-Lipschitz continuous function $f:\calX\to\R$. Further, let $\calX$ be a concatenation of $n$ $m$-dimensional probability simplices. Further let $w$ be the mapping $ w(x) := (1-\varepsilon)x  + \varepsilon/m$ for some $0<\varepsilon<1$.
    $$|f(x) - (f\circ w)(x) |\leq \frac{\sqrt{n}L\varepsilon}{m}.$$
    Further, if for some $\epsilon>0$ and  $x\in\calX$, 
    $$f(x) - f^\star \leq \epsilon,$$
    then 
    $$(f\circ w)(x) - f^\star\leq \epsilon + \frac{\sqrt{n}L\varepsilon}{m}.$$
    The second follows directly.
    \label{varepsilon-bound-on-composed-function}
\end{claim}
\begin{proof}
 The claim follows from the Lipschitz continuity of $f$. It is the case that $(f\circ w)(x) = f((1-\varepsilon) x  + \varepsilon /m ),$
 \begin{align}
     |f(x) - f((1-\varepsilon) x  + \varepsilon /m )| &\leq L \norm{ x - ((1-\varepsilon) x  + \varepsilon /m )}\\
     &\leq \frac{\sqrt{n}L\varepsilon}{m}.
 \end{align}
\end{proof}

% \begin{claim}
%     Let an $L$-Lipschitz continuous function $f:\calX\times\calY\to\R$. Further let $w$ be the mapping $ w(x) := (1-\varepsilon)x  + \varepsilon/m$ for some $0<\varepsilon<1$. Further, define $\Phi(x):= \max_{y\in\calY} f(x,y)$ and $(\Phi\circ w)(x):= \max_{y\in\calY} f(w(x))$. Then, it is the case that:
%     \begin{align}
%         (\Phi \circ w) (x) - \Phi^\star\leq\epsilon + \frac{L\varepsilon}{m}. 
%     \end{align}
% \end{claim}
% \begin{proof}
% The proof directly follows from an application of     \cref{varepsilon-bound-on-composed-function} and the $L$-Lipschitz continuity of $\Phi$ given that $f$ is $L$-Lipschitz continuous.
% \end{proof}

\begin{claim}
    Let an $L$-Lipschitz continuous function $f:\calX\times\calY\to\R$. Let $\calX,\calY$ be concatenations of $n_x$ and $n_y$ $m_x$- and $m_y$-dimensional probability simplices respectively. Further let $w$ be the mapping $ w(x;\varepsilon,m) := (1-\varepsilon)x  + \varepsilon/m$ for some $0<\varepsilon<1$. Further, for some $\varepsilon_x, \varepsilon_y >0$ define
    $f_{w}(x,y) := f\paren{w(x;\varepsilon_x,m_x), w(y;\varepsilon_y,m_y)}$, $\Phi_w(x):= \max_{y\in\calY} f_w (x,y)$, and $\Phi^\star_w := \min_{x\in\calX}\Phi_w(x)$. Then, the following inequalities hold true,
    \begin{itemize}
        \item $|\Phi_w (x) - \Phi(x)|\leq \frac{\sqrt{n_x}L\varepsilon_x}{m_x} + \frac{\sqrt{n_y}L\varepsilon_y}{m_y} $
        \item $|\Phi^\star-\Phi_w^\star|\leq \frac{\sqrt{n_x}L\varepsilon_x}{m_x} + \frac{\sqrt{n_y}L\varepsilon_y}{m_y}$.
    \end{itemize}
    % \begin{align}
    %     (\Phi \circ w) (x) - \Phi^\star\leq\epsilon + \frac{L\varepsilon}{m}. 
    % \end{align}
    \label{minmax-bound-on-composed-function}
\end{claim}
\begin{proof}
The proof of the first item directly follows from an application of \cref{varepsilon-bound-on-composed-function} and the $L$-Lipschitz continuity of $\Phi$ given that $f$ is $L$-Lipschitz continuous.

For the second item, we use the inequality we just attained and write,
\begin{align}
    \Phi_w (x) \leq \Phi(x) + \frac{\sqrt{n_x}L\varepsilon_x}{m_x} + \frac{\sqrt{n_y}L\varepsilon_y}{m_y}.
\end{align}
Minimizing over $x$ for the two sides yields,
\begin{align}
    \Phi_w (x^\star_w) \leq \Phi(x^\star) + \frac{\sqrt{n_x}L\varepsilon_x}{m_x} + \frac{\sqrt{n_y}L\varepsilon_y}{m_y}.
\end{align}
Where, $x^\star,x^\star_w$ are such that $x^\star\in\argmin_{x\in\calX}\Phi(x)$ and $x^\star_w\in\argmin_{x\in\calX}\Phi_w(x)$. Applying the same trick to the other side of the other direction of the inequality:
\begin{align}
    \Phi(x^\star) \leq  \Phi_w (x^\star_w) +\frac{\sqrt{n_x}L\varepsilon_x}{m_x} + \frac{\sqrt{n_y}L\varepsilon_y}{m_y}.
\end{align}
As such,
\begin{align}
    |\Phi(x^\star) -  \Phi_w (x^\star_w)| \leq \frac{\sqrt{n_x}L\varepsilon_x}{m_x} + \frac{\sqrt{n_y}L\varepsilon_y}{m_y}.
\end{align}
% For the second item, we observe that for $x^\star,x^\star_w$ such that $x^\star\in\argmin_{x\in\calX}\Phi(x)$ and $x^\star_w\in\argmin_{x\in\calX}\Phi_w(x)$,
% \begin{align}
%     |\Phi(x^\star) - \Phi_w(x^\star) | &\leq \frac{\sqrt{n_x}L\varepsilon_x}{m_x} + \frac{\sqrt{n_y}L\varepsilon_y}{m_y} \\
%     |\Phi(x_w^\star) - \Phi_w(x_w^\star) |&\leq \frac{\sqrt{n_x}L\varepsilon_x}{m_x} + \frac{\sqrt{n_y}L\varepsilon_y}{m_y}.
% \end{align}
% We combine the two latter inequalities and the desired claims follows.

\end{proof}

\begin{claim}
    Assume an $L$-Lipschitz continuous and $\ell$-smooth function $f:\calX\to\R$ that is $\mu$-pPL. Further, let $\calX$ be a concatenation of $n$ $m$-dimensional probability simplices. Further let $w$ be the mapping $ w(x) := (1-\varepsilon)x  + \varepsilon/m$ for some $0<\varepsilon<1$. Define $\calD^{f }_{\calX}(x,\alpha), \calD^{f \circ w}_{\calX}(x,\alpha)$ for some $\alpha>0$ to be:
    \begin{align}
        \calD^{f }_{\calX}(x,\alpha) &:=-2\alpha \min_{y\in\calX} \left\{ \inprod{\nabla f(x)}{y-x} + \frac{\alpha}{2}\norm{x-y}^2 \right\} 
        \\  
        \calD^{f\circ w}_{\calX}(\pi_x,\alpha) &:= -2\alpha \min_{y\in\calX} \left\{ \inprod{\nabla f\circ w(x)}{y-x} + \frac{\alpha}{2}\norm{x-y}^2 \right\}.
    \end{align}
    Then, it is the case that:
        \begin{align}
            \frac{1}{2}\calD_{\calX}^{f\circ w}(x,\alpha) \geq \mu \paren{ f(x) - f^\star } - 8\sqrt{n} \ell L \alpha\diam{\calX} \varepsilon.
        \end{align}
    \label{claim:epsilon-greedy-D}
\end{claim}
\begin{proof}
    We essentially need to prove that in fact:
    \begin{align}
        \left| \calD_{\calX}^{f\circ w}(x,\alpha) - \calD_{\calX}^{f}(x,\alpha) \right| \leq 16\sqrt{n} \ell L \alpha\diam{\calX} \varepsilon. \label{D-lipschitz}
    \end{align}
    We write $ \calD^{f }_{\calX}(x,\alpha) $ equivalently as:
    \begin{align}
        \calD^{f }_{\calX}(x,\alpha) &= \max_{y\in\calX} \left\{  2\alpha \inprod{\nabla f(x)}{x-y} - {\alpha^2}\norm{x-y}^2 \right\}.
    \end{align}
    We define $G(\cdot;y)$ after we isolate the display inside the $\max$-operator and in place of the gradient put any vector $v$ of the same dimension:
    \begin{align}
       G(v;y):= 2\alpha \inprod{v}{x-y} - {\alpha^2}\norm{x-y}^2.
    \end{align}
    We see that $\nabla_v G(v;y) =2\alpha( x - y)$ and $\norm{x-y}\leq \diam{\calX}.$ As such, $G$ is $2\alpha\diam{\calX}$-Lipschitz continuous in $v$. This consequently means that $\Psi(v):= \max_{y\in\calY}G(v;y)$ is also $2\alpha\diam{\calX}$-Lipschitz in $v$. Now, all that is left to do is to bound the distance between $\nabla_x (f\circ w)(x) $ and $\nabla_x f(x)$. By the chain rule:
    \begin{align}
        \nabla_x (f\circ w)(x) = (1-\varepsilon) \nabla_x f( z )\Bigr|_{z=w(x)} \label{distance-1}.
    \end{align}
    Further, set $x_{\varepsilon}:= w(x)$ and observe that 
    \begin{align}
        \norm{\nabla_x f(x) - \nabla_x f(x_{\varepsilon}) } &\leq \ell\norm{x - x_{\varepsilon}}\\
        &= \ell\norm{x - (1-\varepsilon)x - \varepsilon/m }\\
        &= \ell\norm{\varepsilon x - \varepsilon/m }\\
        &\leq \ell\varepsilon (\sqrt{n} + 1/m).
    \end{align}
    Where in the last inequality we use the fact that for a probability vector $y$, $\norm{y}\leq 1$ and the triangle inequality. Assuming that $n,m>1$, we can further simplify into:
    \begin{align}
        \norm{\nabla_x f(x) - \nabla_x f(x_{\varepsilon}) }\leq 2\sqrt{n}\ell\varepsilon \label{distance-2}
    \end{align}
    Now, by the fact that $\Psi(\nabla_x f(x))= \calD_{x}^{f}(x,\alpha), \Psi(\nabla_x (f\circ w)(x))= \calD_{x}^{f\circ w}(x,\alpha)$, and the Lipschitz continuity of $\Psi$, we write: 
    \begin{align}
        \left| \calD^{f }_{\calX}(x,\alpha) - \calD^{f\circ w }_{\calX}(x,\alpha) \right| &\leq 2 \alpha\diam{\calX} \norm{ \nabla_x (f\circ w)(x) - \nabla_x f(x) }\\
        &\leq 2 \alpha\diam{\calX} \norm{ \nabla_x f(x) -\nabla_x f(x_{\varepsilon})} +  2 \alpha\diam{\calX} \norm{ \varepsilon \nabla_x f(x_{\varepsilon}) }\\
        &\leq 4\sqrt{n} \ell \alpha\diam{\calX} \varepsilon +  2 \alpha\diam{\calX}\varepsilon L\\
        &\leq 16\sqrt{n} \ell L \alpha\diam{\calX} \varepsilon .
    \end{align}
    The claim follows.
\end{proof}
\begin{claim}
    Assume an $L$-Lipschitz continuous and $\ell$-smooth function $f:\calX\to\R$ such that:
    \begin{align}
        \max_{x'\in\calX,\norm{x-x'}\leq 1}\inprod{\nabla_x f(x)}{x-x'}\geq \mu\paren{ f(x) - f^\star }.
    \end{align}
    Further, let $\calX$ be a concatenation of $n$ $m$-dimensional probability simplices. Further let $w$ be the mapping $ w(x) := (1-\varepsilon)x  + \varepsilon/m$ for some $0<\varepsilon<1$. It is the true that:
    \begin{align}
        \max_{x'\in\calX, \norm{x-x'}\leq 1 } \inprod{ \nabla_{x}(f\circ w)(x) }{x-x'}  \geq \mu \paren{ f(x) - f^\star} - 8\sqrt{n} \ell L \diam{\calX} \varepsilon.
    \end{align}
    % $\frac{\epsilon\mu}{8\sqrt{n} \ell L \diam{\calX}} = \varepsilon$
    \label{claim:epsilon-greedy-gradient-dominance}
\end{claim}
\begin{proof}
    Similar to the previous case, we need to show that $G(v;x'):= \inprod{v }{x-x'}$ is Lipschitz continuous in $v$. Indeed, $\nabla_v G(v;x'):= x-x'$ and as such $\norm{\nabla_v G(v;x')}\leq \diam{\calX}$. Consequently, $\Psi(v):= \max_{x'\in\calX} G(v;x')$ is also $\diam{\calX}$-Lipschitz continuous. Further, as previously shown in \eqref{distance-2},
    \begin{align}
        \norm{\nabla_x f(x) - \nabla_x f(x_{\varepsilon}) }\leq 2\sqrt{n}\ell\varepsilon.
    \end{align}
    We can finally write,
    \begin{align}
        \max_{x'\in\calX, \norm{x-x'}\leq 1 } \inprod{ \nabla_{x}(f\circ w) }{x-x'} \geq \max_{x'\in\calX, \norm{x-x'}\leq 1 } \inprod{ \nabla_{x}f }{x-x'} - 8\sqrt{n} \ell L \diam{\calX} \varepsilon.
    \end{align}
\end{proof}

\paragraph{Error due to the regularizer.}
\begin{claim}
    Let a two-player zero-sum cMG with utility $U:\calX\times\calY\to \R$. Assume that the maximizing player employs regularization to their utility of the form:
    $$ U^\mur(x,y) := U(x,y) - \frac{\mur}{2} \norm{\lambda_2(x,y)}^2,$$
    where $\mur>0.$ Then, the following inequality is true:
    \begin{align}
        \norm{\nabla_x U(x,y) - \nabla_x U^\mu(x,y)}\leq \mur \norm{\Llambda}.
    \end{align}
    \label{claim:regularization-error}
\end{claim}
\begin{proof}
\begin{align}
\norm{\nabla_x U(x,y) - \nabla_x U^\mu(x,y)} &= \mur \norm{\nabla_x \lambda_2(x,y)}  \\
&\leq \mur \norm{\lambda_2} \norm{\nabla_x \lambda_2(x,y)}  \\
&\leq  {\mur}\Llambda.
\end{align}
\end{proof}
\subsection{Nested Policy Gradient}
\subsubsection{CMG with Concave Utilities}

\begin{theorem}
Assume a two-player zero-sum cMG and a desried accuracy $\epsilon>0$. Running  \cref{alg:nepg} a tuning of \cref{alg:nepg} with, $\mur = \Theta\paren{\frac{(1-\gamma)(\minrho)\epsilon}{L_F}}$ and,
 \begin{itemize}
     \item step-sizes, $\etax = \Theta\paren{\frac{(\minrho)^{5} (1-\gamma)^{\frac{18}{2}} \epsilon^{\frac{3}{2}} }{\ell_F^{\frac{5}{2}} \gamma^{\frac{5}{2}} |\calS|^{\frac{5}{2}} \paren{|\calA|+|\calB|}^{\frac{15}{4}} } }$ and $\etay = \Theta\paren{\frac{(1 - \gamma)^8}{4 \ell_F \gamma^2 |\mathcal{S}|^{\frac{3}{2}} (|\mathcal{A}|+ |\mathcal{B}|)^4}}$;
     \item exploration parameters $\varepsilon_x = \paren{\frac{\epsilon (\minrho)^{6} \mu^{3} \left(1 - \gamma\right)^{17}}{512 L_{F} |\calS|^{4} \ell_{F}^{5} \gamma^{5} \left(|\calA| + |\calB|\right)^{\frac{17}{2}}}
}$ and $\varepsilon_y =  \Theta\paren{\frac
{ \epsilon (\minrho)^{9} \mu_x^{\frac{9}{2}} \left(1 - \gamma\right)^{\frac{21}{2}}}
{|\calS|^{\frac{5}{4}} \ell_{F}^{\frac{3}{2}} \gamma^{\frac{1}{2}} \left(|\calA| + |\calB|\right)^{\frac{9}{4}}}
} $;
    \item batch-sizes $M_x =\Theta\left(  \frac{L_F^4\ell_F^{2} |\calS|^{7}\paren{|\calA|+|\calB|}^{11} }{ \epsilon^4 (\minrho)^{4} (1-\gamma)^{32} } \right) $ and $M_y = \Theta\paren{\frac{ \ell_F^{10} L_{F}^2 \gamma^8 |\calS|^{10} \paren{ |\calA| + |\calB|}^{22} }{(\minrho)^{16} (1-\gamma)^{52} \epsilon^8}} $;
    \item an outer-loop number of iterations at least $T_x =\Theta\left( \frac{ L_F\ell_F^{\frac{9}{2}} \gamma^{8} |\calS|^{\frac{29}{4}} \paren{|\calA|+|\calB|}^{\frac{29}{2}} }{(1-\gamma)^{\frac{55}{2}} (\minrho)^{\frac{13}{2}} \epsilon^{\frac{7}{2}}  } \right)$ and inner loop iterations that are at least $T_y = \Theta\paren{\frac{\ell_U \Llambda^2 \ell_F \gamma^2 |\calS|^{\frac{3}{2}} \paren{|\calA|+|\calB|}^{4} }{ (\minrho)^{6}(1-\gamma)^{14}\epsilon^2 } \ln\paren{\frac{ \ell_F L_F  \gamma |\calS|\paren{|\calA|+|\calB|} }{(\minrho)(1-\gamma)\epsilon}}};$
 \end{itemize}
 will output an iterate $x\at{t^\star},y\at{t^\star}$, such that:
 \begin{align}
     \E\left[ U(x\at{t^\star},y\at{t^\star+1} ) - \min_{x'\in\calX} U(x',y\at{t^\star+1})\right] &\leq \epsilon; \\
     \E\left[\max_{y\in\calY} U(x\at{t},y') - \E U(x\at{t^\star},y\at{t^\star}) \right] &\leq \epsilon. 
 \end{align}
\end{theorem}
\begin{proof}
In order to simplify our arguments, we formalize $\varepsilon$-greedy parametrization as a composition of the utility function $U(x,y)$ with mappings $w(x;\varepsilon,m)$ where $w(x):= (1-\varepsilon)x + \varepsilon\bm{1}/m$. In particular, our convergence guarantees are w.r.t. the function $U^{\mu}_{w}(x,y):= U^{\mu}\paren{w(x;\varepsilon_x,|\calA|), w(y;\varepsilon_y,|\calB|)} $ and $\tilde{\Phi} (x) := \max_{y\in\calY} U^{\mu}_{w}(x,y)$.

% $\kappa^2 = \frac{8\cdot 41\ell_U}{(1-\gamma)^2\paren{\min_s \varrho(s) }^2 \mur^3 }$

% $\kappa = \frac{\sqrt{\ell}}{(1-\gamma)\paren{\min_s \varrho(s) } \mu^{\frac{3}{2} } }$
% $\ell \leq  \ell_U\elllambda$

% $\ell_{\tilde{\Phi}} = aaaaa$

We will tune $T_x, \etax, H_x, \varepsilon_x, M_x $, by, \cref{theorem:formal-nestedgit-nc-ppl}.
\begin{align}
    \frac{1}{T}\sum_{t=0}^{T-1}\norm{\sresx{\tilde{\Phi},t}}^2 &\leq \frac{5\ell_{\tilde{\Phi}} L_{\tilde\Phi} \diam{\calX} }{ T} +  6\delta_x^2 + \frac{1}{M_x}\cdot \frac{4L_F^2}{(1-\gamma)^6\varepsilon_x^2}.
\end{align}

First, we see that in order to achieve an $\epsilon$-approximate best-response w.r.t. to $x$, then, $x$ needs to be an $(1-\gamma)\minrho\epsilon$-approximate stationary point. In general, then, we need to ensure that,
\begin{align}
\delta_x^2 = \Theta(\epsilon^2 (1-\gamma)^2 (\minrho)^2 ).
\end{align}
Hence, by \cref{claim:regularization-error}, we can pick the regulizer's coefficient to be:
$$\mur = \Theta\paren{\frac{(1-\gamma)(\minrho)\epsilon}{L_F}}.$$
\Cref{lemma:bias:policy-gradient-estimator} dictates the tuning of $H_x, H_y$,
\begin{align}
       H_x, H_y = \Theta\paren{ \frac{1}{1-\gamma} \ln \paren{\frac{\ell_F L_F \gamma |\calS| \paren{|\calA|+ |\calB| } }{\epsilon ( 1-\gamma) \minrho \mu_x  } } }.
\end{align}

%%%%%%%%%%%%%%%%%%%%%%%

Further, $$\varepsilon_x = \Theta{\paren{
\frac{(1-\gamma)(\minrho) \epsilon}{|\calS| \ell_F  L_F \elllambda  \Llambda^4 }
}}  = \Theta \paren{ \frac{\epsilon \, \minrho \, (1 - \gamma)^{12}}{2 L_F \ell_F  \, |\mathcal{S}|^{\frac{7}{2}} \gamma (|\mathcal{A}| + |\mathcal{B}|)^{\frac{11}{2}}}
}$$

\begin{align}
    \etay =\Theta \paren{\frac{1}{\ell_U^{\mu}} } = \Theta \paren{\frac{1}{\ell_F\elllambda^2 L_\lambda }} = \Theta\paren{\frac{(1 - \gamma)^8}{4 \ell_F \gamma^2 |\mathcal{S}|^{\frac{3}{2}} (|\mathcal{A}|+ |\mathcal{B}|)^4}}
\end{align}

\begin{align}
    \sigma_x^2 = \frac{L_F^4\ell_F^{2} |\calS|^{7}\paren{|\calA|+|\calB|}^{11} }{ \epsilon^2 (\minrho)^{2} (1-\gamma)^{30} }
\end{align}

\begin{align}
    M_x = \Theta\paren{ \frac{L_F^4\ell_F^{2} |\calS|^{7}\paren{|\calA|+|\calB|}^{11} }{ \epsilon^4 (\minrho)^{4} (1-\gamma)^{32} } }
\end{align}

As for the inner-loop, we need to set 
$$\epsilon_y =\Theta\paren{ \frac{\muqg(1-\gamma)^2(\minrho)^2\epsilon^2}{2\ell_U^2}} = \Theta\paren{
\frac{ \minrho^5 (1 - \gamma)^{11} \epsilon^3 }{ \ell_F^2 \gamma^2 |\mathcal{S}| (|\mathcal{A}| + |\mathcal{B}|)^3}
}
$$
and hence,
$$\varepsilon_{y} =\Theta\paren{ \frac{\muqg(1-\gamma)^2(\minrho)^2\epsilon^2}{\ell_U^4 L_U |\calS|} } =\Theta\paren{ \frac{ \minrho^5 (1 - \gamma)^{19}\epsilon^3}{ |\mathcal{S}|^{\frac{7}{2}} L_F  \ell_F^4 \gamma^4 (|\mathcal{A}| + |\mathcal{B}|)^7}
}$$

Now, we are ready to tune the number of inner-loop iterations, $T_y$,
\begin{align}
    \frac{\ell_U \Llambda^2 \ell_F \gamma^2 |\calS|^{\frac{3}{2}} \paren{|\calA|+|\calB|}^{4} }{ (\minrho)^{6}(1-\gamma)^{14}\epsilon^2 } \ln\paren{\frac{ \ell_F L_F  \gamma |\calS|\paren{|\calA|+|\calB|} }{(\minrho)(1-\gamma)\epsilon}}
\end{align}
and the batch-size, $M_y$, to be:
\begin{align}
    M_y = \Theta\paren{\frac{ \ell_F^{10} L_{F}^2 \gamma^8 |\calS|^{10} \paren{ |\calA| + |\calB|}^{22} }{(\minrho)^{16} (1-\gamma)^{52} \epsilon^8}}.
\end{align}
Finally, $$T_x = \frac{ L_F\ell_F^{\frac{9}{2}} \gamma^{8} |\calS|^{\frac{29}{4}} \paren{|\calA|+|\calB|}^{\frac{29}{2}} }{(1-\gamma)^{\frac{55}{2}} (\minrho)^{\frac{13}{2}} \epsilon^{\frac{7}{2}}  }. $$
%%%%%%%%%%%%%%%%%%%%%%%%%
\end{proof}

\subsubsection{CMG with Strongly Concave Utilities}
% \begin{align}
%     \E \tilde{\Phi}(x\at{T+1}) - \Tilde{\Phi}(x^\star)\leq \exp\paren{ - \frac{\mu_{pl}}{3\ellphi}T} L \diam{\calX} + \frac{3\ell_{\tilde{\Phi}}\delta_x^2}{\mu_{pl}} + \frac{3 \ell_{\tilde{}}\sigma_x^2 }{\mu_{pl}}.
% \end{align}

% \begin{align}
%     \E \tilde{\Phi}(x\at{T+1}) - \tilde{\Phi}^\star \leq \exp\paren{ - \frac{\mu_{\mathrm{pl},x} }{3\ellphi} T} L_{\Phi}\diam{\calX} + \frac{3\ellphi\delta_x^2}{\mu_{\mathrm{pl},x}} + \frac{3\ellphi \sigma_x^2}{\mu_{\mathrm{pl}, x}}.
% \end{align}
\begin{theorem}
 Assume a two-player zero-sum cMG with a utilities that are strongly concave with moduli $\mu_1, \mu_2 >0$. Let $\epsilon>0$ be given. Then, a tuning of \cref{alg:nepg} with,
 \begin{itemize}
     \item step-sizes, $\etax = \Theta\paren{\frac{(\minrho)^{3} (1-\gamma)^{\frac{15}{2}} \mu_1^{\frac{3}{2}} }{\ell_F^{\frac{5}{2}} \gamma^{\frac{5}{2}} |\calS|^{\frac{5}{2}} \paren{|\calA|+|\calB|}^{\frac{15}{4}} } }$ and $\etay = \Theta\paren{\frac{(1-\gamma)^3}{2\ell_F\gamma|\calS|^{\frac{1}{2}}\paren{|\calA|+|\calB|}^{\frac{3}{2}} }}$;
     \item exploration parameters $\varepsilon_x = \paren{\frac{\epsilon (\minrho)^{6} \mu^{3} \left(1 - \gamma\right)^{17}}{512 L_{F} |\calS|^{4} \ell_{F}^{5} \gamma^{5} \left(|\calA| + |\calB|\right)^{\frac{17}{2}}}
}$ and $\varepsilon_y =  \Theta\paren{\frac
{ \epsilon (\minrho)^{9} \mu_x^{\frac{9}{2}} \left(1 - \gamma\right)^{\frac{21}{2}}}
{|\calS|^{\frac{5}{4}} \ell_{F}^{\frac{3}{2}} \gamma^{\frac{1}{2}} \left(|\calA| + |\calB|\right)^{\frac{9}{4}}}
} $;
    \item batch-sizes $M_x =\Theta\left( \frac{ L_{F}^{4} |\calS|^{\frac{33}{4}} \diam{\calX}^{2} \ell_{F}^{\frac{27}{2}} \gamma^{\frac{27}{2}} \left(|\calA| + |\calB|\right)^{\frac{89}{4}}}{\epsilon^{2} (\minrho)^{19} \mu_x^{\frac{19}{2}} \left(1 - \gamma\right)^{\frac{97}{2}}} \right) $ and $M_y = \Theta\paren{\frac{4 L_{F}^{2} |\calS|^{\frac{7}{2}} \ell_{F}^{5} \gamma^{3} \left(|\calA| + |\calB|\right)^{\frac{15}{2}}}{\epsilon^{2} (\minrho)^{22} \mu_{x}^{9} \mu_{y}^{2} \left(1 - \gamma\right)^{37}}} $;
    \item an outer-loop number of iterations at least $T_x =\Theta\left( \frac{16 \sqrt{2} |\mathcal{S}|^{\frac{7}{4}} \ell_F^{\frac{7}{2}} \gamma^{\frac{7}{2}} (|\mathcal{A}|+|\mathcal{B}|)^{\frac{21}{4}}}{\minrho^7 \mu_x^{\frac{7}{2}} (1-\gamma)^{\frac{29}{2}}} \right)$ and inner loop iterations that are at least $T_y = \Theta\paren{\frac{2\ell_F^2 \gamma^2 |\calS| \paren{|\calA|+|\calB|}^{3}}{(\minrho)^4(1-\gamma)^{7}\mu_y^2 }};$
 \end{itemize}
 will output an iterate $x\at{t^\star},y\at{t^\star}$, such that:
 \begin{align}
     \E\left[ U(x\at{t^\star},y\at{t^\star+1} ) - \min_{x'\in\calX} U(x',y\at{t^\star+1})\right] &\leq \epsilon; \\
     \E\left[\max_{y\in\calY} U(x\at{t},y') - \E U(x\at{t^\star},y\at{t^\star}) \right] &\leq \epsilon. 
 \end{align}
\end{theorem}
\begin{proof}
An optimality gap $\E \tilde{\Phi}(x\at{T+1}) - \tilde{\Phi}^\star\leq \epsilon$ for the utility function that is composed with the greedy exploration mapping, translates to an optimality error of the original utility function:
\begin{align}
    \E {\Phi}(x\at{T+1}) - {\Phi}^\star\leq \epsilon + \frac{|\calS|^{\frac{1}{2}}L_\Phi \varepsilon_x }{|\calA|} + \frac{|\calS|^{\frac{1}{2}}L_\Phi \varepsilon_y }{|\calB|} + 8|\calS|^{\frac{1}{2}}\ellphi^2 \diam{\calX} \varepsilon_x.
\end{align}
First, we can set the number of iterations to be at least,
\begin{align}
    T_x = \Theta\paren{\frac{\ell_F\ellphi \gamma |\calS|^{\frac{1}{2}} \paren{|\calA|+|\calB|}^{\frac{3}{2}} }{(\minrho)^4(1-\gamma)^7\mu_x^2 }} = \Theta\left( \frac{16 \sqrt{2} |\mathcal{S}|^{\frac{7}{4}} \ell_F^{\frac{7}{2}} \gamma^{\frac{7}{2}} (|\mathcal{A}|+|\mathcal{B}|)^{\frac{21}{4}}}{\minrho^7 \mu_x^{\frac{7}{2}} (1-\gamma)^{\frac{29}{2}}} \right).
\end{align}
As such, we tune the exploration parameter $\varepsilon_x$ to be:
\begin{align}
    \varepsilon_x  = \min\left\{ 
    \frac{|\calA|\epsilon}{|\calS|^{\frac{1}{2}} L_U},
     \frac{\epsilon}{8|\calS|^{\frac{1}{2}} \ellphi^2 \diam{\calX}}
    \right\}
\end{align}
So we can set, $\varepsilon_x = \frac{\epsilon}{|\calS|^{\frac{1}{2}} L_U\ellphi^2 \diam{\calX} }$. Substituting yields: 
\begin{align}
    \varepsilon_x &= \Theta\paren{\frac{\epsilon (\minrho)^{6} \mu^{3} \left(1 - \gamma\right)^{17}}{ L_{F} |\calS|^{\frac{9}{2}} \ell_{F}^{5} \gamma^{5} \left(\calA + |\calB|\right)^{\frac{16}{2}}}} =  \Theta \paren{\frac{\epsilon (\minrho)^{6} \mu^{3} \left(1 - \gamma\right)^{17}}{512 L_{F} |\calS|^{4} \ell_{F}^{5} \gamma^{5} \left(|\calA| + |\calB|\right)^{\frac{17}{2}}}
}
\end{align}

\begin{align}
    M_x = 
    % \Theta\paren{\frac{\ellphi\sigma_x^2}{\muplx}} = \Theta\paren{\frac{2\ell_F \gamma \paren{|\calA|+|\calB|}^{\frac{3}{2}}\ellphi L_F^2}{ (\minrho)^4 (1-\gamma)^7 \mu_x^2 \varepsilon_x^2}}=
    % \\
    \Theta\left( \frac{ L_{F}^{4} |\calS|^{\frac{33}{4}} \diam{\calX}^{2} \ell_{F}^{\frac{27}{2}} \gamma^{\frac{27}{2}} \left(|\calA| + |\calB|\right)^{\frac{89}{4}}}{\epsilon^{2} (\minrho)^{19} \mu_x^{\frac{19}{2}} \left(1 - \gamma\right)^{\frac{97}{2}}} \right)
    % =\Theta \paren{\Theta\left( \frac{  L_{F}^{4} calS^{\frac{}{}} diam_{calX}^{2} \ell_{F}^{13.5} \gamma^{13.5} \left(calA + calB\right)^{22.25}}{\epsilon^{2} minrho^{19} \mu^{7.5} \mu_{x}^{2} \left(1 - \gamma\right)^{48.5}} \right)
\end{align}
% \begin{align}
%     \frac{\epsilon (\minrho)^{6} \mu^{3} \left(1 - \gamma\right)^{17}}{ L_{F} |\calS|^{\frac{9}{2}} \ell_{F}^{5} \gamma^{5} \left(\calA + |\calB|\right)^{\frac{16}{2}}}
% \end{align}

Further, we require that $\varepsilon_y \leq \min\left\{ \frac{|\calB| \epsilon}{|\calS|^{\frac{1}{2}} L_\Phi  }, \frac{\epsilon}{8|\calS|^{\frac{1}{2}} \ell_U^2 \diam{\calY} }\right\}$.
In order to tune $\delta_x$, we see that it is a sum of two terms, the bias of the gradient estimator and the expected distance of $y\at{t}$ from $y^\star(x\at{t})$ at each iterate $t$. We see that we need to set,
\begin{align}
    \delta_x \leq c \frac{ \sqrt{\muqgx\epsilon} }{\sqrt{\ellphi}}
\end{align}
for some constant $c>0$ sufficiently small. As such, we will control the horizon of the stochastic gradient estimator, $H_x$, to be,
\begin{align}
    H_x = \Theta\paren{ \frac{1}{1-\gamma} \ln \paren{\frac{\ell_F L_F \gamma |\calS| \paren{|\calA|+ |\calB| } }{\epsilon ( 1-\gamma) \minrho \mu_x  } } }
\end{align}

For the inner loop, we need to ensure that $\E\norm{y\at{t}- y^\star(x\at{t}) }\leq \frac{\sqrt{\muplx\epsilon}}{\sqrt{\ellphi}}$. Hence, by the quadratic growth condition w.r.t. $y$, the optimality gap of the inner loop, $\epsilon_y$, needs to be bounded as:
\begin{align}
    \epsilon_y = \Theta\paren{\frac{\epsilon\muplx\muqgy}{\ellphi \ell_U^2 }} = 
    \Theta\paren{
    \frac{\epsilon (\minrho)^6(1-\gamma)^{14} \mu_x^2\mu_y}{\ell_F\gamma^2 |\calS|^{\frac{3}{2} }\paren{|\calA|+|\calB|}^4  }
    }.
\end{align}
To achieve such an optimality gap at the inner-loop on every outer-loop iteration $t$, by \cref{varepsilon-bound-on-composed-function}, we need to set $\varepsilon_y = \frac{\epsilon \muplx\muqgy}{|\calS|^{\frac{1}{2}}\ellphi \ell_U^2} $.
$$\varepsilon_y = \Theta\paren{\frac
{ \epsilon (\minrho)^{9} \mu_x^{\frac{9}{2}} \left(1 - \gamma\right)^{\frac{21}{2}}}
{|\calS|^{\frac{5}{4}} \ell_{F}^{\frac{3}{2}} \gamma^{\frac{1}{2}} \left(|\calA| + |\calB|\right)^{\frac{9}{4}}}
},
$$
the latter leads to the bound on $\sigma_y^2$,
$$\sigma_y^2 = \Theta\paren{\frac{L_{F}^{2} |\calS|^{\frac{5}{2}} \ell_{F}^{3} \gamma \left(|\calA| + |\calB|\right)^{\frac{9}{2}}}{\epsilon^{2} (\minrho)^{18} \mu_{x}^{9} \left(1 - \gamma\right)^{27}}}$$
while, the previous bound calls for a batch-size,
$$M_y = \Theta\paren{\frac{4 L_{F}^{2} |\calS|^{\frac{7}{2}} \ell_{F}^{5} \gamma^{3} \left(|\calA| + |\calB|\right)^{\frac{15}{2}}}{\epsilon^{2} (\minrho)^{22} \mu_{x}^{9} \mu_{y}^{2} \left(1 - \gamma\right)^{37}}}.
$$
Finally, the horizon of the inner-loop gradient estimator will be set to be:
\begin{align}
    H_x = \Theta\paren{ \frac{1}{1-\gamma} \ln \paren{\frac{\ell_F L_F \gamma |\calS| \paren{|\calA|+ |\calB| } }{\epsilon ( 1-\gamma) \minrho \mupl \mu_{y}  } } },
\end{align}
while the step-size needs to be:
\begin{align}
    \tau_y =\Theta\paren{ \frac{1}{\ell_U}}= \Theta\paren{\frac{(1-\gamma)^3}{2\ell_F\gamma|\calS|^{\frac{1}{2}}\paren{|\calA|+|\calB|}^{\frac{3}{2}} }},
\end{align}
and the total number of iterations needs to be at least,
\begin{align}
    T_y &\geq \Theta\paren{ \frac{\ell_U}{\muply} \ln\paren{\frac{\gamma L_F\ell_F |\calS|\paren{|\calA|+|\calB|} }{\minrho(1-\gamma)\epsilon}} }\\
    &= \Theta\paren{\frac{2\ell_F^2 \gamma^2 |\calS| \paren{|\calA|+|\calB|}^{3}}{(\minrho)^4(1-\gamma)^{7}\mu_y^2 }}
\end{align}
\end{proof}

\subsection{Alternating Policy Gradient Descent-Ascent}
\subsubsection{CMG with concave utilities}
\begin{theorem}
\label{theorem:formal-apgda-nc-ppl}
    Let $\epsilon>0$ be a desired accuracy. After at most $T= \Theta \paren{\frac{\ell_F^4 L_F^4 \gamma^8 |\calS|^{6} \paren{|\calA|+|\calB|}^{19}}{ (1-\gamma)^{49}(\minrho)^{11} \epsilon^6}}$ iterations of \cref{alg:apgda}, \begin{itemize}
        \item 
     with step-sizes, $\tau_x = \Theta\paren{ \frac{ (\minrho)^9 (1-\gamma)^{20} \epsilon^3 }{\ell_F^2 \gamma^5 |\calS|^{3} \paren{|\calA|+ |\calB|}^{\frac{17}{2}} } }$, and $\etay = \Theta\paren{\frac{(1 - \gamma)^8}{4 \ell_F \gamma^2 |\mathcal{S}|^{\frac{3}{2}} (|\mathcal{A}|+ |\mathcal{B}|)^4}}$; \item exploration parameters $\varepsilon_x= \Theta \paren{ \frac{ \, \minrho \, (1 - \gamma)^{12}\epsilon}{2 L_F \ell_F  \, |\mathcal{S}|^{\frac{7}{2}} \gamma (|\mathcal{A}| + |\mathcal{B}|)^{\frac{11}{2}}}
} $ and  $\varepsilon_y = \Theta\paren{\frac{\minrho(1 - \gamma)^{23}\epsilon}{16 L_F \, |\mathcal{S}|^{\frac{11}{2}} \ell_F^2 \gamma^4 (|\mathcal{A}| + |\mathcal{B}|)^{11}} }$; 
\item  batch-sizes, $M_x = \Theta\paren{ \frac{L_F^4\ell_F^{2} |\calS|^{7}\paren{|\calA|+|\calB|}^{11} }{  (\minrho)^{4} (1-\gamma)^{32}\epsilon^4 } } $, and $M_y = \Theta
    \paren{
          \frac{\ell_F^7\gamma^{14}|\calS|^{\frac{31}{2}}(|\calA|+|\calB|)^{34} }{ (\minrho)^{10} (1-\gamma)^{85} \epsilon^5 }
    } $;
    \item sampling horizons $H_x,H_y =\Theta\paren{\frac{1}{(1-\gamma)}\ln\paren{\frac{\gamma\ell_F L_F |\calS| \paren{|\calA|+|\calB|} }{(1-\gamma)(\minrho) \epsilon }}} ,$
    \end{itemize}
    there exists an iterate $t^\star$, such that 
    \begin{align}
        \E U(x\at{t^\star},y\at{t^\star}) - \Phi^\star &\leq \epsilon; \\
        \E\Phi(x\at{t^\star}) - \E U(x\at{t^\star},y\at{t^\star})  &\leq \epsilon.
    \end{align}
\end{theorem}
\begin{proof}
We invoke \cref{theorem:formal-agda-ncppl} and tune the parameters appropriately using \cref{lemma:bias:policy-gradient-estimator,varepsilon-bound-on-composed-function,minmax-bound-on-composed-function}.
We tune the coefficient of the regularizer first. Needing to achieve an $(1-\gamma)(\minrho)\epsilon$-first order stationary point, we tune the regularizer as:
$$\mur = \Theta\paren{\frac{(1-\gamma)(\minrho)\epsilon}{L_F}}.$$
We compute the pPL modulus given the regularizer tuning,
\begin{align}
    \mupl =\Theta\paren{ \frac{ (\minrho)^6(1-\gamma)^{14} \epsilon^2 }{  \ell_F L_F^2 \gamma^2 |\calS|^{\frac{1}{2}} \paren{|\calA|+|\calB|}^{4}  } }.
\end{align}
Further, by  \cref{claim:epsilon-greedy-gradient-dominance,claim:epsilon-greedy-D,minmax-bound-on-composed-function}, we see that $\varepsilon_x,\varepsilon_y$ need to be tuned as, $$\varepsilon_x = \Theta{\paren{
\frac{(1-\gamma)(\minrho) \epsilon}{|\calS| \ell_F  L_F \elllambda  \Llambda^4 }
}}  = \Theta \paren{ \frac{\epsilon \, \minrho \, (1 - \gamma)^{12}}{2 L_F \ell_F  \, |\mathcal{S}|^{\frac{7}{2}} \gamma (|\mathcal{A}| + |\mathcal{B}|)^{\frac{11}{2}}},
}$$
and
\begin{align}
    \varepsilon_y = \Theta\paren{\frac{\minrho(1 - \gamma)^{23}\epsilon}{16 L_F \, |\mathcal{S}|^{\frac{11}{2}} \ell_F^2 \gamma^4 (|\mathcal{A}| + |\mathcal{B}|)^{11}} }
\end{align}
The resulting bounds on the variance are:
\begin{align}
    \sigma_x^2 = \frac{L_F^4\ell_F^{2} |\calS|^{7}\paren{|\calA|+|\calB|}^{11} }{ \epsilon^2 (\minrho)^{2} (1-\gamma)^{30} },
\end{align}
and
\begin{align}
    \sigma_y^2 = \Theta\paren{ \frac{16 L_F^{4} \, |\mathcal{S}|^{11} \ell_F^4 \gamma^8 (|\mathcal{A}| + |\mathcal{B}|)^{22}}{\minrho(1 - \gamma)^{52}\epsilon^2}}.
\end{align}
The latter dictates the batch-sizes,
\begin{align}
    M_x = \Theta\paren{ \frac{L_F^4\ell_F^{2} |\calS|^{7}\paren{|\calA|+|\calB|}^{11} }{ \epsilon^4 (\minrho)^{4} (1-\gamma)^{32} } },
\end{align}
and 
\begin{align}
    M_y = \Theta\paren{\kappa^2 \sigma_y^2}  = 
    \Theta
    \paren{
          \frac{\ell_F^3\gamma^6|\calS|^{\frac{9}{2}}(|\calA|+|\calB|)^{12} \sigma_y^2 }{(\minrho)^{9} (1-\gamma)^{33} \epsilon^3 }
    }
    = 
    \Theta
    \paren{
          \frac{\ell_F^7\gamma^{14}|\calS|^{\frac{31}{2}}(|\calA|+|\calB|)^{34} }{ (\minrho)^{10} (1-\gamma)^{85} \epsilon^5 }
    }
\end{align}
The step-size is tuned straightforwardly as,
\begin{align}
    \tau_x = \Theta\paren{ \frac{ (\minrho)^9 (1-\gamma)^{20} \epsilon^3 }{\ell_F^2 \gamma^5 |\calS|^{3} \paren{|\calA|+ |\calB|}^{\frac{17}{2}} } }.
\end{align}
By \cref{lemma:bias:policy-gradient-estimator} we see easily that $H_x, H_y$ need to be, 
\begin{align}
    H_x, H_y = \Theta\paren{ \frac{1}{(1-\gamma)}\ln\paren{\frac{\gamma\ell_F L_F |\calS| \paren{|\calA|+|\calB|} }{(1-\gamma)(\minrho) \epsilon }}
    }
\end{align}
After we have tuned the regularizer, we can compute the upper bound on the number of iterations:
\begin{align}
    T= \Theta \paren{\frac{\ell_F^4 L_F^4 \gamma^8 |\calS|^{6} \paren{|\calA|+|\calB|}^{19}}{ (1-\gamma)^{49}(\minrho)^{11} \epsilon^6}}.
\end{align}
Now, tuning $\eta_y$ is also straightforward,
\begin{align}
    \etay =\Theta \paren{\frac{1}{\ell_U^{\mu}} } = \Theta \paren{\frac{1}{\ell_F\elllambda^2 L_\lambda }} = \Theta\paren{\frac{(1 - \gamma)^8}{4 \ell_F \gamma^2 |\mathcal{S}|^{\frac{3}{2}} (|\mathcal{A}|+ |\mathcal{B}|)^4}}.
\end{align}

\end{proof}

\subsubsection{cMG with strongly concave utilities}
\begin{theorem}
\label{theorem:formal-apgda-cmg-ppl-ppl}
 Assume $\epsilon>0$ and a two-player zero-sum cMG with utilities that are strongly concave with moduli $\mu_1, \mu_2$. Then if the two players are following \cref{alg:apgda} with
 \begin{itemize}
     \item 
exploration parameters $\varepsilon_x,\varepsilon_y = \Theta\paren{ \frac{\epsilon (\minrho)^{4} \left(1 - \gamma\right)^{15} \min\{\mu_x^2, \mu_y^2\}}{4 L_{F} |\calS|^{\frac{5}{2}} \ell_{F}^{2} \gamma^{2} \left(|\calA| + |\calB|\right)^{4} } }$,
\item step-sizes, $\etax = \Theta\paren{\frac{\mu_{2}^4(\minrho)^{8}(1-\gamma)^{30}}{32\ell_F^2\gamma^2|\calS|\paren{|\calA|+|\calB|}}}$, $\etay = \paren{\frac{(1-\gamma)^3}{2\ell_F\gamma|\calS|^{\frac{1}{2}} \paren{|\calA|+|\calB|}^{\frac{3}{2}} }}$ and \item  batch-sizes, 
$  
  M_x = \Theta\paren{
    \frac{16 L_{F}^4 |\calS|^{\frac{13}{2}} \ell_{F}^{5} \gamma^{6} \left(|\calA| + |\calB|\right)^{12} }{\epsilon^2 (\minrho)^{20} \left(1 - \gamma\right)^{42} \min\{\mu_x^6, \mu_y^6\}}
    }$, $
    M_y = \Theta\paren{
    \frac{16 L_{F}^4 |\calS|^{\frac{19}{2}} \ell_{F}^{5} \gamma^{10} \left(|\calA| + |\calB|\right)^{20} }{\epsilon^2 (\minrho)^{28} \left(1 - \gamma\right)^{50} \min\{\mu_x^{10}, \mu_y^{10}\}}
    }$,
 \end{itemize}
then, it is the case that:
    \begin{align}
       \E U(x\at{T},y\at{T}) - \Phi^\star &\leq \epsilon;\\
       \E \Phi(x\at{T}) - \E U(x\at{T},y\at{T}) \Phi^\star &\leq \epsilon.
    \end{align}
    after at most $T$ iterations, with $T = \Theta\paren{\frac{4\ell_F\gamma^6|\calS|^{\frac{9}{2}} \paren{|\calA|+|\calB|}^{12} }{\minrho^{12}(1-\gamma)^{36} \mu_{1}^2\mu_{2}^4 } \log\paren{\frac{L_F\ell_F \gamma |\calS|\paren{|\calA|+|\calB|} }{(\minrho)(1-\gamma) \mu_{1}, \mu_{2} }}} $.
\end{theorem}
 Assume the function $\tilde{U}(x,y):= U\paren{w(x;\varepsilon_x,|\calA|), w(y;\varepsilon_y,|\calB|)}$ where $w(z;\varepsilon,m) := (1-\varepsilon)z + \frac{\varepsilon\bm{1}}{m}$.  \cref{claim:epsilon-greedy-D,minmax-bound-on-composed-function} make sure that we can bound the optimality gap on the initial function $U$ by running \cref{alg:apgda} on $\tilde{U}$. Hence, combining the aforementioned claims with \cref{theorem:formal-agda-2sided-ppl} and \cref{lemma:variance:policy-gradient-estimator}, we see that we need to set: 
\begin{align}
    \varepsilon_x,\varepsilon_y = \Theta\paren{ \frac{\epsilon (\minrho)^{4} \left(1 - \gamma\right)^{15} \min\{\mu_x^2, \mu_y^2\}}{4 L_{F} |\calS|^{2} \ell_{F}^{2} \gamma^{2} \left(|\calA| + |\calB|\right)^{4} \left(\diam{\calX} + \diam{\calY}\right)} }
\end{align}
The resulting variances will be:
\begin{align}
    \sigma_x^2,\sigma_y^2 = \Theta\paren{ \frac{16 L_{F}^4 |\calS|^{4} \ell_{F}^{4} \gamma^{4} \left(|\calA| + |\calB|\right)^{8} \left(\diam{\calX} + \diam{\calY}\right)^2}{\epsilon^2 (\minrho)^{8} \left(1 - \gamma\right)^{30} \min\{\mu_x^4, \mu_y^4\}}
    }.
\end{align}
We will control the resulting variances using batches of the following size, which will also counter the $\mu_1$
\begin{align}
    M_x &= \Theta\paren{
    \frac{16 L_{F}^4 |\calS|^{\frac{13}{2}} \ell_{F}^{5} \gamma^{6} \left(|\calA| + |\calB|\right)^{12} }{\epsilon^2 (\minrho)^{20} \left(1 - \gamma\right)^{42} \min\{\mu_x^6, \mu_y^6\}}
    };\\
    M_y &= \Theta\paren{
    \frac{16 L_{F}^4 |\calS|^{\frac{19}{2}} \ell_{F}^{5} \gamma^{10} \left(|\calA| + |\calB|\right)^{20} }{\epsilon^2 (\minrho)^{28} \left(1 - \gamma\right)^{50} \min\{\mu_x^{10}, \mu_y^{10}\}}
    }.
\end{align}
We can easily see that the sampling horizons need to be:
\begin{align}
    H_x, H_y = \Theta\paren{ \frac{1}{1-\gamma} \ln \paren{\frac{\ell_F L_F \gamma |\calS| \paren{|\calA|+ |\calB| } }{\epsilon ( 1-\gamma) \minrho \mu_x\mu_y  } } }
\end{align}
The step-sizes are tuned to be:
\begin{gather}
    \etax = \frac{\mu_{2,H}^4(\minrho)^{8}(1-\gamma)^{30}}{32\ell_F^2\gamma^2|\calS|\paren{|\calA|+|\calB|}};\quad \etay = \frac{(1-\gamma)^3}{2\ell_F\gamma|\calS|^{\frac{1}{2}} \paren{|\calA|+|\calB|}^{\frac{3}{2}} }
\end{gather}
Finally, the iteration complexity is at least:
\begin{align}
    T = \Theta\paren{\frac{4\ell_F\gamma^6|\calS|^{\frac{9}{2}} \paren{|\calA|+|\calB|}^{12} }{\minrho^{12}(1-\gamma)^{36} \mu_{H,1}^2\mu_{H,2}^4 } \log\paren{\frac{L_F\ell_F \gamma |\calS|\paren{|\calA|+|\calB|} }{(\minrho)(1-\gamma) \mu_{1,H}, \mu_{2,H} }}}.
\end{align}

\end{document}